\documentclass[11pt,a4paper]{article}
\pdfoutput=1

\usepackage[margin=1in]{geometry}
\usepackage{amsmath,amssymb,amsthm,booktabs,color,doi,enumitem,graphicx,latexsym,url,xcolor}
\usepackage[numbers,sort&compress]{natbib}
\usepackage{microtype,hyperref}
\usepackage{thmtools,thm-restate}
\usepackage{algorithm}
\usepackage{footmisc}
\usepackage{libertine}

\newtheorem{theorem}{Theorem}
\newtheorem{lemma}{Lemma}
\newtheorem{corollary}{Corollary}
\newtheorem{definition}{Definition}
\theoremstyle{remark}
\newtheorem{remark}{Remark}

\newcommand{\namedref}[2]{\hyperref[#2]{#1~\ref*{#2}}}
\newcommand{\sectionref}[1]{\namedref{Section}{#1}}

\newcommand{\theoremref}[1]{\namedref{Theorem}{#1}}
\newcommand{\corollaryref}[1]{\namedref{Corollary}{#1}}
\newcommand{\remarkref}[1]{\namedref{Remark}{#1}}
\newcommand{\figureref}[1]{\namedref{Figure}{#1}}
\newcommand{\lemmaref}[1]{\namedref{Lemma}{#1}}
\newcommand{\tableref}[1]{\namedref{Table}{#1}}
\newcommand{\constrref}[1]{\hyperref[#1]{Constraint~(\ref*{#1})}} 

\newcommand{\definitionref}[1]{\namedref{Definition}{#1}}

\renewcommand{\vec}[1]{\mathbf{#1}}

\newcommand{\nnR}{\mathbb{R}^{+}}

\newcommand{\sstart}{\textsc{start}}
\newcommand{\sreset}{\textsc{reset}}
\newcommand{\sready}{\textsc{ready}}

\newcommand{\spropose}{\textsc{propose}}
\newcommand{\spulse}{\textsc{pulse}}
\newcommand{\srecover}{\textsc{recover}}

\newcommand{\sresync}{\textsc{act}} 
\newcommand{\sidle}{\textsc{idle}}
\newcommand{\svote}{\textsc{vote}}
\newcommand{\spass}{\textsc{pass}} 
\newcommand{\sgo}{\textsc{go}} 
\newcommand{\sfail}{\textsc{fail}}
\newcommand{\signore}{\textsc{ignore}}
\newcommand{\shold}{\textsc{hold}}
\newcommand{\swait}{\textsc{wait}}
\newcommand{\swakeup}{\textsc{listen}}

\newcommand{\sone}{\textsc{input 1}}
\newcommand{\szero}{\textsc{input 0}}
\newcommand{\sexeczero}{\textsc{run 0}}
\newcommand{\sexecone}{\textsc{run 1}}
\newcommand{\slisten}{\textsc{listen}}
\newcommand{\sread}{\textsc{read}} 
\newcommand{\sinputzero}{\textsc{input 0}}
\newcommand{\sinputone}{\textsc{input 1}}
\newcommand{\soutputzero}{\textsc{output 0}}
\newcommand{\soutputone}{\textsc{output 1}}

\newcommand{\tconsensus}{T_{\text{consensus}}}
\newcommand{\tactive}{T_{\text{active}}}
\newcommand{\tmin}[1]{T_{\text{min},#1}}
\newcommand{\tmax}[1]{T_{\text{max}, #1}}
\newcommand{\tcool}{T_{\text{cool}}}
\newcommand{\tlisten}{T_{\text{listen}}}
\newcommand{\tidle}{T_\text{idle}}
\newcommand{\tvote}{T_\text{vote}}
\newcommand{\tlarge}{T_\text{att}}
\newcommand{\twait}{T_\text{wait}}
\newcommand{\tminh}{\tmin{h}}
 
\newcommand{\tmaxh}{\tmax{h}}

\newcommand{\guard}[1]{\text{Guard G}#1}

\newcounter{constrnumbers}
\newcommand\constrnumber{(\refstepcounter{constrnumbers}\arabic{constrnumbers})}

\DeclareMathOperator{\poly}{poly}
\DeclareMathOperator{\polylog}{polylog}

\DeclareMathOperator{\N}{\mathbb N}

\newenvironment{mycover}
               {\list{}{\listparindent 0pt
                        \itemindent    \listparindent
                        \leftmargin    0pt
                        \rightmargin   0pt
                        \parsep        0pt}%
                \raggedright
                \item\relax}
               {\endlist}

\begin{document}

\pagenumbering{Alph}

\hypersetup{
    pdfauthor={Christoph Lenzen and Joel Rybicki},
    pdftitle={Self-stabilising Byzantine Clock Synchronisation is Almost as Easy as Consensus},
}

\begin{mycover}
{\LARGE \textbf{Self-stabilising Byzantine Clock Synchronisation is Almost as Easy as Consensus}\par}

\bigskip

\bigskip
\textbf{Christoph Lenzen}\, $\cdot$\, \href{mailto:clenzen@mpi-inf.mpg.de}{\small{clenzen@mpi-inf.mpg.de}}

\smallskip
{\small Department of Algorithms and Complexity, \\
Max Planck Institute for Informatics, \\
Saarland Informatics Campus \par}

\bigskip
\textbf{Joel Rybicki}\, $\cdot$\, \href{mailto:joel.rybicki@ist.ac.at}{\small{joel.rybicki@ist.ac.at}}

\smallskip
{\small Institute of Science and Technology Austria (IST Austria) \par}

\end{mycover}

\paragraph{Abstract.}
We give fault-tolerant algorithms for establishing synchrony in distributed systems in which each of the $n$ nodes has its own clock. Our algorithms operate in a very strong fault model: we require self-stabilisation, i.e., the initial state of the system may be arbitrary, and there can be up to $f<n/3$ ongoing Byzantine faults, i.e., nodes that deviate from the protocol in an arbitrary manner. Furthermore, we assume that the local clocks of the nodes may progress at different speeds (clock drift) and communication has bounded delay. In this model, we study the pulse synchronisation problem, where the task is to guarantee that eventually all correct nodes generate well-separated local pulse events (i.e., unlabelled logical clock ticks) in a synchronised manner. 

Compared to prior work, we achieve \emph{exponential} improvements in stabilisation time and the number of communicated bits, and give the first sublinear-time algorithm for the problem:
\begin{itemize}
 \item In the deterministic setting, the state-of-the-art solutions stabilise in time $\Theta(f)$ and have each node broadcast $\Theta(f \log f)$ bits per time unit. We exponentially reduce the number of bits broadcasted per time unit to $\Theta(\log f)$ while retaining the same stabilisation time.
 \item In the randomised setting, the state-of-the-art solutions stabilise in time $\Theta(f)$ and have each node broadcast $O(1)$ bits per time unit. We exponentially reduce the stabilisation time to $\polylog f$ while each node broadcasts $\polylog f$ bits per time unit.
\end{itemize}
These results are obtained by means of a recursive approach reducing the above task of \emph{self-stabilising} pulse synchronisation in the \emph{bounded-delay} model to \emph{non-self-stabilising} binary consensus in the \emph{synchronous} model. In general, our approach introduces at most logarithmic overheads in terms of stabilisation time and broadcasted bits over the underlying consensus routine.

\thispagestyle{empty}
\newpage

\setcounter{tocdepth}{2}
\tableofcontents
\thispagestyle{empty}

\setcounter{page}{0}
\newpage

\pagenumbering{arabic}

\section{Introduction}\label{sec:intro}

Many of the most fundamental problems in distributed computing relate to timing and fault tolerance. Even though most distributed systems are inherently asynchronous, it is often convenient to design such systems by assuming some degree of synchrony provided by reliable global or distributed clocks. For example, the vast majority of existing Very Large Scale Integrated (VLSI) circuits operate according to the synchronous paradigm: an internal clock signal is distributed throughout the chip neatly controlling alternation between computation and communication steps. Of course, establishing the synchronous abstraction is of high interest in numerous other large-scale distributed systems, as it makes the design of algorithms considerably easier. 

However, as the accuracy and availability of the clock signal is typically one of the most basic assumptions, clocking errors affect system behavior in unpredictable ways that are often hard -- if not impossible -- to tackle at higher system layers. Therefore, \emph{reliably} generating and distributing a joint clock is an essential task in distributed systems. Unfortunately, the cost of providing fault-tolerant synchronisation and clocking is still poorly understood.

\subsection{Pulse synchronisation}

In this work, we study the \emph{self-stabilising Byzantine pulse synchronisation} problem~\cite{dolev04clock-synchronization,dolev14fatal}, which requires the system to achieve synchronisation despite severe faults. We assume a fully connected message-passing system of $n$ nodes, where 
\begin{enumerate}%[nosep]
 \item an unbounded number of transient faults may occur anywhere in the network, and 
 \item up to $f<n/3$ of the nodes can be faulty and exhibit \emph{arbitrary} ongoing misbehaviour.
\end{enumerate}
In particular, the transient faults may arbitrarily corrupt the state of the nodes and result in loss of synchrony. Moreover, the nodes that remain faulty may deviate from any given protocol, behave adversarially, and collude to disrupt the other nodes by sending them \emph{different} misinformation even after transient faults have ceased. Note that this also covers faults of the communication network, as we may map faults of communication links to one of their respective endpoints.  The goal is now to (re-)establish synchronisation once transient faults cease, despite up to $f < n/3$ Byzantine nodes. That is, we need to consider algorithms that are simultaneously (1) self-stabilising~\cite{dijkstra74control,dolev00self-stabilization} and (2) Byzantine fault-tolerant~\cite{lamport82byzantine}.

More specifically, the problem is as follows: after transient faults cease, no matter what is the initial state of the system, the choice of up to $f<n/3$ faulty nodes, and the behaviour of the faulty nodes, we require that after a bounded \emph{stabilisation time} all the \emph{non-faulty} nodes must generate pulses that 
\begin{itemize}%[nosep]
 \item occur almost simultaneously at each correctly operating node (i.e., have small \emph{skew}), and 
 \item satisfy specified minimum and maximum frequency bounds (\emph{accuracy}).
\end{itemize}
While the system may have arbitrary behaviour during the initial stabilisation phase due to the effects of transient faults, eventually the above conditions provide synchronised unlabelled clock ticks for all non-faulty nodes:
\begin{center}
 \includegraphics[page=4]{figures.pdf}
\end{center}

In order to meet these requirements, it is necessary that nodes can estimate the progress of time. To this end, we assume that nodes are equipped with (continuous, real-valued) hardware clocks that run at speeds that may vary arbitrarily within 1 and $\vartheta$, where $\vartheta \in O(1)$. That is, we normalize minimum clock speed to $1$ and assume that the clocks have drift bounded by a constant.
Observe that in an asynchronous system, i.e., one in which communication and/or computation may take unknown and unbounded time, even perfect clocks are insufficient to ensure any relative timing guarantees between the actions of different nodes. Therefore, we additionally assume that the nodes can send messages to each other that are received and processed within at most $d \in \Theta(1)$ time.
The clock speeds and message delays can behave adversarially within the respective bounds given by $\vartheta$ and $d$.

In summary, this yields a highly adversarial model of computing, where further restrictions would render the task infeasible: 
\begin{enumerate}
 \item transient faults are arbitrary and may involve the entire network,
 \item ongoing faults are arbitrary, cover erroneous behavior of both the nodes and the communication links, and the problem is not solvable if $f\ge n/3$~\cite{dolev86impossibility}, and 
 \item without any bounds on the accuracy of local clocks and on the communication delay, good synchronisation cannot be achieved: Even without clock drift, unbounded message delays lead to unbounded skew~\cite{lundelius84}, and if clocks have unbounded drift, trivial indistinguishability arguments show that no bounds on pulse frequency can be maintained.
\end{enumerate}

\subsection{Background and related work}

If one takes any one of the elements described above out of the picture, then this greatly simplifies the problem. Without permanent/ongoing faults, the problem becomes trivial: it suffices to have all nodes follow a designated leader. Without transient faults~\cite{lamport85clocks}, straightforward solutions are given by elegant classics~\cite{srikanth87clock,welch88fault}, where~\cite{welch88fault} also guarantees asymptotically optimal skew~\cite{lundelius84}. Taking the uncertainty of unknown message delays and drifting clocks out of the equation leads to the so-called digital clock synchronisation problem~\cite{ben-or08fast,dolev16counting,lenzen17efficient,lenzen16firing}, where communication proceeds in synchronous rounds and the task is to agree on a consistent (bounded) round counter. While this abstraction is unrealistic as a basic system model, it yields conceptual insights into the pulse synchronisation problem in the bounded-delay model. Moreover, it is useful to assign numbers to pulses after pulse synchronisation is solved, in order to get a fully-fledged shared system-wide clock~\cite{fuegger13efficient}.

In contrast to these relaxed problem formulations, the pulse synchronisation problem was initially considered to be very challenging -- if not impossible -- to solve. In a seminal article, Dolev and Welch~\cite{dolev04clock-synchronization} proved otherwise, albeit with an algorithm having an impractical exponential stabilisation time. In a subsequent line of work, the stabilisation time was reduced to polynomial~\cite{daliot03self-stabilizing} and then linear in $f$~\cite{dolev07bounded}.
However, the linear-time algorithm relies on simulating multiple instances of synchronous \emph{consensus} algorithms~\cite{pease80reaching} concurrently, which results in a high communication complexity.

The consensus problem~\cite{pease80reaching,lamport82byzantine} is one of the fundamental primitives in fault-tolerant computing. Most relevant to this work is synchronous binary consensus with (up to $f$) Byzantine faults. Here, node $v$ is given an input $x(v) \in \{0,1\}$, and it must output $y(v) \in \{0,1\}$ such that the following properties hold:
\begin{enumerate}
  \item \textbf{Agreement:} There exists $y \in \{0,1\}$ such that $y(v)=y$ for all correct nodes $v$.
  \item \textbf{Validity:} If for $x \in \{0,1\}$ it holds that $x(v)=x$ for all correct nodes $v$, then $y=x$.
  \item \textbf{Termination:} All correct nodes eventually decide on $y(v)$ and terminate.
\end{enumerate}
In this setting, two of the above main obstacles are not present: the system is properly initialised (no self-stabilisation required) and computation proceeds in synchronous rounds, i.e., well-ordered compute-send-receive cycles. This confines the task to understanding how to deal with the interference from Byzantine nodes. Synchronous consensus is extremely well-studied; see e.g.\ \cite{raynal10survey} for a survey. It is known that precisely $\lfloor (n-1)/3\rfloor$ faults can be tolerated in a system of $n$ nodes~\cite{pease80reaching}, $\Omega(nf)$ messages need to be sent in total~\cite{dolev85bounds}, the connectivity of the communication network must be at least $2f+1$~\cite{dolev82byzantine}, deterministic algorithms require $f+1$ rounds~\cite{fischer82lower,aguilera99simple}, and randomised algorithms can solve the problem in constant expected time~\cite{feldman95optimal}. In constrast, no non-trivial lower bounds on the time or communication complexity of pulse synchronisation are known.

The linear-time pulse synchronisation algorithm in~\cite{dolev07bounded} relies on simulating (up to) one synchronous consensus instance for each node simultaneously. Accordingly, this protocol requires each node to broadcast $\Theta(f\log f)$ bits per time unit. Moreover, the use of \emph{deterministic} consensus is crucial, as failure of any consensus instance to generate correct output within a prespecified time bound may result in loss of synchrony, i.e., the algorithm would fail \emph{after} apparent stabilisation. In~\cite{dolev14fatal}, these obstacles were overcome by avoiding the use of consensus by reducing the pulse synchronisation problem to the easier task of generating at least one well-separated ``resynchronisation point'', which is roughly uniformly distributed within any period of $\Theta(f)$ time. This can be achieved by trying to initiate such a resynchronisation point at random times, in combination with threshold voting and locally checked timing constraints to rein in the influence of Byzantine nodes. In a way, this seems much simpler than solving consensus, but the randomisation used to obtain a suitable resynchronisation point strongly reminds of the power provided by shared coins~\cite{rabin83randomized,ben-or83agreement,feldman95optimal,ben-or08fast} -- and this is exactly what the core routine of the expected constant-round consensus algorithm from~\cite{feldman95optimal} provides.

\subsection{Contributions}

\begin{table}[t!]
\centering
\caption{Summary of pulse synchronisation algorithms for $f \in \Theta(n)$. For each respective algorithm, the first two columns give the stabilisation time and the number of bits broadcasted by a node per time unit. The third column denotes whether algorithm is deterministic or randomised. The randomised algorithms stabilise in the given time with high probability. The fourth column indicates additional details or model assumptions. All algorithms tolerate $f < n/3$ faulty nodes except for (*), where it is required that $f < n/(3+\varepsilon)$ for an arbitrary, but fixed constant $\varepsilon>0$.}\label{table:algorithms}
\begin{tabular}{@{}l@{\quad\ \ }l@{\quad\ \ }l@{\quad\ \ }l@{\quad\ \ }l@{}}
  \toprule
  time & bits & type & notes & reference \\
  \midrule
  $\poly f$ & $O(\log f)$ & det. & & \cite{daliot03self-stabilizing} \\
  $O(f)$ & $O(f \log f )$  & det. & & \cite{dolev07bounded} \\
  $O(f)$ & $O(\log f)$ & det. & & {this work} and \cite{berman89consensus} \\
  \midrule
  $2^{O(f)}$ & $O(1)$ & rand. & adversary cannot predict coin flips & \cite{dolev04clock-synchronization} \\
  $O(f)$ & $O(1)$ & rand. & adversary cannot predict coin flips & \cite{dolev14fatal} \\
  $\polylog f $ & $\polylog f$ & rand. & private channels, (*) & {this work} and \cite{king11breaking} \\  
  $O(\log f)$ & $\poly f$ & rand. & private channels & {this work} and \cite{feldman95optimal}  \\
  \bottomrule
\end{tabular}
\end{table}

Our main result is a framework that reduces pulse synchronisation to an arbitrary synchronous binary consensus routine at very small overheads. In other words, given \emph{any} efficient algorithm that solves consensus in the standard synchronous model of computing \emph{without} self-stabilisation, we show how to obtain an efficient algorithm that solves the \emph{self-stabilising} pulse synchronisation problem in the bounded-delay model with clock drift.

While we build upon existing techniques, our approach has many key differences. First of all, while Dolev et al.~\cite{dolev14fatal} also utilise the concept of resynchronisation pulses, these are generated probabilistically. Moreover, their approach has an inherent time bound of $\Omega(f)$ for generating such pulses. In contrast, we devise a new recursive scheme that allows us to (1) \emph{deterministically} generate resynchronisation pulses in $\Theta(f)$ time and (2) \emph{probabilistically} generate resynchronisation pulses in $o(f)$ time. To construct algorithms that generate resynchronisation pulses, we employ resilience boosting and filtering techniques inspired by our recent line of work on digital clock synchronisation in the \emph{synchronous} model~\cite{lenzen17efficient,lenzen16firing}. One of its main motivations was to gain a better understanding of the linear time/communication complexity barrier that research on pulse synchronisation ran into, without being distracted by the additional timing uncertainties due to communication delay and clock drift. The challenge here is to port these newly developed tools from the synchronous model to the bounded-delay bounded-drift model in a way that keeps them in working condition.

The key to efficiency is a recursive approach, where each node participates in only $\lceil\log f\rceil$ consensus instances, one for each level of recursion. On each level, the overhead of the reduction over a call to the consensus routine is a constant multiplicative factor both in time and bit complexity; concretely, this means that both complexities increase by overall factors of $O(\log f)$. Applying suitable consensus routines yields \emph{exponential improvements} in bit complexity of deterministic and time complexity of randomised solutions, respectively:
\begin{enumerate}%[noitemsep]
  \item In the deterministic setting, we exponentially reduce the number of bits each node broadcasts per time unit to $\Theta(\log f)$, while retaining $\Theta(f)$ stabilisation time. This is achieved by employing the phase king algorithm~\cite{berman89consensus} in our construction.
  
  \item In the randomised setting, we exponentially reduce the stabilisation time to $\polylog f$, where each node broadcasts $\polylog f$ bits per time unit. This is achieved using the algorithm by King and Saia~\cite{king11breaking}. We note that this slightly reduces resilience to $f<n/(3+\varepsilon)$ for any fixed constant $\varepsilon>0$ and requires private communication channels. 
  
  \item In the randomised setting, we can also obtain a stabilisation time of $O(\log f)$, polynomial communication complexity, and optimal resilience of $f < n/3$ by assuming private communication channels. This is achieved using the consensus routine of Feldman and Micali~\cite{feldman95optimal}. This almost settles the open question by Ben-Or et al.~\cite{ben-or08fast} whether pulse synchronisation can be solved in expected constant time.
\end{enumerate}
The running time bounds of the randomised algorithms (2) and (3) hold with high probability and the additional assumptions on resilience and private communication channels are inherited from the employed consensus routines. Here, private communication channels mean that Byzantine nodes must make their decision on which messages to send in round $r$ based on knowledge of the algorithm, inputs, and all messages faulty nodes receive up to and including round $r$. The probability distribution is then over the independent internal randomness of the correct nodes (which the adversary can only observe indirectly) and any possible randomness of the adversary. Our framework does not impose these additional assumptions: stabilisation is guaranteed for $f < n/3$ on each recursive level of our framework as soon as the underlying consensus routine succeeds (within prespecified time bounds) constantly many times in a row. Our results and prior work are summarised in \tableref{table:algorithms}.

Regardless of the employed consensus routine, we achieve a skew of $2d$, where $d$ is the maximum message delay. This is optimal in our model, but overly pessimistic if the sum of communication and computation delay is not between $0$ and $d$, but from $(d^-,d^+)$, where $d^+-d^-\ll d^+$. In terms of $d^+$ and $d^-$, a skew of $\Theta(d^+-d^-)$ is asymptotically optimal~\cite{lundelius84,welch88fault}. We remark that in~\cite{khanchandani16self-stabilizing}, it is shown how to combine the algorithms from~\cite{dolev14fatal} and~\cite{welch88fault} to achieve this bound without affecting the other properties shown in~\cite{dolev14fatal}; we are confident that the same technique can be applied to the algorithm proposed in this work. Finally, all our algorithms work with any clock drift parameter $1 <\vartheta \le 1.004$, that is, the nodes' clocks can have up to $0.4\%$ drift. In comparison, cheap quartz oscillators achieve $\vartheta \approx 1 + 10^{-5}$. 

\subsection{Hardness of pulse synchronisation}

We consider our results of interest beyond the immediate improvements in complexity of the best known algorithms for pulse synchronisation. Since our framework may employ any consensus algorithm, it proves that pulse synchronisation is, essentially, \emph{as easy} as synchronous consensus -- a problem without the requirement for self-stabilisation or any timing uncertainty. Apart from the possibility for future improvements in consensus algorithms carrying over, this accentuates the following fundamental open question:
\begin{quote}
Is pulse synchronisation \emph{at least as hard} as synchronous consensus?
\end{quote}
Due to the various lower bounds and impossibility results on consensus~\cite{pease80reaching,fischer82lower,dolev82byzantine,dolev85bounds} mentioned earlier, a positive answer would immediately imply that the presented techniques are near-optimal. However, one may speculate that pulse synchronisation may rather have the character of (synchronous) approximate agreement~\cite{dolev86reaching,fekete90approximate}, as \emph{precise} synchronisation of the pulse events at different nodes is not required. Considering that approximate agreement can be deterministically solved in $O(\log c)$ rounds, where $c$ is the range of the input values, a negative answer is a clear possibility as well. Given that all currently known solutions either explicitly solve consensus, leverage techniques that are likely to be strong enough to solve consensus, or are very slow, this would suggest that new algorithmic techniques and insights into the problem are necessary.

\section{Preliminaries}\label{sec:preliminaries}

In this section, we describe the model of computation, introduce notation used in the subsequent sections, and formally define the pulse synchronisation and resynchronisation problems.

\subsection{Notation}

We use $\N = \{1, 2, \ldots\}$ to denote positive integers and $\N_0 = \N \cup \{0\}$. For any $k \in \N$, we define the short-hand $[k] = \{0, 1, \ldots, k - 1\}$. Finally, we write $\nnR = [0, \infty)$ for the set of non-negative real numbers.  For $a,b \in \nnR$ we use the notation $[a,b)$ and $(a,b]$ for half-open intervals and $[a,b]$ for closed intervals. Finally, we write $\nnR \cup \{\infty \} = [0, \infty]$.

\subsection{Reference time and clocks} 

Throughout this work, we assume a global \emph{reference time} that is \emph{not} available to the nodes in the distributed system. The reference time is only used to reason about the behaviour of the system. A \emph{clock} is a strictly increasing function $C \colon \nnR \to \nnR$ that maps the reference time to local (perceived) time. That is, at reference time $t$ clock $C$ indicates that the local time is $C(t)$. We say that a clock $C$ has \emph{drift} at most $\vartheta -1 > 0$ if for any $t,t' \in \nnR$, where $t < t'$, the clock satisfies
\[
 t' - t \le C(t') - C(t) \le \vartheta(t' - t).
\]
That is, if we have two such clocks, then their measurements of elapsed time are at most factor $\vartheta$ apart.

\subsection{The bounded-delay model}

We consider a \emph{bounded-delay message-passing model} of distributed computation. The system is modelled as a fully connected network of $n$ nodes, where $V$ denotes the set of all nodes. We assume that each node has a unique identifier from the set $[n]$. Each node $v \in V$ has local clock $C(v)$ with maximum drift $\vartheta - 1$ for a known global constant $\vartheta > 1$. We assume that the nodes cannot directly read their local clock values, but instead they can set up local timeouts of predetermined length. That is, a node $v$ can request to be signalled after $T$ time units have passed on the node's own local clock $C(v)$ since the timeout was started.

For communication, we assume sender authentication, that is, each node can distinguish the senders of the messages it receives. In other words, every incoming communication link is labelled with the identifier of the sender. Unlike in fully synchronous models, where communication and computation proceeds in lock-step at all nodes, we consider a model in which each message has an associated delay in $(0, d)$. For simplicity, we assume that the \emph{maximum delay} $d \in \Theta(1)$ is a known constant and we consider $d$ as the basic time unit in the system. We note that even though we assume continous, real-valued clocks, any constant offset in clock readings, e.g.\ due to discrete clocks, can be modelled by increasing $d$ if needed. 

We assume that the system can experience transient faults that arbitrarily corrupt the state of the entire system; we formally define below what this entails in our model. Once the transient faults cease, we assume that up to $f$ of the $n$ nodes in the system may remain Byzantine faulty, that is, they have arbitrary (mis)behaviour and do not necessarily follow the given protocol. We use $F \subseteq V$, where $|F| \le f$, to denote an arbitrary set of \emph{faulty nodes} and $G = V \setminus F$ is the set of \emph{correct} nodes. 

\subsection{Algorithms, configurations, and executions}

%###
\paragraph{Algorithms.}
%###
We assume that each node executes a finite state machine whose state transitions can depend on the current state of the node, the set of recently received messages, and local timeouts. Formally, an algorithm is a tuple $\vec A = ( \mathcal{S}, \mathcal{P}, \mathcal{M},\mathcal{T}, \delta, \mu )$, where 
\begin{itemize}[noitemsep]
    \item $\mathcal{S}$ is a finite set of \emph{states}, 
    \item $\mathcal{P} \subseteq \mathcal{S}$ is a subset of states that trigger a \emph{pulse event},
    \item $\mathcal{M}$ is a finite set of messages, 
    \item $\mathcal{T} \subseteq \nnR \times 2^\mathcal{S}$ is a finite set of \emph{timers}, 
    \item $\delta \colon V \times \mathcal{S} \times \mathcal{M}^n \times \{0,1\}^{h} \to \mathcal{S}$, where $h=|\mathcal{T}|$, is the state transition function, and 
    \item $\mu \colon V \times V \times \mathcal{S} \to \mathcal{M}$ is a message function.
\end{itemize}
We now explain in detail how the system state evolves and algorithms operate.
%###
\paragraph{Local configurations, timers, and timeouts.}
%###
The local configuration $x(v,t)$ of a node $v \in G$ at time $t \in \nnR$ consists~of 
\begin{enumerate}[noitemsep]
    \item its current state $s(v,t) \in \mathcal{S}$, 
    \item the state of its \emph{input channels} $m(v,t) \in \mathcal{M}^n$, 
    \item its local clock value $C(v,t) \in \nnR$, and
    \item timer states $T_{k}(v,t) \in [0,T_k]$ for each $(T_k, S_k) \in \mathcal{T}$ and $k \in [h]$. 
\end{enumerate}
Recall that we assume that the transient faults have left the system in an arbitrary state at time $t=0$.
This entails that for each node $v \in V$ the initial values at $t=0$ for (1)--(5) are arbitrary.

In the following, we use the shorthand $T_k$ for timer $(T_k,S_k)\in \mathcal{T}$. We say that timer $T_k$ of node $v$ expires at time $t$ if $T_k(v,t)$ changes to $0$ at time $t$. It is expired at time $t$ if $T_k(v,t)=0$. Timers may cause state transitions of nodes when expiring. Let $e(v,t)\in \{0,1\}^h$ indicate which timers are expired, that is, $e_k(v,t)=1$ if $T_k$ is expired and $e(v,t)=0$ otherwise. If at time $t$ the value of either $m(v,t)$ changes (that is, the input channels of node $v \in G$ are updated due to a received message) or some local timer expires, the node updates its current state to $s = \delta(v, s', m(v,t), e(v,t))$, where $s'$ is the node's state prior to this computation. If $s \neq s'$, we say that node $v$ transitions to state $s$ at time $t$ (and write $s(v,t)=s$). We remark that this definition allows for the possibility that state transitions happen in arbitrary short succession. However, our algorithms are designed such that only a (small) constant number of transitions is possible in constant time, and computational delays can be treated by interpreting them as part of communication delays.

For convenience, let us define the predicate $\Delta(v,s,t) = 1$ if $v \in G$ transitions to $s$ at time $t$ and $\Delta(v,s,t) = 0$ otherwise. When node $v$ transitions to state $s \in \mathcal{S}$, it resets all timers $(T_k,S_k)$ for which $s \in S_k$. Accordingly, at each time $t > 0$, the timer state is defined as
\[
 T_{k}(v,t) = \max \{ 0, T_k(v,t_{\text{reset}}) - (C(v,t)-C(v,t_{\text{reset}} ) \},
\]
where $t_{\text{reset}}$ is the most recent time node $v$ reset the timer $T_k$ or time $0$, that is,
\[
t_{\text{reset}} = \max (\{ 0 \} \cup \{ t' \le t : \Delta(v,s,t') = 1, s \in S \})	.
\]
Note that $T_k(v,t_{\text{reset}})=T_k$ unless $t_{\text{reset}}=0$, since with the exception of the arbitrary initial states the timer state is reset to $T_k$ at time $t_{\text{reset}}$.

The bottom line is that, at all times $t\in \nnR$, the timer state $T_{k}(v,t) \in [0, T_k]$ indicates how much time needs to pass on the local clock $C(v,\cdot)$ of node $v$ until the timer expires;
the rather involved definition of how timers behave in order to achieve this property is owed to the requirement of self-stabilisation.

%###
\paragraph{Communication.}
%###
We say that a node $u$ sends the message $\mu(u,v,s(u,t)) \in \mathcal{M}$ to node $v$ at time $t$, if the value of $\mu(u,v,s(u,t))$ changes at time $t$. Moreover, node $u$ is said to \emph{broadcast} the message $a$ at time $t$ if it sends the message $a$ to every $v$ at time $t$.

As we operate in the bounded-delay setting, sent messages do not arrive at their destinations immediately. To model this, let the \emph{communication delay function} $d_{uv} \colon \nnR \to \nnR$ be a strictly increasing function such that $0 < d_{uv}(t) - t < d$. The input channels of node $u \in G$ satisfy 
\[
 m_{v}(u,d_{uv}(t)) = \begin{cases}
                       \mu(v,u,s(v,t)) & \text{if } v \in G \\
                       b(u,v,t) & \text{otherwise,}
                      \end{cases}
\]
where $b(u,v,t) \in \mathcal{M}$ is the message a faulty node $v \in F$ decides to transmit to a correct node $u \in G$ at time $t$. We assume the adversary can freely choose the communication delay functions $d_{uv}$. Thus, the adversary can control what correct nodes receive from faulty nodes \emph{and} how long the messages sent by correct nodes traverse (up to the maximum delay bound $d$). Intuitively, $m_v(u,t) \in \mathcal{M}$ denotes the most recent message node $v$ received from node $u$ at time $t$. Since transient faults may result in arbitrarily corrupted communication channels at time $0$, we assume that $m_{v}(u,t) \in \mathcal{M}$ is arbitrary for $t < d_{uv}(0)$. 

%###
\paragraph{The adversary and executions.}
%###
After fixing $f,n \in \N$ and an algorithm $\vec A$, we assume that an adversary chooses 
\begin{enumerate}[noitemsep]
  \item the set $F \subseteq V$ of faulty nodes such that $|F|\leq f$, 
  \item the initial local configuration $x(v, 0)$ for all $v \in V$,
\end{enumerate}
and for all $t \in \nnR$ and any $u,v \in V$
\begin{enumerate}[resume,noitemsep]
  \item the local clock values $C(v, t)$, 
  \item the message delay functions $d_{uv}(t)$, and 
  \item the messages $b(u,v,t)$ sent by faulty nodes.
\end{enumerate}

Note that if the algorithm $\vec A$ is deterministic, then the adversary's choices for (1)--(5) together with $\vec A$ determine the execution, that is, local configurations $x(v,t)$ for all $v \in G$ and $t \ge 0$. Randomisation may be used in black-box calls to a consensus subroutine only. For brevity, we postpone the discussion of randomisation to \sectionref{sec:randomisation}, which covers the results obtained by utilising randomised consensus routines. We remark that minor adjustments to the above definitions may be necessary depending on the precise model of randomness and power of the adversary;
however, this does not affect the reasoning about our framework, which is oblivious to how the employed consensus routine operates.

%###
\paragraph{Logical state machines and sliding window memory buffers.}
%###
For ease of presentation, we do not describe our algorithms in the above low-level state machine formalism, but instead use high-level state machines, where state transitions are conditioned on timer expiration and sliding window memory buffers. While these are not part of the above described formalism, they are straightforward to implement using additional local timers and states.

Formally, we use a set $\mathcal{X}$ of \emph{logical states} and identify each state $s \in \mathcal{S}$ with a logical state $\ell(s) \in \mathcal{X}$. That is, we have a surjective projection $\ell \colon \mathcal{S} \to \mathcal{X}$ that maps each state $s$ onto its equivalence class $\ell(s)$, i.e., the logical state.  In addition, we associate all timers with some logical state, that is, for every $(T_k, S_k) \in \mathcal{T}$, we have that $S_k \in \mathcal{X}$ is an equivalence class of states. 

We employ \emph{sliding window buffers} in our algorithms. A sliding window buffer of length $T$ stores the set of nodes from which (a certain type of) a message has been received within time $T$ on the node's local clock. Since the local configuration $x(v,0)$ of a node $v$ is arbitrary at time 0, we have that by time $T+d$ the contents of the sliding window buffer are guaranteed to be valid: if the buffer of $v \in G$ contains a message $m$ from $u \in G$ at time $t \ge T+d$, then $u$ must have sent an $m$ message to $v$ during the interval $(t-T-d, t)$ of reference time. Vice versa, if $u$ sends a message $m$ at time $t$ to $v$, the buffer is guaranteed to contain the message during the interval $(t+d, t+T/\vartheta)$ of reference time. We also allow the algorithms to \emph{clear} the sliding window buffers at any point in time by removing all the messages currently contained in the buffer. That is, when a node clears its sliding window buffer at time $t$, then the buffer contains no message seen before time $t$.

\subsection{Pulse synchronisation algorithms}

 In the pulse synchronisation problem, the task is to have all the correct nodes locally generate pulse events in an almost synchronised fashion, despite arbitrary initial states and the presence of Byzantine faulty nodes. In addition, these pulses have to be well-separated. Let $p(v,t) \in \{0,1\}$ indicate whether a correct node $v \in G$ generates a pulse at time $t$. Moreover, let $p_k(v,t) \in [t,\infty)$ denote the time when node $v$ generates the $k$th pulse event at or after time $t$ and $p_k(v,t) = \infty$ if no such time exists. We say that the system has stabilised from time $t$ onwards if 
\begin{enumerate}[noitemsep]
 \item $p_1(v,t) \le t + \Phi^+$ for all $v \in G$,
 \item $|p_k(v,t) - p_k(u,t)| < \sigma$ for all $u,v \in G$ and $k \geq 1$,
 \item $\Phi^- \le p_{k+1}(v,t) - \min \{ p_k(u,t) : u \in G \} \le \Phi^+$ for all $v \in G$ and $k\geq 1$,
\end{enumerate}
where $\Phi^-$ and $\Phi^+$ are the accuracy bounds controlling the separation of the generated pulses. That is, (1) all correct nodes generate a pulse during the interval $[t,t+\Phi^+]$, (2) the $k$th pulse of any two correct nodes is less than $\sigma$ time apart, and (3) for any pair of correct nodes their subsequent pulses are at least $\Phi^--\sigma$, but at most $\Phi^+$ time apart.

We say that $\vec A$ is an $f$-resilient pulse synchronisation algorithm with \emph{skew} $\sigma$ and \emph{accuracy} $\Phi = (\Phi^-, \Phi^+)$ with stabilisation time $T(\vec A)$, if for any choices of the adversary such that $|F| \le f$, there exists a time $t \le T(\vec A)$ such that the system stabilises from time $t$ onwards. This scenario is illustrated below:
\begin{center}
 \includegraphics[page=4]{figures.pdf}
\end{center}

Finally, we call a pulse synchronisation algorithm $\vec A$ a \emph{$T$-pulser} if the accuracy bounds satisfy $\Phi^-, \Phi^+ \in \Theta(T)$ and $\vec A$ has skew $\sigma \le 2d$. We use $M(\vec A)$ to denote the maximum number of bits a correct node sends on a channel per unit time when executing algorithm $\vec A$.

\subsection{Resynchronisation algorithms}

In our pulse synchronisation algorithms, we use so-called resynchronisation pulses to facilitate stabilisation. The resychronisation pulses are provided by resynchronisation algorithms that solve a weak variant of pulse synchronisation: the guarantee is that eventually all correct nodes generate a single resynchronisation pulse almost synchronously, which is followed by a long period of silence (i.e.\ no new resynchronisation pulse). At all other times, the behaviour can be arbitrary.

Formally, we say that $\vec B$ is an $f$-resilient resynchronisation algorithm with skew $\rho$ and separation window $\Psi$ that stabilises in time $T(\vec B)$, if the following holds: for any choices of the adversary such that $|F| \le f$, there exists a time $t \le T(\vec B)$ such that every correct node $v \in G$ locally generates a \emph{resynchronisation pulse} at time $r(v) \in [t, t+\rho)$ and no other resynchronisation pulse before time $t+\rho + \Psi$. We call such a resynchronisation pulse \emph{good}. In particular, we do not impose any restrictions on what the nodes do outside the interval $[t, t+ \rho + \Psi)$, that is, there may be \emph{spurious} resynchronisation pulses outside this interval:
\begin{center}
\includegraphics[page=5]{figures.pdf}
\end{center}
Again, we denote by $M(\vec B)$ the maximum number of bits a correct node sends on a channel per unit time when executing $\vec B$.

\subsection{Synchronous consensus routines}

As we rely on synchronous consensus algorithms, we briefly define the synchronous model of computation for the sake of completeness. In the synchronous model, the computation proceeds in discrete rounds, that is, the nodes have access to a common global clock. In each round $r \in \N$ the nodes (1) send messages based on their current state, (2) receive messages, and (3) perform local computations and update their state for the next round.

In synchronous binary consensus, we assume that each $v \in V$ is given a private input bit $x(v) \in \{0,1\}$, starts from a fixed initial state (no self-stabilisation), and is to compute output $y(v)\in \{0,1\}$. However, there are $f$ Byzantine faulty nodes. An $f$-resilient synchronous consensus routine $\vec C$ with round complexity $T(\vec C)$ guarantees:
\begin{enumerate}[noitemsep]
  \item \textbf{Agreement:} There exists $y \in \{0,1\}$ such that $y(v)=y$ for all $v \in G$.
  \item \textbf{Validity:} If for $x \in \{0,1\}$ it holds that $x(v)=x$ for all $v \in G$, then $y=x$.
  \item \textbf{Termination:} Each $v \in G$ decides on $y(v)$ and terminates by round $T(\vec C)$.
\end{enumerate}
We use $M(\vec C)$ to denote the maximum number of bits any $v \in G$ sends to any other node in a single round of any execution of $\vec C$. 

\section{The transformation framework}\label{sec:framework}

Our main contribution is a modular framework that allows us to turn any \emph{non-self-stabilising} synchronous consensus algorithm into a self-stabilising pulse synchronisation algorithm in the bounded-delay model. In particular, the transformation yields only a small overhead in time and communication complexity. As our construction is relatively involved, we opt to present it in a top-down fashion. First, we give our main theorem together with its corollaries. Then we state the auxiliary results we need to prove the main theorem and later discuss how these auxiliary results are established.

\subsection{The main result} 

Before we formally state our main result, we make the following definition.

\begin{definition}[Family of consensus routines]
 Let $R,M,N \colon \N_0 \to \N_0$ be functions that satisfy the following conditions:
\begin{enumerate}[noitemsep,label=(\roman*)]
    \item for any $f_0,f_1\in \N$, we have $N(f_0+f_1) \leq N(f_0) + N(f_1)$, and 
    \item both $M(f)$ and $R(f)$ are increasing. 
\end{enumerate}
 We say that $\langle \mathcal{C}, R, M, N \rangle$ is a \emph{family of synchronous consensus routines} if for any $f \ge 0$ and $n\geq N(f)$, there exists a synchronous consensus routine $\vec C \in \mathcal{C}$ such that
\begin{itemize}[noitemsep]
 \item $\vec C$ runs on $n$ nodes and is $f$-resilient,
 \item each correct node terminates in $T(\vec C) = R(f)$ rounds, 
 \item and each correct node sends at most $M(\vec C) = M(f)$ bits to any other node per round.
\end{itemize}
\end{definition}

Our main technical result states that given such a family of consensus routines, we can obtain pulse synchronisation algorithms with only small additional overhead. We emphasise that the algorithms in $\mathcal{C}$ are not assumed to be self-stabilising. 

\begin{restatable}{theorem}{masterthm}\label{thm:master}
    Let $\langle \mathcal{C}, R, M, N \rangle$  be a family of synchronous consensus routines and $f \ge 0$, $n \geq N(f)$, and $1<\vartheta \leq 1.004$. Then there exists an $f$-resilient $R(f)$-pulser
    $\vec A$ whose stabilisation time $T(\vec A)$ and number of bits $M(\vec A)$ sent over each channel per time unit satisfy 
\[
    T(\vec A) \in O\left(d+\sum_{k=0}^{\lceil \log f \rceil} R(2^k) \right) \quad \text{and} \quad M(\vec A) \in O\left(1+\sum_{k=0}^{\lceil \log f \rceil} M(2^k) \right),
\]
where the sums are empty when $f=0$.
\end{restatable}

In the deterministic case, the \emph{phase king algorithm}~\cite{berman92optimal} provides a family of synchronous consensus routines that satisfy the requirements. Since in the phase king algorithm all nodes communicate by broadcasts (i.e., send the same information to all other nodes) and the additional communication by our framework satisfies this property as well, the same is true for the derived pulser. Moreover, the phase king protocol achieves optimal resilience~\cite{pease80reaching} with $N(f)=3f+1$, constant message size $M(f) \in O(1)$, and asymptotically optimal~\cite{fischer82lower} round complexity $R(f) \in \Theta(f)$. Thus, this immediately yields the following result.
\begin{corollary}\label{coro:deterministic}
    For any $f \ge 0$ and $n > 3f$, there is a deterministic $f$-resilient $\Theta(f)$-pulser over $n$ nodes that stabilises in $O(f)$ time and has correct nodes broadcast $O(\log f)$ bits per time unit.
\end{corollary}

Employing randomised consensus algorithms in our framework is straightforward. We now summarise the main results related to randomised pulse synchronisation algorithms; the details are discussed later in \sectionref{sec:randomisation}. First, by applying our construction to a fast and communication-efficient randomised consensus algorithm, e.g.\ the one by King and Saia~\cite{king11breaking}, we get an efficient randomised pulse synchronisation algorithm.

\begin{restatable}{corollary}{cororandKS}\label{coro:randKS}
    Suppose we have private channels. For any $f \ge 0$, constant $\varepsilon>0$, and $n > (3+\varepsilon)f$, there is a randomised $f$-resilient $(\polylog f)$-pulser over $n$ nodes that stabilises in $\polylog f$ time w.h.p.\ and has nodes broadcast $\polylog f$ bits per time unit.
\end{restatable}

We can also utilise the constant expected time protocol by Feldman and Micali~\cite{feldman95optimal}. With some care, we can show that for $R(f)\in O(1)$, Chernoff's bound readily implies that the stabilisation time is not only in $O(\log f)$ in expectation, but also with high probability. 

\begin{restatable}{corollary}{cororandFM}\label{coro:feldmanmicali}
    Suppose we have private channels. For any $f \ge 0$ and $n > 3f$, there is a randomised $f$-resilient $\Theta(\log f)$-pulser over $n$ nodes that stabilises in $O(\log f)$ time w.h.p.\ and has nodes broadcast $\poly f$ bits per time unit.
\end{restatable}

\subsection{Proof of \texorpdfstring{\theoremref{thm:master}}{the main theorem}}

The proof of the main result takes an inductive approach. In the inductive step, we assume two pulse synchronisation algorithms with small resilience. We then use these to construct (via some hoops we discuss later) a new pulse synchronisation algorithm with higher resilience. This step is formalised in the following technical lemma, which we prove later.

\begin{lemma}\label{lemma:pulser-induction}
Let $f,n_0,n_1 \in \N$ and define the values
\[
 n=n_0+n_1, \quad f_0 = \lfloor (f-1)/2 \rfloor, \quad f_1 = \lceil (f-1)/2 \rceil.
\]
Suppose there exists
\begin{itemize}[noitemsep]
    \item for both $i \in \{0,1\}$ an $f_i$-resilient $R$-pulser $\vec A_i$ that runs on $n_i$ nodes with accuracy $\Phi_i = (\Phi^-_i, \Phi^+_i)$ satisfying $\Phi^+_i/ \Phi^-_i \le \varphi$ for a sufficiently small constant $\varphi>\vartheta$, and
    \item an $f$-resilient consensus routine $\vec C$ for a network of $n$ nodes that has running time $R$ and uses messages of at most $M$ bits.
\end{itemize}
    Then there exists a $R$-pulser $\vec A$ that
\begin{itemize}[noitemsep]
 \item runs on $n$ nodes and has resilience $f$,
 \item stabilises in time $T(\vec A) \in \max \{ T(\vec A_0), T(\vec A_1) \} + O(R)$,
 \item sends $M(\vec A) \in \max \{ M(\vec A_0), M(\vec A_1) \} + O(M)$ bits over each channel per time unit,
 \item has skew $2d$, and 
 \item has accuracy bounds $\Phi^-$ and $\Phi^+$ that satisfy $\Phi^+/\Phi^- \le \varphi$.
\end{itemize}
\end{lemma}

We observe that \theoremref{thm:master} is a relatively straightforward consequence of the above lemma.

\masterthm*
\begin{proof}
    We prove the claim for $f\in \N_0 = \{0 \} \cup \bigcup_{k\in \N_0}([2^k,2^{k+1})\cap \N)$ using induction on $k$. As base case, we use $f=0$. This is trivial for all $n > 0$, as the following algorithm shows. Let $T>0$ be arbitrary. We can pick a single node as a designated leader who generates a pulse whenever $T$ time units have passed on its local clock. Whenever the leader node pulses, all other nodes observe this within $d$ time units. When the other nodes observe a pulse from the leader, they generate a pulse locally. Thus, for $f=0$ we obtain a $T$-pulser that stabilises in $O(d)$ time and sends messages of $O(1)$ bits at most once every $T/\vartheta \in \Theta(T)$ time. Choosing $T=R(0)$ and noting that $R(0)\geq 1$ (even without faults, consensus requires communication if $n>1$), the claim follows for $f=0$.

For the inductive step, consider $f \in [2^k,2^{k+1})$ and suppose that, for all $f'<2^k$ and $n'\geq N(f')$, there exists an $f'$-resilient $R(f')$-pulser algorithm $\vec B$ on $n'$ nodes with
\[
    T(\vec B) \le \alpha \left(d+\sum_{k'=0}^{\lceil \log f' \rceil} R(2^{k'}) \right) \text{ and }
 M(\vec B) \le \beta \left(1+\sum_{k'=0}^{\lceil \log f' \rceil} M(2^{k'})\right),
\]
where $\alpha$ and $\beta$ are sufficiently large constants. In particular, we can now apply \lemmaref{lemma:pulser-induction} with $f_0, f_1 \leq f/2<2^k$ and any $n \geq N(f)$, as $N(f) \geq N(f_0) + N(f_1)$ guarantees that we may choose some $n_0 \geq N(f_0)$ and $n_1 \geq N(f_1)$ such that $n=n_0+n_1$. This yields an $f$-resilient $R(f)$-pulser $\vec A$ over $n$ nodes, with stabilisation time
\begin{align*}
T(\vec A) &\le \max \{ T(\vec A_0), T(\vec A_1) \} + \gamma R(f) \\
          &\le \alpha \left(d+\sum_{k=0}^{\lceil \log f/2 \rceil} R(2^k)\right) + \gamma R(2^{\lceil\log f\rceil}), \\
          &\le \alpha \left(d+\sum_{k=0}^{\lceil \log f \rceil} R(2^k)\right),
\end{align*}
where $\gamma\leq \alpha$ is a constant; the second step uses that $R(f)$ is increasing. Similarly, the bound on the number of sent bits follows from \lemmaref{lemma:pulser-induction} and the induction assumption:
\begin{align*}
M(\vec A) &\le \max \{ M(\vec A_0), M(\vec A_1) \} + \gamma' M(f) \\
          &\le \beta \left(1+\sum_{k=0}^{\lceil \log f/2 \rceil} M(2^k) \right) + \gamma' M(2^{\lceil\log f\rceil}), \\
          &\le \beta \left(1+\sum_{k=0}^{\lceil \log f \rceil} M(2^k)\right),
\end{align*}
where $\gamma'\leq \beta$ is a constant and we used that $M(f)$ is increasing.
\end{proof}

\subsection{The auxiliary results} 

In order to show \lemmaref{lemma:pulser-induction}, we use two main ingredients: (1) a pulse synchronisation algorithm whose stabilisation mechanism is triggered by a resynchronisation pulse and (2) a resynchronisation algorithm providing the latter. These ingredients are formalised in the following two theorems which are proven in Sections~\ref{sec:pulse-gen} and~\ref{sec:resync}, respectively.

\begin{restatable}{theorem}{resynctopulse}\label{thm:resync-to-pulse}
    Let $f \ge 0$,  $n > 3f$, and $(2+\sqrt{32})/7 > \vartheta > 1$. Suppose for a network of $n$ nodes there exist
\begin{itemize}[noitemsep]
 \item an $f$-resilient synchronous consensus algorithm $\vec C$, and
 \item an $f$-resilient resynchronisation algorithm $\vec B$ with skew $\rho \in O(d)$ and sufficiently large separation window $\Psi\in O(R)$ that tolerates clock drift of $\vartheta$,
\end{itemize}
    where $\vec C$ runs in $R=R(f)$ rounds and lets nodes send at most $M=M(f)$ bits per round and channel. Then there exists $\varphi_0(\vartheta) \in 1 + O(\vartheta-1)$ such that for any constant $\varphi > \varphi_0(\vartheta)$ and sufficiently large $T \in O(R)$, there exists an $f$-resilient pulse synchronisation algorithm $\vec A$ for $n$ nodes that 
\begin{itemize}[noitemsep]
 \item has skew $\sigma = 2d$,
 \item satisfies the accuracy bounds $\Phi^- = T$ and $\Phi^+ = T \varphi $,
 \item stabilises in $T(\vec B)+O(R)$ time, and
 \item has nodes send $M(\vec B)+O(M)$ bits per time unit and channel.
\end{itemize}
\end{restatable}

To apply the above theorem, we require suitable consensus and resynchronisation algorithms. We rely on consensus algorithms from prior work and construct efficient resynchronisation algorithms ourselves. The idea is to combine pulse synchronisation algorithms that have \emph{low resilience} to obtain resynchronisation algorithms with \emph{high resilience}.

\begin{restatable}{theorem}{resynctheorem}\label{thm:resync}
Let $f,n_0,n_1 \in \N$ and $1 < \vartheta \leq 1.004$. Define 
\[
 n=n_0+n_1, \quad f_0 = \lfloor (f-1)/2 \rfloor, \quad f_1 = \lceil (f-1)/2 \rceil.
\]
 For any $\Psi \in \Omega(1)$ and sufficiently small constant $\varphi > \varphi_0(\vartheta)$, there exists a bound $T_0 \in \Theta(\Psi)$ such that the following claim holds. 
    If for both $i \in \{0,1\}$ there exists pulse synchronisation algorithm $\vec A_i$ that 
\begin{itemize}[noitemsep]
    \item runs on $n_i$ nodes and has resilience $f_i$,
    \item has skew $\sigma = 2d$, and
    \item has accuracy bounds $\Phi_i^- = T$ and $\Phi_i^+ = T \varphi$, where $T_0 \le T$ and $T \in O(\Psi)$,
\end{itemize}
    then there exists a resynchronisation algorithm $\vec B$ that 
\begin{itemize}[noitemsep]
    \item runs on $n$ nodes and has resilience $f$,
    \item has skew $\rho \in O(d)$ and separation window of length $\Psi$,
    \item generates a resynchronisation pulse by time $T(\vec B) \in \max \{ T(\vec A_0), T(\vec A_1) \} + O(\Psi)$, and
    \item has nodes send $M(\vec B) \in \max \{ M(\vec A_0), M(\vec A_1) \} + O(1)$ bits per time unit and channel.
\end{itemize}
\end{restatable}

Given a suitable consensus algorithm, one can readily combine Theorems~\ref{thm:resync-to-pulse} and~\ref{thm:resync} to obtain \lemmaref{lemma:pulser-induction}. Note that for both theorems, it turns out that $\varphi_0(\vartheta) = 1+5(\vartheta-1)/(2+2\vartheta-3\vartheta^2)$ will do. Therefore, we can reduce the problem of constructing an $f$-resilient pulse synchronisation algorithm to finding algorithms that tolerate up to $\lfloor f/2\rfloor$ faults and recurse; \figureref{fig:overview_simple} illustrates how these two types of algorithms are interleaved.

\begin{figure}[t!]
\begin{center}
 \includegraphics[page=18,scale=0.99]{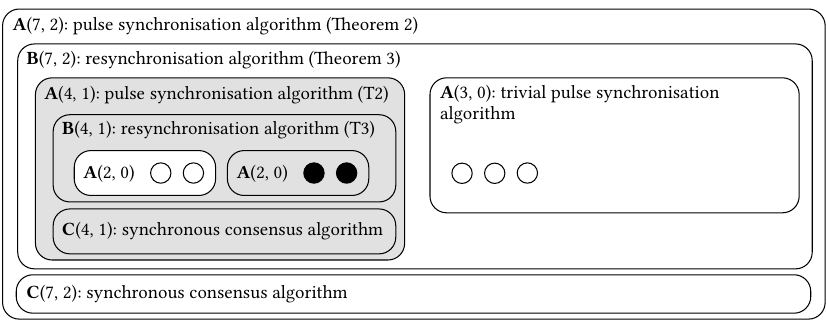}
\end{center}
    \caption{Recursively building a $2$-resilient pulse synchronisation algorithm $\vec A(7,2)$ over $7$ nodes. The construction utilises low resilience pulse synchronisation algorithms to build high resilience resynchronisation algorithms which can then be used to obtain highly resilient pulse synchronisation algorithms. Here, the base case consists of trivial $0$-resilient pulse synchronisation algorithms $\vec A(2,0)$ and $\vec A(3,0)$ over 2 and 3 nodes, respectively. Two copies of $\vec A(2,0)$ are used to build a 1-resilient resynchronisation algorithm $\vec B(4,1)$ over 4 nodes using \theoremref{thm:resync}. The resynchronisation algorithm $\vec B(4,1)$ is used together with a synchronous consensus algorithm $\vec C(4,1)$ to obtain a pulse synchronisation algorithm $\vec A(4,1)$ via \theoremref{thm:resync-to-pulse}. Now, the 1-resilient pulse synchronisation algorithm $\vec A(4,1)$ over 4 nodes is used together with the trivial 0-resilient algorithm $\vec A(3,0)$ to obtain a 2-resilient resynchronisation algorithm $\vec B(7,2)$ for 7 nodes. This is then used together with a 2-resilient consensus algorithm $\vec C(7,2)$ to obtain the final pulse synchronisation algorithm $\vec A(7,2)$. White nodes represent correct nodes and black nodes represent faulty nodes. The gray blocks contain too many faulty nodes for the respective algorithms to correctly operate, and hence, they may have arbitrary output. 
\label{fig:overview_simple}}
\end{figure}

\subsection{Proof of \texorpdfstring{\lemmaref{lemma:pulser-induction}}{Lemma 1}}
\begin{proof}
From \theoremref{thm:resync}, we get that for any sufficiently large $\Psi \in \Theta(R)$, there exists a resynchronisation algorithm $\vec B$ with skew $\rho\in O(d)$ and separation window of length $\Psi$ that 
\begin{itemize}[noitemsep]
  \item runs on $n$ nodes and has resilience $f$,
  \item stabilises in time $\max \{ T(\vec A_0), T(\vec A_1) \} + O(\Psi)=\max \{ T(\vec A_0),  T(\vec A_1) \} + O(R)$, and
  \item has nodes send $\max \{ M(\vec A_0), M(\vec A_1) \} + O(1)$ bits per time unit and channel.
\end{itemize}
We feed $\vec B$ and $\vec C$ into \theoremref{thm:resync-to-pulse}, yielding a pulse synchronisation algorithm $\vec A$ with the claimed properties, as the application of \theoremref{thm:resync-to-pulse} increases the stabilisation time by an additional $O(R)$ time units and adds $O(M)$ bits per time unit and channel.
\end{proof}

\subsection{Organisation of the remainder of the paper}

We dedicate the remaining sections to fill in the details we have omitted above. Namely, 
\begin{itemize}
    \item \sectionref{sec:bz-pulse-sync} describes a Byzantine-tolerant pulse synchronisation algorithm that is \emph{not} self-stabilising. We utilise the algorithm in \sectionref{sec:pulse-gen}, but the section also serves to provide a gentle introduction to the notation and style of proofs we use in the following sections.
 \item \sectionref{sec:pulse-gen} gives the proof of \theoremref{thm:resync-to-pulse}.
 \item \sectionref{sec:resync} gives the proof of \theoremref{thm:resync}.
 \item \sectionref{sec:randomisation} extends our framework to operate with randomised consensus algorithms. This establishes Corollaries~\ref{coro:randKS}~and~\ref{coro:feldmanmicali}.
\end{itemize}

\section{Byzantine-tolerant pulse synchronisation}\label{sec:bz-pulse-sync}

In this section, we describe a \emph{non-self-stabilising} pulse synchronisation algorithm, which we utilise later in our construction of the self-stabilising algorithm. The algorithm given here is a variant of the Byzantine fault-tolerant clock synchronisation algorithm by Srikanth and Toeug~\cite{srikanth87clock} that avoids transmitting clock values in favor of unlabelled pulses. 

\subsection{Pulse synchronisation without self-stabilisation}

As said, we do not require self-stabilisation for now. However, instead of assuming a synchronous start, we devise an algorithm where all correct nodes synchronously start generating pulses once all correct nodes have received an \emph{initialisation signal} within a short time window. We will later show how to generate such initialisation signals in a self-stabilising manner. In particular, we allow that \emph{before} receiving an initialisation signal, a correct node can have arbitrary behaviour, but \emph{after} a correct node has received this signal, it waits until all correct nodes have received the signal and start to synchronously generate well-separated pulses:
\begin{center}
 \includegraphics[page=15]{figures.pdf}
\end{center}

In the following, suppose that all correct nodes receive an initialisation signal during the time window $[0, \tau)$. In other words, nodes can start executing the algorithm at different times, but they all do so by some bounded (possibly non-constant) time $\tau$. When a node receives the initialisation signal, it immediately transitions to a special \sreset\ state, whose purpose is to consistently initialise local memory and wait for other nodes to receive the initialisation signal as well; before this correct nodes can have arbitrary behaviour.

Later, we will repeatedly make use of this algorithm as a subroutine for a self-stabilising algorithm and we need to consider the possibility that there are still messages from earlier (possibly corrupted) instances in transit, or nodes may be executing a previous instance in an incorrect way. Given the initialisation signal, this is easily overcome by waiting for sufficient time before leaving the starting state: waiting $\vartheta (\tau + d) \in O(\tau)$ local time guarantees that (i) all correct nodes transitioned to the starting state and (ii) all messages sent before these transitions have arrived. Clearing memory buffers when leaving the starting state thus ensures that no obsolete information from previous instances is stored by correct nodes.

The goal of this section is to establish the following theorem:
\begin{restatable}{theorem}{stpulser}\label{thm:st-pulser}
Let $n > 1$, $f < n/3$, and $\tau > 0$. If every correct node receives an initialisation signal during $[0, \tau)$, then there exists a pulse synchronisation algorithm $\vec P$ that
\begin{itemize}[noitemsep]
    \item runs on $n$ nodes and has resilience $f$,
    \item has $v \in G$ generate its first pulse (after the initialisation signal) at a time $t_0(v) \in O(\vartheta^2 d \tau)$,
 \item has skew $2d$,
 \item has accuracy bounds $\Phi^- \in \Omega(\vartheta d)$ and $\Phi^+ \in O(\vartheta^2 d)$, and
 \item lets each node broadcast at most one bit per time unit.
 \end{itemize}
\end{restatable}

\subsection{Description of the algorithm}

\begin{figure}
\begin{center}
 \includegraphics[page=9]{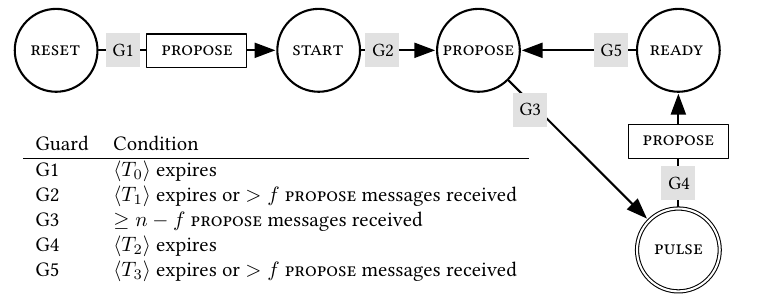}
\end{center}
    \caption{The state machine for the non-self-stabilising pulse synchronisation algorithm. State transitions occur when the condition of the guard in the respective edge is satisfied (labelled gray boxes). Here, all transition guards involve checking whether a local timer expires or a node has received sufficiently many messages from nodes in state \spropose. The only communication that occurs is when a node transitions to state \spropose; when this happens a node broadcasts this information to all others. The notation $\langle T \rangle$ indicates the expiration of a timer of length $T$ that was reset when transitioning to the current state, that is, $T$ time units have passed on the local clock since the transition to the current state. The box labelled \textsc{propose} indicates that a node clears its sliding window messages buffers when transitioning from \sreset\ to \sstart\ and \spulse\ to \sready. That is, the node forgets who it has ``seen'' in \spropose\ the previous iteration. The algorithm assumes that during the interval $[0, \tau)$ all nodes transition to \sreset. This starts the initialisation phase of the algorithm. Eventually, all nodes transition to \spulse\ within a short time window and start executing the algorithm. Whenever a node transitions to state \spulse\ it generates a local pulse event. \tableref{table:st-constraints} lists the constraints imposed on the timeouts.\label{fig:bz-pulser}}
\end{figure}

The algorithm is illustrated in \figureref{fig:bz-pulser}. In the figure, the circles denote the basic logical states (\sreset, \sstart, \sready, \spropose, \spulse) of the state machine for each node. The two states \sreset\ and \sstart\ are used in the initialisation phase of the algorithm, which takes place during $[0, \tau)$ when the nodes receive their initialisation signals. In \figureref{fig:bz-pulser}, directed edges between the states denote possible state transitions and labels give the conditions (transition guards) when the transition is allowed to (and must) occur. The notation is as follows: $\langle T_k \rangle$ denotes the condition that timer $(T_k,\{s\})$ has expired, where $s$ is the state that is left when the guard is satisfied. 

The boxes labelled with \spropose\ indicate that when a node transitions to the designated state, it clears its memory buffers immediately. Throughout this section, we use $H(v,t)$ to denote the total number of nodes at time $t$ from which $v$ has received a $\spropose$ message since last clearing its memory buffers. For the purposes of our analysis, we use $K(v,t)$ to denote the numher of \emph{correct} nodes from which $v$ has received a $\spropose$ message at time $t$ since last clearing its sliding window buffers. Moreover, without loss of generality, we may assume that the sliding window buffers have infinite length, that is, messages never expire unless the buffer is explicitly cleared during a transition to $\sstart$ and $\sready$.

The constraints we impose on the timeouts are given in \tableref{table:st-constraints}. The expression ``$> f$ propose messages received'' denotes the condition $H(v,t) > f$. For simplicity, we assume in all our descriptions that every time a node transitions to a logical state, it broadcasts the name of the state to all other nodes. Given that we use a constant number of (logical) states per node (and in our algorithms nodes can only undergo a constant number of state transitions in constant time), this requires broadcasting $O(1)$ bits of information per time unit. In fact, closer inspection reveals that 1 bit per iteration of the cycle suffices here: the only relevant information is whether a node is in state \spropose\ or not.

\subsection{Algorithm analysis}

The algorithm relies heavily on the property that there are at most $f < n/3$ faulty nodes. This allows the use of the following ``vote-and-pull'' technique. If some correct node receives a \spropose\ message from at least $n-f$ different nodes at time $t$, then we must have that at least $n-2f > f$ of these originated from correct nodes during the interval $(t-d, t)$, as every message has a positive delay of less than $d$. Furthermore, it follows that before time $t+d$ \emph{all} correct nodes receive more than $f$ \spropose\ messages.

In particular, this ``vote-and-pull'' technique is used in the transition to states \spropose\ and \spulse. Suppose at some point all nodes are in \sready. If some node transitions to \spulse, then it must have observed at least $n-f$ nodes in \spropose\ by \guard{3}. This in turn implies that more than $f$ correct nodes have transitioned to \spropose. This in turn will (in short time) ``pull'' nodes that still remain in state \sready\ into state \spropose. Thus, \guard{3} will eventually be satisfied at all the nodes. The same technique is also used in the transition from \sstart\ to \spropose\ during the initialisation phase of the algorithm. 

\begin{remark}\label{remark:H2K}
Suppose $f < n/3$. Let $u,v \in G$, $t \ge d$ and $I = (t-d, t+d)$. If $H(v,t) \ge n-f$, then $K(u,t') > f$ for some $t' \in I$ assuming $u$ does not clear its message buffers during the interval $I$.
\end{remark}

\begin{table}
\center
\begin{tabular}{l l}
 \toprule
  \constrnumber\label{constr:bz-t0} & $T_0/\vartheta \ge \tau + d$ \\
  \constrnumber\label{constr:bz-t1} & $T_1/\vartheta \ge (1-1/\vartheta)T_0 + \tau$ \\
  \constrnumber\label{constr:bz-t2} & $T_2/\vartheta \ge 3d$ \\
  \constrnumber\label{constr:bz-t3} & $T_3/\vartheta \ge (1-1/\vartheta)T_2 + 2d$ \\
 \bottomrule
\end{tabular}
\caption{The list of conditions used in the non-self-stabilising pulse synchronisation algorithm given in \figureref{fig:bz-pulser}. Recall that $d, \vartheta \in O(1)$ and $\tau$ is a parameter of the algorithm. \label{table:st-constraints}}
\end{table}
For all $v \in G$ and $t > 0$, let $p(v,t) \in \{0,1\}$ indicate whether $v$ transitions to state \spulse\ at time $t$. That is, we have $p(v,t)=1$ if node $v \in G$ transitions to state \spulse\ at time $t$ and $p(v,t)=0$ otherwise.
\begin{lemma}\label{lemma:bz-init}
There exists $t_0 < \tau + T_0 + T_1 + d$ such that for all $v \in G$ it holds that $p(v,t)=1$ for $t \in [t_0, t_0+2d)$.
\end{lemma}
\begin{proof}
Let $v \in G$. Node $v$ receives the initialisation signal during some time $t_\text{reset}(v) \in [0, \tau)$ and transitions to state \sreset. From \sreset\ the node transitions to \sstart\ at some time $t_\text{start}(v) \in [t_\text{reset}(v) + T_0/\vartheta, t_\text{reset}(v) + T_0]$ when the timer $T_0$ in \guard{1} expires. Since $T_0/\vartheta \ge \tau + d$ by \constrref{constr:bz-t0}, we get that $t_\text{start}(v) \ge \tau + d$. Thus, for all $u,v \in G$ we have $t_\text{reset}(u) + d \le t_\text{start}(v)$. 

Moreover, $v$ transitions to \spropose\ at some time $t_\text{propose}(v) \in [t_\text{start}(v), t_\text{start}(v) + T_1]$ when \guard{2} is satisfied. Hence, any $v \in G$ transitions to state \spropose\ no later than time $t_\text{start}(v) + T_1 \le t_\text{reset}(v) + T_0 + T_1 \le \tau + T_0 + T_1$. Let $t_\text{propose} \ge \tau+d$ be the minimal time some node $v \in G$ transitions to state \spropose\ after transitioning to \sreset\ during $[0, \tau)$. 
    Observe that since a correct node $v$  clears its message buffers when transitioning from $\sreset$ to $\sstart$, we have that for any $t \in [t_\text{start}(v), t_\text{propose}) \subseteq [\tau + d, t_\text{propose})$ the sliding window memory buffer of $v$ contains no messages from correct nodes at time $t$, i.e., $K(v,t) = 0$ and $H(v,t) \le f$. Thus, node $v$ will not receive a \spropose\ message from any correct node $u \in G$ before time $t_\text{propose}$. 

Note that $t_\text{propose} \in [(T_0+T_1)/\vartheta, \tau + T_0 + T_1)$ by \guard{1} and \guard{2}. By \constrref{constr:bz-t1} and our previous bounds we have that $t_\text{propose} \ge T_0/\vartheta + (1-1/\vartheta)T_0 + \tau = T_0 + \tau \ge t_\text{start}(u)$ for any $u \in G$. Hence, after time $T_0 + \tau$, no $u \in G$ clears its memory buffer before transitioning to \spulse\ at time $t_\text{pulse}(u)$. In particular, we now have that $t_\text{propose}(v) \le \tau + T_0 + T_1$ and hence all nodes transition to \spulse\ by some time $t_\text{pulse}(v) < \tau + T_0 + T_1 + d$, as each $u \in G$ must have received a \spropose\ message by this time from least $n-f$ correct nodes, meeting the condition of \guard{3}.

    Let $t_0 = \min \{ t_\text{pulse}(v) : v \in G \} < \tau + T_0 + T_1 + d$ be the minimal time some correct node transitions to state \spulse. It remains to argue that $t_\text{pulse}(v) \in [t_0, t_0 + 2d)$. By \constrref{constr:bz-t2}, no correct node clears its memory buffer before time $t_0+3d$. Since some node $v \in G$ transitioned to \spulse\ at time $t_0$, we must have that its condition in \guard{3} was satisfied. That is, node $v$ must have received a $\spropose$ message from at least $n-f$ nodes since clearing its memory buffer at time $t_\text{start}(v)$, that is, $H(v,t_0) \ge n-f$ and thus $K(v,t_0)>f$. 
    
    As these messages must have been received after time $t_\text{propose}$, by then each $u\in G$ already reached state \sstart, and by \constrref{constr:bz-t2} no correct node can reset its \spropose\ flags again before time $t_0+3d$, it follows that $K(u,t_0+d)>f$ for each $u\in G$. In particular, each $u\in G$ transitions to state \spropose\ by time $t_0+d$. It now follows that at time $t_0' < t_0+2d$ we have $K(u,t_0') \ge n-f$ for all $u \in G$, implying that \guard{2} is satisfied for each such $u$. Thus, $t_\text{pulse}(u) \in [t_0, t_0'] \subseteq [t_0, t_0+2d)$ for each $u\in G$, as claimed.
\end{proof}

Let us now fix $t_0$ as given by the previous lemma. For every correct node $v \in G$, we define
\[
 p_0(v) = \inf \{ t \ge t_0 : p(v,t) = 1 \} \quad \text{ and } \quad p_{i+1}(v) = p_{\text{next}}(v,p_i(v)),
\]
where $p_{\text{next}}(v,t) = \inf\{ t' > t : p(v,t') = 1 \}$ is the next time after time $t$ node $v$ generates a pulse.

\begin{lemma}\label{lemma:bz-pulses}
For all $i \ge 0$, there exist 
\[
t_{i+1} \in [t_i + (T_2+T_3)/\vartheta, t_i + T_2 + T_3 + 3d) \quad \text{ such that } \quad p_i(v) \in [t_i, t_i + 2d) \text{ for all } v \in G.
\]
\end{lemma}
\begin{proof}
    We show the claim using induction on $i$. For the case $i=0$, the claim $p_0(v) \in [t_0, t_0+2d)$ follows directly from \lemmaref{lemma:bz-init}. For the inductive step, suppose $p_i(v) \in [t_i, t_i+2d)$ for all $v \in G$. Each $v \in G$ transitions to state \sready\ at a time $t_{\text{ready}}(v) \in [t_i + T_2/\vartheta,t_i + 2d + T_2)$ by \guard{4}. Moreover, by \constrref{constr:bz-t3} we have that $t_\text{ready}(v) > t_i + T_2/\vartheta \ge t_i + 3d$. As no correct node transitions to \spropose\ during $[t_i + 2d, t_i + (T_2+T_3)/\vartheta)$, this implies that no node receives a \spropose\ message from a correct node before the time $t_\text{propose}(u)$ when some node $u$ transitions to \spropose\ from \sready\ (for the next time after $t_i+3d$). Observe that $t_{\text{propose}}(u) > t_i + (T_2 + T_3) / \vartheta > t_i + 2d + T_2$ by \guard{5} and \constrref{constr:bz-t3}. Thus, we have $t_{\text{ready}}(v) < t_i + 2d + T_2 < t_{\text{propose}}(u)$ for all $u,v \in G$. Therefore, there exists a time $t_\text{ready} < t_i + 2d + T_2$ such that all correct nodes are in state \sready\ and $K(v,t_\text{ready}) = 0$ for all $v \in G$.

Next observe that $t_\text{propose}(v) \le t_i + 2d + T_2 + T_3$ for any $v \in G$. Hence, every $u \in G$ will receive a \spropose\ message from every $v \in G$ before time $t_\text{propose}(v) + d \le t_i + 3d + T_2 + T_3$. Thus, by \guard{3} we have that $u$ transitions to \spulse\ yielding that $p_{i+1}(v) \in [t_i + t_\text{ready}, t_i + 3d + T_2 + T_3) \subseteq [t_i + (T_2+T_3)/\vartheta, t_i + T_2 + T_3 + 3d)$. Let $t_{i+1} = \inf \{ p_{i+1}(v) : v \in G \}$. We have already established that $t_{i+1} \in [t_i + (T_2+T_3)/\vartheta, t_i + T_2 + T_3 + 3d)$. Now using the same arguments as in \lemmaref{lemma:bz-init}, it follows that for each $u\in G$, $t_\text{propose}(u) < t_{i+1} + d$, as $u$ must have received more than $f$ \spropose\ messages before time $t_{i+1} + d$ triggering the condition in \guard{5} for node $u$. Thus, \guard{3} will be satisfied before time $t_{i+1}+2d$ at each $u\in G$, implying that $p_{i+1}(u) \in [t_{i+1}, t_{i+1}+2d)$ for each~$u \in G$. 
\end{proof}

\stpulser*
\begin{proof}
The constraints in \tableref{table:st-constraints} are satisfied by setting 
\begin{align*}
    T_0 &= \vartheta(\tau+d) \\
    T_1 &= \vartheta^2(1-1/\vartheta)(\tau + d) + \tau \\
    T_2 &=\vartheta 3d \\
    T_3 &= \vartheta^2(1-1/\vartheta)3d + 2d. 
\end{align*}
By \lemmaref{lemma:bz-init} we get that there exists $t_0 \in O(\vartheta^2 d \tau)$ such that all nodes generate the first pulse during the interval $[t_0, t_0+2d)$. Applying \lemmaref{lemma:bz-pulses} we get that for all $i > 0$, we have that nodes generate the $i$th pulse during the interval $[t_{i}, t_{i} +2d)$, where $t_{i} \in [t_{i-1} + (T_2+T_3)/\vartheta, t_{i-1} + T_2 + T_3 + 3d) \subseteq [\Phi^-, \Phi^+)$. Note that $T_2+T_3 \in \Theta(\vartheta^2 d)$ and $\vartheta, d \in O(1)$. These observations give the first four properties. For the final property, observe that nodes only need to communicate when they transition to \spropose. By \guard{4}, correct nodes wait at least for $T_2/\vartheta = 3d$ reference time before transitioning to \spropose\ again after generating a pulse. Hence, nodes need to broadcast at most one bit every $3d>d$ time.
\end{proof}

% --- RESYNC TO PULSES ---

\section{Self-stabilising pulse synchronisation}\label{sec:pulse-gen}

In this section, we show how to use a resynchronisation algorithm and a synchronous consensus routine to devise \emph{self-stabilising} pulse synchronisation algorithms. We obtain the following result:

\resynctopulse*

\subsection{Overview of the ingredients}

The pulse synchronisation algorithm presented in this section consists of two state machines running in parallel:
\begin{enumerate}[noitemsep]
    \item the so-called \emph{main state machine} that is responsible for pulse generation, and
    \item an \emph{auxiliary state machine}, which assists in initiating consensus instances and stabilisation.
\end{enumerate}

The main state machine indicates when pulses are generated and handles all the communication between nodes except for messages sent by simulated consensus instances. The latter are handled by the auxiliary state machine. The transitions in the main state machine are governed by a series of threshold votes, local timeouts, and signals from the auxiliary state machine. 

As we aim to devise self-stabilising algorithms, the main and auxiliary state machines may be arbitrarily initialised. To handle this, a stabilisation mechanism is used in conjunction to ensure that, regardless of the initial state of the system, all nodes eventually manage to synchronise their state machines. The stabilisation mechanism relies on the following three subroutines which are summarised in \figureref{fig:overview_pulse}:
\begin{enumerate}[noitemsep,label=(\alph*)]
 \item a resynchronisation algorithm $\vec B$, 
 \item the non-self-stabilising pulse synchronisation algorithm $\vec P$ from \sectionref{sec:bz-pulse-sync},
 \item a synchronous consensus algorithm $\vec C$.
\end{enumerate}

%###
\paragraph*{Resynchronisation pulses.}
%###
Recall that \emph{resynchronisation algorithm} $\vec B$ solves a weak variant of a pulse synchronisation: it guarantees that \emph{eventually}, within some bounded time $T(\vec B)$, all correct nodes generate a \emph{good} resynchronisation pulse such that no new resynchronisation pulse is generated before $\Psi$ time has passed. Note that at all other times, the algorithm $\vec B$ is allowed to generate pulses at arbitrary frequencies and not necessarily at all correct nodes. Nevertheless, at some point all correct nodes are bound to generate a good resynchronisation pulse in rough synchrony. We leverage this property to cleanly re-initialise the stabilisation mechanism from time to time. Observe that the idea is somewhat similar to the use of initialisation signals in \sectionref{sec:bz-pulse-sync}, but now we have less control over the incoming ``initialisation signals'' (i.e.\ resynchronisation pulses). To summarise, we assume throughout this section that every correct node $v \in G$
\begin{itemize}[noitemsep]
    \item receives a single resynchronisation pulse at time $t_v \in [0, \rho)$, and 
    \item does not receive another resynchronisation pulse before time $t_v + \Psi$, 
\end{itemize}
where $\Psi$ is sufficiently large value we determine later. Later in \sectionref{sec:resync}, we devise efficient algorithms that produce the needed resynchronisation pulses.

%###
\paragraph*{Simulating synchronous consensus.}
%###
The two subroutines (b) and (c) are used in conjunction as follows. We use the variant of the Srikanth--Toueg pulse synchronisation algorithm $\vec P$ described in \sectionref{sec:bz-pulse-sync} to simulate a synchronous consensus algorithm (in the bounded-delay model). Note that while this pulse synchronisation algorithm is not self-stabilising, it works properly even if non-faulty nodes initialise the algorithm at different times as long as they do so within a time interval of length~$\tau$. 

Assuming that the nodes initialise the non-self-stabilising pulse synchronisation algorithm $\vec P$ within at most time $\tau$ apart, it is straightforward to simulate round-based (i.e., synchronous) algorithms: a pulse generated by $\vec P$ indicates that a new round of the synchronous algorithm can be started. By setting the delay between two pulses large enough, we can ensure that 
\begin{enumerate}[noitemsep]
    \item all nodes have time to execute the local computations of the synchronous algorithm and
    \item all messages related to a single round arrive before a new pulse occurs.
\end{enumerate}

%###
\paragraph*{Employing silent consensus.}
%###
We utilise so-called silent consensus routines in our construction. Silent consensus routines satisfy exactly the same properties as usual consensus routines (validity, agreement, and termination) with the addition that correct nodes send no messages in executions in which all nodes have input~0.

\begin{definition}[Silent consensus]\label{def:silent}
A consensus routine is \emph{silent}, if in each execution in which all correct nodes have input $0$, correct nodes send no messages.
\end{definition}

Any synchronous consensus routine can be converted into a silent consensus routine essentially for free. In our prior work~\cite{lenzen16firing}, we showed that there exists a simple transformation that induces only an overhead of two rounds while keeping all other properties of the algorithm the same.

\begin{theorem}[\cite{lenzen16firing}]\label{thm:silent-consensus}
Any consensus protocol $\vec C$ that runs in $R$ rounds can be transformed into a silent consensus protocol $\vec C'$ that runs in $R+2$ rounds. Moreover, the resilience and message size of $\vec C$ and $\vec C'$ are the same.
\end{theorem}

Thus, without loss of generality, we assume throughout this section that the given consensus routine $\vec C$ is silent. Moreover, this does not introduce any asymptotic loss in the running time or number of bits communicated.

\begin{figure}[t!]
\begin{center}
 \includegraphics[page=6,scale=0.99]{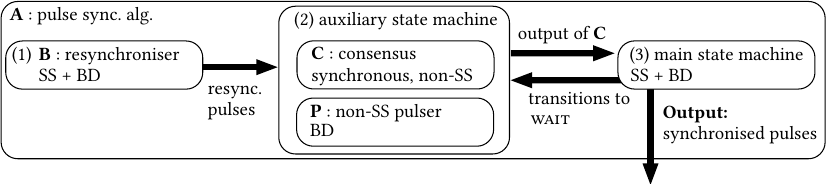}
\end{center}
    \caption{Constructing a self-stabilising (SS) and Byzantine fault-tolerant pulse synchronisation algorithm $\vec A$ in the bounded-delay model (BD) out of Byzantine fault-tolerant non-stabilising pulse synchronisation algorithm $\vec P$, synchronous consensus algorithm~$\vec C$, and resynchronisation algorithm $\vec B$. All algorithms run on the same set of nodes. (1) The resynchronisation algorithm~$\vec B$ eventually outputs a good resynchronisation pulse, which resets the stabilisation mechanism used by the auxiliary state machine. (2) The auxiliary state machine simulates the executions of~$\vec C$ using~$\vec P$. Simulations are initiated either due to nodes transitioning to a special $\swait$ state of the main state machine (see \figureref{fig:statemachine}) or a certain time after a resynchronisation pulse. (3) The main state machine. It generates pulses when a consensus instance outputs ``1'' and, when stabilised, guarantees re-initialisation of the consensus algorithm by the auxiliary state machine.\label{fig:overview_pulse}}
\end{figure}

\subsection{High-level idea of the construction}

The high-level strategy used in our construction is as follows. We run the resynchronisation algorithm in parallel to the self-stabilising pulse synchronisation algorithm we devise in this section. The resynchronisation algorithm will send the resynchronisation signals it generates to the pulse synchronisation algorithm as shown in \figureref{fig:overview_pulse}.

The pulse synchronisation algorithm consists of the main state machine given in \figureref{fig:statemachine} and the auxiliary state machine given in \figureref{fig:auxmachine}. The auxiliary state machine is responsible for generating the output signals that drive the main state machine (\guard{2} and \guard{2'}). 

\begin{figure}
\begin{center}
 \includegraphics[page=10,scale=0.95]{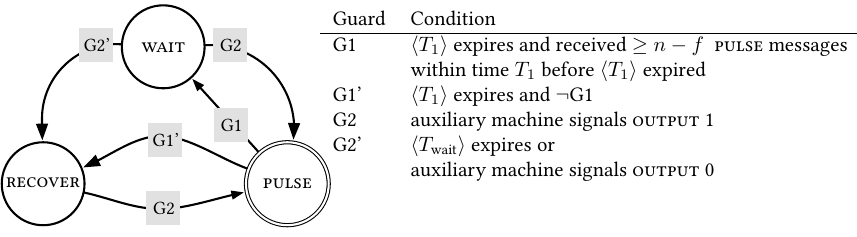}
\end{center}
    \caption{The main state machine. When a node transitions to state $\spulse$ (double circle) it will generate a local pulse event and send a $\spulse$ message to all nodes. When the node transitions to state $\swait$ it broadcasts a $\swait$ message to all nodes. \guard{1} employs a sliding window memory buffer, which stores any $\spulse$ messages that have arrived within time $T_1$ (as measured by the local clock). When a correct node transitions to $\spulse$ it resets a local $T_1$ timeout. Once this expires, either \guard{1} or \guard{1'} become satisfied. Similarly, the timer $\twait$ is reset when the node transitions to $\swait$. Once it expires, \guard{2'} is satisfied and the node transitions from $\swait$ to $\srecover$. The node can transition to $\spulse$ state when \guard{2} is satisfied, which requires an \soutputone\ signal from the auxiliary state machine given in \figureref{fig:auxmachine}.\label{fig:statemachine}} 
\end{figure}

The auxiliary state machine employs a consensus routine to facilitate agreement among the nodes on whether a new pulse should occur. If the consensus simulation outputs 1 at some node, then the auxiliary state machine signals the main state machine to generate a pulse. Otherwise, if the consensus instance outputs 0, then this is used to signal that something is wrong and the node can detect that the system has not stabilised. We carefully set up our construction so that once the system stabilises, any consensus instance run by the nodes is guaranteed to always output 1 at every correct node.

\begin{figure}[t!]
\begin{center}
 \includegraphics[page=11,scale=0.99]{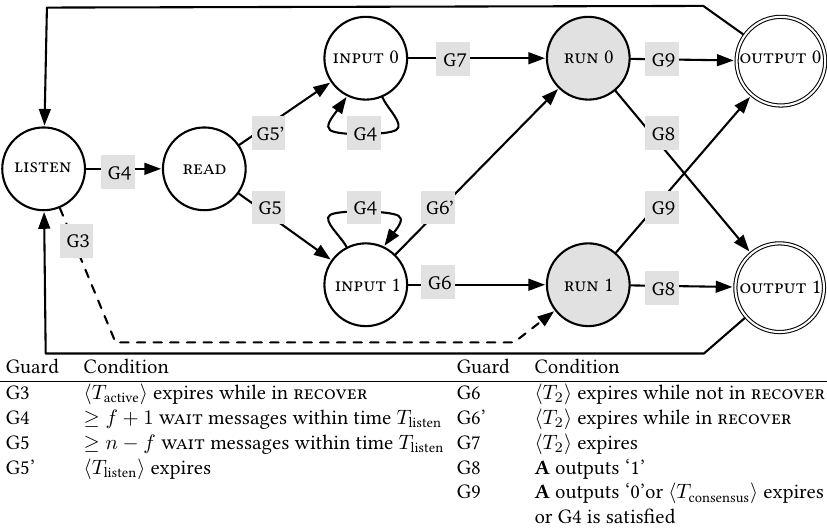}
\end{center}
    \caption{The auxiliary state machine. The auxiliary state machine is responsible for initialising and simulating the consensus routine. The gray boxes denote states which represent the simulation of the consensus routine $\vec C$. When transitioning to either $\sexeczero$ or $\sexecone$, the node locally initialises the (non-self-stabilising) pulse synchronisation algorithm from \sectionref{sec:bz-pulse-sync} and a new instance of $\vec C$. If the node transitions to $\sexeczero$, it uses input 0 for the consensus routine. If the node transitions to $\sexecone$, it uses input 1. When the consensus simulation declares an output, the node transitions to either $\soutputzero$ or $\soutputone$ (sending the respective output signal to the main state machine) and immediately to state $\slisten$. The timeouts $\tlisten$, $T_2$, and $\tconsensus$ are reset when a node transitions to the respective states that use a guard referring to them. The timeout $\tactive$ in \guard{3} (dashed line) is reset by the resynchronisation signal from the underlying resynchronisation algorithm $\vec B$. Both $\sinputzero$ and $\sinputone$ have a self-loop that is activated if \guard{4} is satisfied. This means that if \guard{4} is satisfied while in these states, the timer $T_2$ is reset.\label{fig:auxmachine}} 
\end{figure}

As we operate under the assumption that the initial state is arbitrary, the non-trivial part in our construction is to get all correct nodes synchronised well enough to even start simulating consensus jointly in the first place. This is where the resynchronisation algorithm comes into play. We make sure that the algorithm either stabilises or all nodes get ``stuck'' in a recovery state $\srecover$. To deal with the latter case, we use the resynchronisation pulse to let all nodes synchronously reset a local timeout. Once this timeout expires, nodes that are in state $\srecover$ start a consensus instance with input ``1''. By the time this happens, either
\begin{itemize}
    \item the algorithm has already stabilised (and thus no correct node is in state $\srecover$), or 
    \item all correct nodes are in state $\srecover$ and jointly start a consensus instance that will output ``1'' (by validity of the consensus routine).
\end{itemize} 
In both cases, stabilisation is guaranteed.

%###
\paragraph*{Receiving a resynchronisation signal.}
%###
The use of the resynchronisation signal is straightforward: when a correct node $u \in G$ receives a resynchronisation signal from the underlying resynchronisation algorithm $\vec B$, node $u$ resets its local timeout $\tactive$ used by the auxiliary state machine in \figureref{fig:auxmachine}. Upon expiration of the timeout, \guard{3} in the auxiliary state machine is activated only if the node is in state $\srecover$ at the time. 

%###
\paragraph*{Main state machine.}
%###

The main state machine, which is given in \figureref{fig:statemachine}, is responsible for generating the pulse events and operates as follows. If a node is in state $\spulse$, it generates a local pulse event and sends a $\spulse$ message to all other nodes. Now suppose a node $u \in G$ transitions to state $\spulse$. Two things can happen:
\begin{itemize}
    \item If a node $u \in G$ is in state $\spulse$ and observes at least $n-f$ nodes also generating a pulse within a short enough time window (\guard{1}), it is possible that all correct nodes generated a pulse in a synchronised fashion. If this happens, then \guard{1} ensures that node~$u$ proceeds to the state $\swait$. As the name suggests, the $\swait$ state is used to wait before generating a new pulse, ensuring that pulses obey the desired frequency bounds. 
    \item Otherwise, if a node is certain that not all correct nodes are synchronised, it transitions from $\spulse$ to state $\srecover$ (\guard{1'}).
\end{itemize}
Once a node is in either $\swait$ or $\srecover$, it will not leave the state before the consensus algorithm outputs ``1'', as \guard{2} needs to be satisfied in order for a transition to $\spulse$ to take place. The simulation of consensus is handled by the auxiliary state machine, which we discuss below. The nodes use consensus to agree whether sufficiently many nodes transitioned to the $\swait$ state within a small enough time window. If the system has stabilised, all correct nodes transition to $\swait$ almost synchronously, and hence, after stabilisation every correct node always uses input ``1'' for the consensus instance.

Once a node transitions to state $\swait$, the node keeps track of how long it has been there. If the node observes that it has been there longer than it would take for a consensus simulation to complete under correct operation (indicating that the system has not yet stabilised), it transitions to state $\srecover$. Also, if the consensus instance outputs ``0'', the node knows something is wrong and transitions to $\srecover$. During the stabilisation phase, nodes that transition to $\srecover$ refrain from using input ``1'' for any consensus routine before the local timeout $\tactive$ expires; we refer to the discussion of the auxiliary state machine.

Once the system stabilises, the behaviour of the main state machine is simple, as only \guard{1} and \guard{2} can be satisfied. This implies that correct nodes alternate between the $\spulse$ and $\swait$ states. Under stabilised operation, we get that all correct nodes:
\begin{itemize}
    \item transition to $\spulse$ within a time window of length $2d$,
    \item observe that at least $n-f$ nodes transitioned to $\spulse$ within a short enough time window ensuring that \guard{1} is satisfied at every correct node,
    \item transition to $\swait$ within a time window of length $O(d)$,
    \item correctly initialise a simulation of the consensus algorithm $\vec C$ with input ``1'', as correct nodes transitioned to $\swait$ in a synchronised fashion (see auxiliary state machine),
    \item all correct nodes remain in $\swait$ until \guard{2} or \guard{2'} become satisfied.
\end{itemize}

Finally, we ensure that (after stabilisation) all correct nodes remain in state $\swait$ in the main state machine longer than it takes to properly initialise and simulate a consensus instance. This is achieved by using the $\twait$ timeout in \guard{2'}. Due to the validity property of the consensus routine and the fact that all correct nodes use input $1$, this entails that \guard{2} is always satisfied before \guard{2'}, such that all correct nodes again transition to $\spulse$ within a time window of length $2d$.

%###
\paragraph*{Auxiliary state machine.}
%###

The auxiliary state machine given in \figureref{fig:auxmachine} is slightly more involved. However, the basic idea is simple: 
\begin{enumerate}[label=(\alph*)]
    \item nodes try to check whether at least $n-f$ nodes transition to the $\swait$ state \emph{in a short enough time window} (that is, a time window consistent with correct operation) and
    \item then use a consensus routine to agree whether all nodes saw this. 
\end{enumerate}
Assuming that all correct nodes participate in the simulation of consensus, we get the following:
\begin{itemize}
    \item If the consensus algorithm $\vec C$ outputs ``0'', then some correct node did not see $n-f$ nodes transitioning to $\swait$ in a short time window, and hence, the system has not yet stabilised.
    \item If the consensus algorithm $\vec C$ outputs ``1'', then all correct nodes agree that a transition to $\swait$ happened recently.
\end{itemize}
In particular, the idea is that when the system operates correctly, the consensus simulation will always succeed and output ``1'' at every correct node.

The above idea is implemented in the auxiliary state machine as follows. Suppose that a correct node $u \in G$ is in the $\slisten$ state and the local timeout $\tactive$ is not about to expire (recall that $\tactive$ is only reset by the resynchronisation signal). Node $u$ remains in this state until it is certain that at least one correct node transitions to $\swait$ in the main state machine. Once this happens, \guard{4} is satisfied and node $u$ transitions to the $\sread$ state. In the $\sread$ state, node $u$ waits for a while to see whether it observes (1) at least $n-f$ nodes transitioning to $\swait$ in a short time window or (2) less than $n-f$ nodes doing this.

In case (1), node $u$ can be certain that at least $n-2f>f$ correct nodes transitioned to $\swait$. Thus, node $u$ can also be certain that every correct node observes at least $f+1$ correct nodes transitioning to $\swait$; this will be a key property later. In case (2), node~$u$ can be certain that the system has not stabilised. If case (1) happens, we have that \guard{5} is eventually satisfied. Node $u$ then transitions to $\sinputone$ indicating that node $u$ is willing to use input ``1'' in the next simulation of consensus \emph{unless} it is in the \srecover\ state in the main state machine. In case (2), we get that \guard{5'} becomes satisfied and $u$ transitions to $\sinputzero$. This means that $u$ insists on using input ``0'' for the next consensus simulation.

Once node $u \in G$ transitions to either $\sinputzero$ or $\sinputone$, it will remain there until the local timeout of length $T_2$ expires (see \guard{6}, \guard{6'} and \guard{7}). However, if \guard{4} becomes satisfied \emph{while} node $u$ is in either of the input states, then the local timeout is reset again. We do this because if \guard{4} becomes satisfied while $u$ is in one of the input states, (i) the same may be true for other correct nodes that are in state $\slisten$ and (ii) node $u$ can be certain that the system has not stabilised. Resetting the timeout helps in ensuring that all correct nodes jointly start the next consensus instance (guaranteeing correct simulation), if \guard{4} is satisfied at all correct nodes at roughly the same time. In case this does not happen, resetting the timeout at least makes sure that there will be a time when \emph{no} correct node is currently trying to simulate a consensus instance. These properties are critical for our proof of stabilisation.

\subsection{Outline of the proof}\label{ssec:outline}

The key difficulty in achieving stabilisation is to ensure the proper simulation of a consensus routine despite the arbitrary initial state. In particular, after the transient faults cease, we might have some nodes attempting to execute consensus, whereas some do not. Moreover, nodes that are simulating consensus might be simulating different rounds of the consensus routine, and so on. To show that such disarray cannot last indefinitely long, we use the following arguments:
\begin{itemize}
    \item if some correct node attempts to use input ``1'' for consensus, then at least $f+1$ correct nodes have transitioned to $\swait$ in the main state machine (\lemmaref{lemma:input-wait}), that is, all correct nodes see if some other correct node might be initialising a new consensus instance with input ``1'' soon,
 \item if some correct node transitions to $\swait$ at time $t$, then there is a long interval of length $\Theta(T_2)$ during which no correct node transitions to $\swait$ (\lemmaref{lemma:separation}), that is, correct nodes cannot transition to $\swait$ state too often,
 \item if some correct node attempts to use input ``1'' for consensus, then all correct nodes initialise a new consensus instance within a time window of length $\tau \in \Theta((1-1/\vartheta)T_2)$ (\lemmaref{lemma:consistent-init}),
 \item if all correct nodes initialise a new consensus instance within a time window of length $\tau$, then all correct nodes participate in the same consensus instance and succesfully simulate an entire execution of $\vec C$ (\lemmaref{lemma:consensus-simulation}).
\end{itemize}
The idea is that the timeout $T_2$ will be sufficiently large to ensure that consensus instances are well-separated: if a consensus instance is initialised with input ``1'' at some correct node, then there is enough time to properly simulate a complete execution of the consensus routine before any correct node attempts to start a new instance of consensus.

Once we have established the above properties, it is easy to see that if synchronisation is established, then it \emph{persists}. More specifically, we argue that if all correct nodes transition to $\spulse$ at most time $2d$ apart, then all correct nodes initialise a new consensus instance within a time window of length $\tau$ using input ``1'' (\lemmaref{lemma:synchronisation}). Thus, the system stabilises if all correct nodes eventually generate a pulse with skew at most $2d$.

Accordingly, a substantial part of the proof is arguing that all nodes eventually transition to $\spulse$ within time window of $2d$. To see that this is bound to occur eventually, we consider an interval $[\alpha,\beta]$ of length $\Theta(R)$ and use the following line of reasoning:
\begin{itemize}
    \item if all correct nodes are simultaneously in state $\srecover$ at some time before timeout $\tactive$ expires at any correct node, then \guard{3} in the auxiliary state machine becomes satisfied at all correct nodes and a new consensus instance with all-1 input is initialised within a time window of length $\tau$ (\lemmaref{lemma:passive-stab}),
    \item if some correct node attempts to use input ``1'' during the interval $[\alpha,\beta]$, then either (a) all correct nodes end up in $\srecover$ before timeout $\tactive$ expires at any node or (b) all correct nodes eventually transition to $\spulse$ within time $2d$ (\lemmaref{lemma:input-1-cases}),
    \item if no correct node attempts to use input ``1'' during the time interval $[\alpha, \beta]$, all correct nodes will be in state $\srecover$ before the timeout $\tactive$ expires at any node (\lemmaref{lemma:no-pulse}).
\end{itemize}
In either of the latter two cases, we can use the first argument to guarantee stabilisation (\corollaryref{coro:stab-by-one} and \corollaryref{coro:stab-by-no-input}). Finally, we need to argue that all the timeouts employed in the construction can be set so that our arguments work out. The constraints related to all the timeouts are summarised in \tableref{table:constraints} and \lemmaref{lemma:constraints} shows that these can be satisfied. We now proceed to formalise and prove the above arguments in detail. The structure of the proof is summarised in \figureref{fig:proof-structure}. 

\begin{figure}[t!]
\begin{center}
 \includegraphics[page=12,scale=0.75]{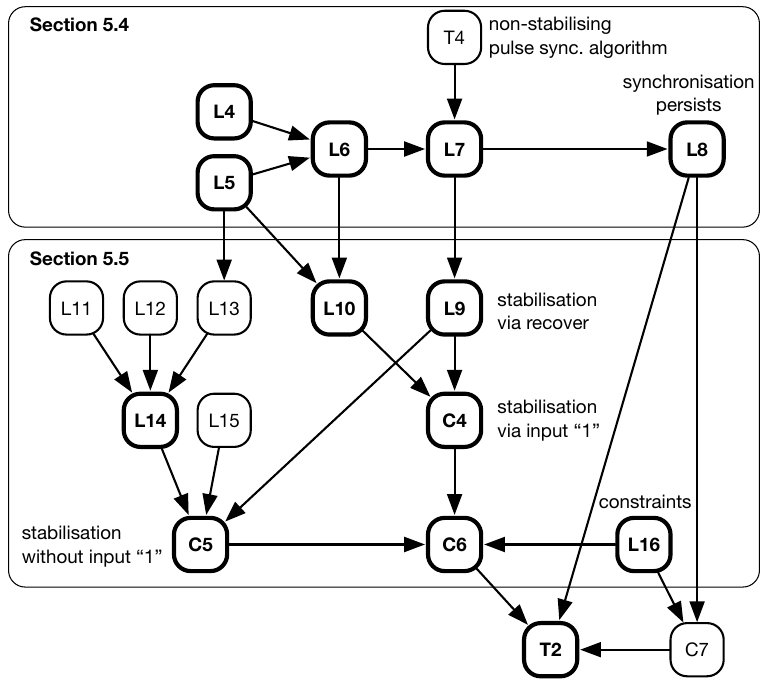}
\end{center}
    \caption{The overall structure of the proof of \theoremref{thm:resync-to-pulse}. The bold rectangles denote results that are informally discussed in \sectionref{ssec:outline}.\label{fig:proof-structure}} 
\end{figure}

\subsection{Analysing the state machines}\label{ssec:analysis}

\begin{table}
\caption{The timeout conditions employed in the construction of \sectionref{sec:pulse-gen}. \label{table:constraints}}
\center
\begin{tabular}{l l}
 \toprule
  \constrnumber\label{constr:d-theta} & $d, \vartheta \in O(1)$ \\ 
 \midrule
  \constrnumber\label{constr:t1} & $T_1 = 3 \vartheta d$ \\ 
  \constrnumber\label{constr:tlisten} & $\tlisten= (\vartheta-1)T_1+3\vartheta d$\\
  \constrnumber\label{constr:t2} & $T_2 > \vartheta(\tlisten + 3T_1 + 3d)$\\
  \constrnumber\label{constr:t2thetabound} & $(2/\vartheta-1)T_2 > 2\tlisten + \tconsensus + 5T_1 + 4d$\\

  \constrnumber\label{constr:tau} & $\tau = \max\left\{(1-1/\vartheta)T_2+\tlisten+d+\max\{\tlisten+d,3T_1+2d\},(1-1/\vartheta)\tactive+\rho\right\}$\\
    \constrnumber\label{constr:consensus-time} & $\tconsensus = \vartheta(\tau + T(R))$ \\
    \constrnumber\label{constr:twait} & $\twait = T_2 + \tconsensus$\\
    \constrnumber\label{constr:tactive} & $\tactive \geq 4T_2 + \tlisten + \vartheta(\tlisten+\twait-5T_1-4d + \rho)$ \\
    \constrnumber\label{constr:beta} & $\tactive\geq 2T_2+\tconsensus+\vartheta(2\tlisten + T_1+\twait +3d+2T_2+2\tconsensus)$\\
   
 \bottomrule
\end{tabular}
\end{table}

First let us observe that after a short time, the arbitrary initial contents of the sliding window message buffers have been cleared. 

\begin{remark}\label{remark:cleanup}
    By time $t = \max \{ T_1, \tlisten \} + d \in O(\vartheta^2 d)$ the sliding window memory buffers used in \guard{2}, \guard{4}, and \guard{5} for each $u \in G$ have valid contents: if the buffer of \guard{2} contains a message $m$ from $v \in G$ at any time $t' \ge t$, then $v \in G$ sent the message $m$ during $(t-T_1-d, t)$; similarly, for \guard{4} and \guard{5} this holds for the interval $(t-\tlisten-d,t)$.
\end{remark}

Without loss of generality, we assume that this has happened by time $0$. Moreover, we assume that every correct node received the resynchronisation signal during the time interval $[0, \rho)$. Thus, the contents of message buffers are valid from time 0 on and every node has reset its $\tactive$ timeout during $[0,\rho)$. Hence, the timeout $\tactive$ expires at any node $u \in G$ during $[\tactive/\vartheta, \tactive+\rho)$.

We use $T(R) \in O(\vartheta^2 d R)$ to denote the maximum time a simulation of the $R$-round consensus routine $\vec C$ takes when employing the non-stabilising pulse synchronisation algorithm given in \sectionref{sec:bz-pulse-sync}. We assume that the consensus routine $\vec C$ is silent, as by \theoremref{thm:silent-consensus} we can convert any consensus routine into a silent one without any asymptotic loss in the running time.

First, we highlight some useful properties of the simulation scheme implemented in the auxiliary state machine.

\begin{remark}\label{remark:st-properties}
If node $v \in G$ transitions to $\sexeczero$ or $\sexecone$ at time $t$, then the following holds:
\begin{itemize}[noitemsep]
    \item node $v$ remains in the respective state during $[t,t+\tau)$, 
    \item node $v$ does not execute the first round of $\vec C$ before time $t+\tau$,
    \item node $v$ leaves the respective state and halts the simulation of $\vec C$ before time $t+\tconsensus$.
\end{itemize}
\end{remark}

Now let us start by showing that if some node transitions to $\sone$ in the auxiliary state machine, then there is a group of at least $f+1$ correct nodes that transition to $\swait$ in the main state machine in rough synchrony.

\begin{lemma}\label{lemma:input-wait}
    Suppose node $v\in G$ transitions to $\sone$ at time $t \geq \tlisten+d$. Then there is a time $t'\in (t-\tlisten-d,t)$ and a set $A\subseteq G$ with $|A|\geq f+1$ such that each $w \in A$ transitions to $\swait$ during $[t',t]$.
\end{lemma}
\begin{proof}
    Since $v$ transitions to $\sone$, it must have observed at least $n-f$ distinct $\swait$ messages within time $\tlisten$ in order to satisfy \guard{5}. As $f<n/3$, we have that at least $f+1$ of these messages came from nodes $A \subseteq G$, where $|A| \ge f+1$. The claim follows by choosing $t'$ to be the minimal time during $(t-\tlisten-d, t)$ at which some $w \in A$ transitioned to $\swait$. 
\end{proof}

Next we show that if some correct node transitions to $\swait$, then this is soon followed by a long time interval during which no correct node transitions to $\swait$. Thus, transitions to $\swait$ are well-separated.

\begin{lemma}\label{lemma:separation}
    Suppose node $v\in G$ transitions to $\swait$ at time $t\leq (\tactive-T_2)/\vartheta$. Then no $u \in G$ transitions to $\swait$ during $[t+3T_1+d,t+T_2/\vartheta-2T_1-d)$.
\end{lemma}
\begin{proof}
    Since $v \in G$ transitioned to $\swait$ at time $t$, it must have seen at least $n-f$ nodes transitioning to $\spulse$ during the time interval $(t-2T_1,t)$. Since $f < n/3$, it follows that $n-2f \ge f+1$ of these messages are from correct nodes. Let us denote this set of nodes by $A \subseteq G$.

    Consider any node $w \in A$. As node $w$ transitioned to $\spulse$ during $(t-2T_1-d,t)$, it transitions to state $\srecover$ or to state $\swait$ at time $t_w \in (t-2T_1-d, t+T_1)$. Either way, as it also transitioned to $\slisten$, transitioning to $\spulse$ again requires to satisfy \guard{3} or one of \guard{6}, \guard{6'}, and \guard{7} while being in states $\sinputzero$ or $\sinputone$, respectively. By assumption, \guard{3} is not satisfied before time $\tactive/\vartheta \ge t+T_2/\vartheta$, and the other options require a timeout of $T_2$ to expire, which takes at least time $T_2/\vartheta$. It follows that $w$ is not in state $\spulse$ during $[t+T_1,t+T_2/\vartheta-2T_1-d)$.
    
    We conclude that no $w\in A$ is observed transitioning to $\spulse$ during $[t + T_1 + d, t + T_2/\vartheta - 2T_1 - d)$. Since $|A| > f$, we get that no $u \in G$ can activate \guard{1} and transition to $\swait$ during $[t + 3T_1 + d, t + T_2/\vartheta - 2T_1 - d)$, as $n-|A| < n-f$.
\end{proof}

Using the previous lemmas, we can show that if some correct node transitions to state $\sinputone$ in the auxiliary state machine, then every correct node eventually initialises and participates in the same new consensus instance. That is, every correct node initialises the underlying Srikanth--Toueg pulse synchronisation algorithm within a time interval of length~$\tau$.

\begin{lemma}\label{lemma:consistent-init}
    Suppose node $u \in G$ transitions to $\sone$ at time $t \in [\tlisten + d,(\tactive-T_2)/\vartheta]$. Then each $v \in G$ transitions to state $\sexeczero$ or state $\sexecone$ at time $t_v \in [t_0, t_0 + \tau)$, where $t_0 = t - \tlisten - d + T_2/\vartheta$. Moreover, \guard{4} cannot be satisfied at any node $v \in G$ during $[t+3T_1+2d, t^*+T_1/\vartheta)$, where $t^*:=\min_{v\in G}\{p(v,t+d)\}$.
\end{lemma}
\begin{proof}
    By \lemmaref{lemma:input-wait} there exists a set $A \subseteq G$ such that $|A| \ge f+1$ and each $w \in A$ transitions to $\swait$ at a time $t_w \in (t-\tlisten-d,t)$. This implies that \guard{4}, and thus also \guard{9}, becomes satisfied for $v \in G$ at time $t'_v \in [t-\tlisten-d,t+d)$. Thus, every $v \in G$ transitions to state $\sread$, $\szero$, or $\sone$ at time $t'_v$; note that if $v$ was in state $\szero$ or $\sone$ before this happened, it transitions back to the same state due to \guard{4} being activated and resets its local $T_2$ timer. Moreover, by time $t_v'\leq r_v < t'_v +\tlisten < t + d + \tlisten$ node $v$ transitions to either $\szero$ or $\sone$, as either \guard{5} or \guard{5'} becomes activated in case $v$ transitions to state $\sread$ at time $t_v'$.

    Now we have that node $v$ remains in either $\sone$ or $\szero$ for the interval $[r_v, r_v + T_2/\vartheta)$, as none of \guard{6}, \guard{6'}, and \guard{7} are satisfied before the local timer $T_2$ expires. Moreover, by applying \lemmaref{lemma:separation} to any $w \in A$, we get that no $v \in G$ transitions to $\swait$ during the interval 
\[
    [t_w+3T_1+d, t_w + T_2/\vartheta-2T_1-d) \supseteq (t + 3T_1 + d, t + T_2/\vartheta - 2T_1 - \tlisten - 2d).
\]
    Recall that for each $v\in G$, $t_v'<t+d$. After this time, $v$ cannot transition to $\spulse$ again without transitioning to $\sexeczero$ or $\sexecone$ first. Since $t + T_2/\vartheta - 2T_1 - \tlisten - 2d > t + T_1 + d$ by \constrref{constr:t2}, we get that every $w \in G$ has arrived in state $\swait$ or $\srecover$ by time $t + T_2/\vartheta - 2T_1 - \tlisten - 2d$. Thus, no such node transitions to state $\swait$ during $[t + T_2/\vartheta - 2T_1 - \tlisten - 2d,t^*+T_1/\vartheta)$: first, it must transition to $\spulse$, which requires to satisfy \guard{2}, i.e., transitioning to state $\soutputone$, and then a timeout of $T_1$ must expire; here, we use that we already observed that $t + T_2/\vartheta - 2T_1 - \tlisten - 2d > t + d$, i.e., by definition the first node $w\in G$ to transition to $\spulse$ after time $t + T_2/\vartheta - T_1 - \tlisten - 2d$ does so at time $t^*$. We conclude that \guard{4} cannot be satisfied at any $v\in G$ during the interval $[t + 3T_1 + 2d, t^*+T_1/\vartheta)$, i.e., the second claim of the lemma holds.
    
    We proceed to showing that each $v \in G$ transitions to state $\sexeczero$ or state $\sexecone$ at time $t_v \in [t_0, t_0 + \tau)$, i.e., the first claim of the lemma. To this end, observe that $w$ transitions to either state $\sexeczero$ or $\sexecone$ at some time $t'\in (t_w',t^*)$. By the above observations, $t' \ge r_w + T_2/\vartheta \ge t - \tlisten - d + T_2/\vartheta = t_0$. Node $w$ initialises the Srikanth-Toueg algorithm given in \figureref{fig:bz-pulser} locally at time $t'$. In particular, by the properties of the simulation algorithm given in \remarkref{remark:st-properties}, we have that $w$ waits until time $t'+\tau$ before starting the simulation of $\vec C$, and hence, $w$ remains in $\sexeczero$ or $\sexecone$ at least until time $t'+\tau$ before the simulation of $\vec C$ produces an output value. Thus, we get that $t^* \ge t'+\tau \ge t_0 + \tau$.
    
     Recall that each $v\in G$ resets timeout $T_2$ at time $r_v\in [t_v',t+d+\tlisten)\subseteq [t-\tlisten-d,t+\tlisten+d)$ and does not reset it during $[t + 3T_1 + 2d, t^*+T_1/\vartheta)$, as it does not satisfy \guard{4} at any time from this interval. When $T_2$ expires at $v$, it transitions to $\sexeczero$ or $\sexecone$. Because $t^*+T_1/\vartheta> t_0+\tau\geq t+\max\{\tlisten+d,3T_1+2d\}+T_2$ by \constrref{constr:tau}, this happens at time $t_v\in [t-\tlisten-d+T_2/\vartheta,t_0+\tau]= [t_0,t_0+\tau]$, as claimed.
\end{proof}

Next, we show that if all correct nodes initialise a new instance of the Srikanth--Toueg pulse synchronisation algorithm within a time interval of length $\tau$, then every correct node initialises, participates in, and successfully completes simulation of the consensus routine $\vec C$.

\begin{lemma}\label{lemma:consensus-simulation}
    Suppose there exists a time $t_0$ such that each node $v \in G$ transitions to $\sexeczero$ or $\sexecone$ at some time $t_v \in [t_0, t_0 + \tau)$. Let $t'$ be the minimal time larger than $t_0$ at which some $u \in G$ transitions to either $\soutputzero$ or $\soutputone$. If \guard{4} is not satisfied at any $v \in G$ during $[t_0, t'+2d)$, then $t'\leq t_0+\tconsensus/\vartheta-2d$ and there are times $t'_v \in [t',t'+2d)$, $v\in G$, such that:
   \begin{itemize}[noitemsep]
       \item each $v\in G$ transitions to $\soutputone$ or $\soutputzero$ at time $t_v'$ (termination),
       \item this state is the same for each $v\in G$ (agreement), and
       \item if each $v\in G$ transitioned to state $\sexecone$ at time $t_v$, then this state is $\soutputone$ (validity).
    \end{itemize}
\end{lemma}
\begin{proof}
    When $v \in G$ transitions to either $\sexeczero$ or $\sexecone$, it sends an initialisation signal to the non-self-stabilising Srikanth--Toueg algorithm described in \sectionref{sec:bz-pulse-sync}. Using the clock given by this algorithm, the nodes simulate the consensus algorithm $\vec C$. If node $v \in G$ enters state $\sexeczero$, it uses input ``0'' for $\vec C$. Otherwise, if $v$ enters $\sexecone$ it uses input ``1''.
    
    Note that \theoremref{thm:st-pulser} implies that if all nodes initialise the Srikanth--Toueg algorithm within time $\tau$ apart, then the simulation of $\vec C$ takes at most $\tau + T(R) \in O(\vartheta^2 d (\tau + R))$ time. Moreover, all nodes will declare the output in the same round, and hence, declare the output within a time window of $2d$, as the skew of the pulses is at most $2d$.

   Now let us consider the simulation taking place in the auxiliary state machine. If \guard{4} is not satisfied during $[t_0, t'+2d)$ and the timer $\tconsensus$ does not expire at any node $v \in G$ during the simulation, then by time $t' +2d \le t_0 + \tau + T(R) \in O(\vartheta^2 d (\tau + R))$ the nodes have simulated $R$ rounds of $\vec C$ and declared an output. By assumption, \guard{4} cannot be satisfied prior to time $t'+2d$. At node $v\in G$, the timer $\tconsensus$ is reset at time $t_v\geq t_0$. Hence, it cannot expire again earlier than time $t_0 + \tconsensus/\vartheta \ge t_0 + \tau + T(R)$ by \constrref{constr:consensus-time}. Hence, the simulation succeeds.

    Since the simulation of the consensus routine $\vec C$ completes at each $v \in G$ at some time $t'_v \in [t',t'+2d)$, we get that \guard{8} or \guard{9} is satisfied at time $t'_v$ at node $v$. Hence, $v$ transitions to either of the output states depending on the output value of $\vec C$. The last two claims of the lemma follow from the agreement and validity properties of the consensus routine~$\vec C$.
\end{proof}

Now we can show that if all correct nodes transition to $\spulse$ within a time window of length $2d$, then all correct nodes remain synchronised with skew $2d$ and controlled accuracy bounds. Thus, the system stabilises.

\begin{lemma}\label{lemma:synchronisation}
    Suppose there exists an interval $[t, t+2d)$ such that for all $v \in G$ it holds that $p(v,t_v)=1$ for some $t_v \in [t, t+2d)$. Then there exists $t' \in [t + T_2/\vartheta, t + (T_2+\tconsensus)/\vartheta-2d)$ such that $p_\text{next}(v, t_v) \in [t', t'+2d)$ for all $v \in G$.
\end{lemma}
\begin{proof}
First, observe that if any node $v \in G$ transitions to \spulse\ at time $t_v$, then node $v$ transitions to state $\slisten$ in the auxiliary state machine at time $t_v$. To see this, note that node $v$ must have transitioned to $\soutputone$ in the auxiliary state machine at time $t_v$ in order to satisfy \guard{2} leading to state $\spulse$. Furthermore, once this happens, node $v$ transitions immediately from $\soutputone$ to $\slisten$ in the auxiliary state machine. Note that no $v\in G$ transitioned to \swait\ during $(t_v-T_2/\vartheta,t_v)$, as $v$ waits for at least this time before initiating another consensus instance after transitioning to \soutputone. Hence, \constrref{constr:t2} ensures that no correct node stores any \swait\ messages from any other correct node in its \swait\ sliding window buffer. It follows that each $v\in G$ is in state \slisten\ at time $t_v$ and cannot leave before time $t+T_1/\vartheta$.

Next, note that $v \in G$ will not transition to $\srecover$ before time $t_v + T_1/\vartheta \ge t + 3d$ by \guard{1} and \constrref{constr:t1}. By assumption, every $u \in G$ transitions to $\spulse$ by time $t+2d$, and thus, node $v$ observes a $\spulse$ message from at least $n-f$ correct nodes $u \in G$ by time $t + 3d$. Thus, all correct nodes observe a $\spulse$ message from at least $n-f$ nodes during $[t,t+T_1/\vartheta) = [t,t+3d)$ satisfying \guard{1}. Hence, every correct node $v \in G$ transitions to $\swait$ during the interval $[t+T_1/\vartheta, t+T_1+2d)$ and remains there until \guard{2} or \guard{2'} is activated. Denote by $t'_v$ the next transition of $v\in G$ to \soutputone\ or \soutputzero\ after time $t_v$ ($t_v':=\infty$ if no such time exists) and set $t':=\min_{v\in G}\{t_v'\}$. \guard{2'} cannot be activated before time $\min\{t',t+\twait/\vartheta\}$.
    
Now let us consider the auxiliary state machine. From the above reasoning, we get that every node $v \in G$ will observe at least $n-f$ nodes transitioning from $\spulse$ to $\swait$ during the interval $[t+T_1/\vartheta, t+T_1+3d)$. As we have $\tlisten/\vartheta \ge (1-1/\vartheta)T_1+3d$ by \constrref{constr:tlisten}, both \guard{4} and \guard{5} become satisfied for $v$ during the same interval. Thus, node $v$ transitions to state $\sone$ during the interval $[t+T_1/\vartheta, t+T_1+3d)$. Here, we use that \guard{3} is not active at any $v\in G$ that is not in \srecover, implying that $t'\geq \min\{t+T_2/\vartheta,t+\twait/\vartheta\}=t+T_2/\vartheta>t+T_1+3d$ by \constrref{constr:t2} and \constrref{constr:twait}. Thus, \guard{2'} cannot become active before time $t+T_1+3d$, yielding that each $v\in G$ transitions to \sread\ before \guard{3} can become active. We claim that $v$ transitions to $\sexecone$ during $[t+(T_1+T_2)/\vartheta, t + T_1 + T_2 + 3d) \subseteq [t + T_2/\vartheta,  t+ T_2/\vartheta + \tau)$ by \constrref{constr:tau}. Assuming otherwise, some $v\in G$ would have to transition to state \swait\ before time $T_2/\vartheta + \tau$. However, $t'\geq T_2/\vartheta + \tau$ by \remarkref{remark:st-properties} and $t+\twait/\vartheta = t+(T_2+\tconsensus)/\vartheta > t + T_2/\vartheta + \tau$ by \constrref{constr:twait} and \constrref{constr:consensus-time}.

Note that $t+T_1+4d<t+T_2/\vartheta$ by \constrref{constr:t1} and \constrref{constr:t2} and that no correct node transitions to $\swait$ again after time $t+T_1+3d$ before transitioning to $\soutputone$ and spending, by \constrref{constr:t1}, at least $T_1/\vartheta > 2d$ time in $\spulse$. Therefore, we can apply \lemmaref{lemma:consensus-simulation} with $t_0=t+T_2/\vartheta$, yielding a time $t'< t +T_2/\vartheta + \tconsensus/\vartheta-2d$ such that each $v \in G$ transitions to $\soutputone$ in the auxiliary state machine at time $t'_v \in [t',t'+2d)$. This triggers a transition to \spulse\ in the main state machine.
\end{proof}

\subsection{Ensuring stabilisation}\label{ssec:ensuring}

We showed above that if all correct nodes eventually generate a pulse with skew $2d$, then the pulse synchronisation routine stabilises. In this section, we show that this is bound to happen. The idea is that if stabilisation does not take place within a certain interval, then all nodes end up being simultaneously in state $\srecover$ in the main state machine and state $\slisten$ in the auxiliary state machine, until eventually timeout $\tactive$ expires. This in turn allows the ``passive'' stabilisation mechanism to activate by having timer $\tactive$ expire at every node and stabilise the system as shown by the lemma below.

\begin{lemma}\label{lemma:passive-stab}
    Let $t < \tactive/\vartheta$ and $t^* = \min_{v\in G}\{p_\text{next}(v, t)\}$. Suppose the following holds for every node $v \in G$:
\begin{itemize}[noitemsep]
    \item node $v$ is in state $\srecover$ and $\slisten$ at time $t$, and
    \item \guard{4} is not satisfied at $v$ during $[t,t^*+2d)$,
\end{itemize}
Then $t^* < \tactive + \rho + \tconsensus/\vartheta$ and every $v \in G$ transitions to $\spulse$ at time $t_v \in [t^*,t^*+2d)$.
\end{lemma}
\begin{proof}
    Observe that \guard{3} is not satisfied before time $\tactive/\vartheta$. As \guard{4} is not satisfied during $[t, t^*)$, no correct node leaves state $\tlisten$ before time $\tactive/\vartheta$. Since no $v \in G$ can activate \guard{2} during $[t, t^*)$, we have that every $v \in G$ remains in state $\srecover$ during this interval. Let $t_0\in [\tactive/\vartheta, \tactive+\rho)$ be the minimal time (after time $t$) when \guard{3} becomes satisfied at some node $v \in G$. From the properties of the simulation algorithm given in \remarkref{remark:st-properties}, we get that no $w\in G$ transitions away from $\sexecone$ before time $t_0+\tau \geq \tactive/\vartheta + \tau$.

Since $\tau \ge (1-1/\vartheta)\tactive +\rho$ by \constrref{constr:tau}, we conclude that no $w\in G$ transitions to $\soutputone$ (and thus $\spulse$) before time $\tactive + \rho$. Therefore, each $w\in G$ transitions to $\sexecone$ at some time $r_w \in [\tactive/\vartheta, \tactive+\rho)\subseteq [t_0,t_0+\tau)$. Recall that \guard{4} is not satisfied during $[t_0, t^* + 2d)$. Hence, we can apply \lemmaref{lemma:consensus-simulation} to the interval $[t_0,t_0+\tau)$, implying that every $w \in G$ transitions to $\soutputone$ at some time $t_w \in [t^*, t^* +2d)$. Thus, each $w \in G$ transitions from $\srecover$ to $\spulse$ at time $t_w$. Finally, \lemmaref{lemma:consensus-simulation} also guarantees that $t^* < t_0+\tconsensus/\vartheta \leq \tactive + \rho + \tconsensus/\vartheta$.
\end{proof}

We will now establish a series of lemmas in order to show that we can either apply \lemmaref{lemma:synchronisation} directly or in conjunction with \lemmaref{lemma:passive-stab} to guarantee stabilisation. In the remainder of this section, we define the following abbreviations:
\begin{align*}
 \alpha &= \tlisten +d \\
 \beta &= \alpha + \twait + \gamma + 4T_1 + 3d + 2\delta \\
 \beta' &= (\tactive-T_2-\tconsensus)/\vartheta \\ 
 \delta &= T_1+d+2\tlisten+T_2+\tconsensus\\
 \gamma &= T_2/\vartheta - \tlisten - 5T_1 - 3d.
\end{align*}

\begin{remark}
    We have that $\gamma < \delta < \beta \le \beta' < \tactive/\vartheta$, where the inequality $\beta\leq \beta'$ is equivalent to \constrref{constr:beta}.
\end{remark}

In the following, we consider the time intervals $[\alpha, \beta]$ and $[\alpha,\beta']$ that depend on different timeouts. We distinguish between two cases and show that in either case stabilisation is ensured. The cases are:
\begin{itemize}[noitemsep]
    \item no correct node transitions to $\sone$ during $[\alpha,\beta]$, and
    \item some correct node transitions to $\sone$ during $[\alpha,\beta']$.
\end{itemize}
These cases correspond to our proof strategy described in \sectionref{ssec:outline}: If the first case occurs, all nodes end up in the $\srecover$ state and the ``passive'' stabilisation mechanism guarantees stabilisation after the timer $\tactive$ expires. If the first case does not hold, then we must be in the latter case if $\beta' \ge \beta$ holds, which happens to be true when \constrref{constr:beta} is satisfied. In the second case, we show that a consensus instance will be run by all correct processors, which then can be used to ensure that all nodes agree that a pulse should be generated (output ``1'') or that the system has not stabilised (output ``0''), which leads to everyone transitioning to the $\srecover$ state. We start by analysing the case where some correct node transitions to $\sone$ during the interval $[\alpha,\beta']$.

\begin{lemma}\label{lemma:input-1-cases}
    Suppose node $v \in G$ transitions to $\sone$ at time $t \in [\alpha,\beta']$. Then there exists a time $t' \in [t, \tactive/\vartheta-2d)$ such that one of the following holds:
    \begin{enumerate}[noitemsep]
      \item every $u \in G$ satisfies $p(u,t_u)=1$ for some $t_u \in [t', t'+2d)$, or
      \item every $u \in G$ is in state $\srecover$ and $\slisten$ at time $t'+2d$ and \guard{4} is not satisfied at $u$ during $[t'+2d,t^*+2d]$, where $t^*:=\min_{w\in G}\{p(w,t'+2d)\}$.
    \end{enumerate}
\end{lemma}
\begin{proof}
    By \constrref{constr:t1}, we have $T_1/\vartheta > 2d$. As $\alpha\leq t \le \beta' < (\tactive-T_2)/\vartheta$, we can apply \lemmaref{lemma:consistent-init} to time $t$. Due to \constrref{constr:t2}, we can apply \lemmaref{lemma:consensus-simulation} with time $t_0 = t - \tlisten - d +T_2/\vartheta$, yielding a time $t' \le \beta' - \tlisten - 3d +T_2/\vartheta + \tconsensus/\vartheta<\tactive/\vartheta-2d$ such that each $u \in G$ transitions to the same output state $\soutputzero$ or $\soutputone$ in the auxiliary state machine at time $t_u \in [t',t'+2d)$. If this state is $\soutputone$, each $u \in G$ transitions to $\soutputone$ at time $t_u \in [t',t'+2d)$. This implies that \guard{2} is activated and $u$ transitions to $\spulse$ so that $p(u,t_u)=1$. If this state is $\soutputzero$, \guard{2'} implies that each $u\in G$ either remains in or transitions to state $\srecover$ at time $t_u$. Moreover, node $u$ immediately transitions to state $\slisten$ in the auxiliary state machine. Finally, note that \lemmaref{lemma:consistent-init} also states that \guard{4} cannot be satisfied again before time $t^*+T_1/\vartheta>t^*+2d$.
\end{proof}

\begin{corollary}\label{coro:stab-by-one}
    Suppose node $v \in G$ transitions to $\sone$ at time $t \in [\alpha,\beta']$. Then there exists $t' < \tactive + \rho + \tconsensus/\vartheta$ such that every $u \in G$ satisfies $p(v,t_u)=1$ for some $t_u \in [t', t'+2d)$.
\end{corollary}
\begin{proof}
    We apply \lemmaref{lemma:input-1-cases}. In the first case of \lemmaref{lemma:input-1-cases}, the claim immediately follows. In the second case, it follows by applying \lemmaref{lemma:passive-stab}.
\end{proof}

We now turn our attention to the other case, where no node $v \in G$ transitions to $\sone$ during the time interval $[\alpha,\beta]$.

\begin{lemma}\label{lemma:no-exec-one}
    If no $v \in G$ transitions to $\sone$ during $[\alpha,\beta]$, then no $v \in G$ transitions to state $\sexecone$ during $[\alpha+T_1+\twait,\beta]$.
\end{lemma}
\begin{proof}
    Observe that any $v\in G$ that does not transition to $\spulse$ during $[\alpha,\alpha+T_1+\twait]$ must transition to $\srecover$ at some time from that interval once \guard{2'} is satisfied. Note that leaving state $\srecover$ requires \guard{2} to be satisfied, and hence, a transition to $\soutputone$ in the auxiliary state machine. However, as no node $v\in G$ transitions to $\sone$ during $[\alpha, \beta]$, it follows that each $v\in G$ that is in state $\sinputone$ in the auxiliary state machine at time $t \in [\alpha+T_1+\twait, \beta]$ is also in state $\srecover$ in the main state machine. Thus, \guard{6} cannot be satisfied. We conclude that no correct node transitions to state $\sexecone$ during the interval $[\alpha+T_1+\twait,\beta]$.
\end{proof}

We now show that if \guard{4} cannot be satisfied for some time, then there exists an interval during which no correct node is simulating a consensus instance. 

\begin{lemma}\label{lemma:no-consensus-interval}
    Let $t \in (d, \beta - 2\delta)$ and suppose \guard{4} is not satisfied during the interval $[t,t+\gamma]$. Then there exists a time $t' \in [t, \beta-\delta]$ such that no $v \in G$ is in state $\sexeczero$ or $\sexecone$ during $(t'-d,t']$. 
\end{lemma}
\begin{proof}
Observe that \guard{3} cannot be satisfied before time $t+\gamma+T_2/\vartheta < \tactive/\vartheta$ at any correct node, and by the assumption, \guard{4} is not satisfied during the interval $[t,t+\gamma]$. We proceed via a case analysis.

    First, suppose $v \in G$ is in state $\slisten$ at time $t$. As \guard{3} or \guard{4} cannot be satisfied during $[t,t+\gamma]$ node $v$ remains in $\slisten$ until time $t+\gamma$. Moreover, if $v$ leaves $\slisten$ during $[t+\gamma,t+\gamma+T_2/\vartheta]$, it must do so by transitioning to $\sread$. Hence, it cannot transition to $\sexeczero$ or $\sexecone$ before \guard{6}, \guard{6'} or \guard{7} is satisfied. Either way, $v$ cannot reach states $\sexeczero$ or $\sexecone$ before time $t+\gamma+T_2/\vartheta$. Hence, in this case, $v$ is not in the two execution states during $I_1 = [t,t+\gamma+T_2/\vartheta) = [t, t+2T_2/\vartheta - \tlisten - 5T_1 - 3d)$.

Let us consider the second case, where $v$ is not in $\slisten$ at time $t$. Note that the timer $T_2$ cannot be reset at $v$ during the interval $[t+\tlisten, t+\gamma]$, as this can happen only if $v$ transitions to $\sinputzero$ or $\sinputone$ and \guard{4} cannot be satisfied during $[t,t+\gamma]$. Hence, the only way for this to happen is if $v$ transitions from $\sread$ to either $\sinputzero$ or $\sinputone$ during the interval $[t,t+\tlisten]$.

    It follows that $v$ cannot \emph{transition} to $\sexeczero$ or $\sexecone$ during the interval $(t + \tlisten + T_2, t + \gamma+ T_2/\vartheta)$. Moreover, if $v$ transitions to these states before $t+\tlisten+T_2$, then $v$ must transition away from them by time $t+\tlisten+T_2+\tconsensus$, as \guard{8} or \guard{9} become satisfied. Therefore, node~$v$ is in $\slisten$, $\sread$, $\sinputzero$ or $\sinputone$ during $I_2 = [t+\tlisten+T_2+\tconsensus, t+\gamma+T_2/\vartheta)$. Applying \constrref{constr:t2thetabound}, we get that 
\[
    \tlisten + T_2 + \tconsensus < 2T_2/\vartheta - \tlisten - 5T_1-4d,
\]
    and hence, by setting $t' = t + \tlisten + T_2 + \tconsensus+d < \beta - \delta$, we get that $(t'-d,t'] \subseteq I_1 \cap I_2$. That is, in either case $v$ is not in state $\sexeczero$ or $\sexecone$ during this short interval.
\end{proof}

Next, we show that the precondition of the previous lemma will be satisfied in due time.
\begin{lemma}\label{lemma:disable-g4}
    There exists $t \in [\alpha + T_1+\twait+d, \alpha + 4T_1+\twait+3d+\gamma)\subset (d,\beta - 2\delta)$ such that \guard{4} is not satisfied during $[t,t+\gamma]$ at any $v\in G$.
\end{lemma}
\begin{proof}
Let $t' \in [\alpha + T_1 + \twait, \alpha + T_1 + \twait + d + \gamma]$ be a time when some $v \in G$ transitions to state $\swait$. If such $t'$ does not exist, the claim trivially holds for $t = \alpha + T_1 + \twait + d$. Otherwise, since $t' \leq  \tlisten+T_1+\twait+2d+ \gamma < (\tactive-T_2)/\vartheta$ by \constrref{constr:tactive}, we can apply \lemmaref{lemma:separation} to time $t'$, which yields that no $u \in G$ transitions to $\swait$ during 
\[
[t'+3T_1+d, t'+T_2/\vartheta - 2T_1 -d) = [t-\tlisten-d,t+\gamma),
\]
where $t = t'+\tlisten+3T_1+2d$. Thus, \guard{4} is not satisfied at any $v\in G$ during $[t,t+\gamma]$.
\end{proof}

Using the above statements, we can now infer that if no node transitions to \sinputone\ for a while, this implies that no node transitions to \spulse\ for a certain period. Subsequently, we will use this to conclude that then the preconditions of \lemmaref{lemma:passive-stab} are satisfied.

\begin{lemma}\label{lemma:no-pulse}
    Suppose no $v\in G$ transitions to $\sone$ during $[\alpha,\beta]$. Then no $v \in G$ transitions to $\spulse$ during $[\beta - \delta, \beta]$.
\end{lemma}
\begin{proof}
First, we apply \lemmaref{lemma:disable-g4} to obtain an interval $[t,t+\gamma] \subseteq [\alpha+T_1+\twait+d,\beta - 2\delta]$ during which \guard{4} cannot be satisfied at any $v \in G$. Applying \lemmaref{lemma:no-consensus-interval} to this interval yields an interval $(t'-d,t']\subset (t,\beta-\delta]$ during which no $v \in G$ is in state $\sexeczero$ or $\sexecone$. This implies that no node is running a consensus instance during this time interval, and moreover, no messages from prior consensus instances are in transit to or arrive at any correct node at time $t'$. In particular, any $v \in G$ that (attempts to) simulate a consensus instance at time $t'$ or later must first reinitialise the simulation by transitioning to $\sexeczero$ or $\sexecone$. 
    
Next, let us apply \lemmaref{lemma:no-exec-one}, to see that no $v$ transitions to $\sexecone$ during the interval $[\alpha + T_1 + \twait,\beta]\supset [t',\beta]$. Thus, if any node $v \in G$ attempts to simulate $\vec C$, it must start the simulation by transitioning to $\sexeczero$. This entails that any correct node attempting to simulate $\vec C$ does so with input 0. Because $\vec C$ is a silent consensus routine (see \definitionref{def:silent}), this entails that $v$ does not send any message related to $\vec C$ unless it receives one from a correct node first, and in absence of such a message, it will not terminate with output $1$. We conclude that no $v\in G$ transitions to $\soutputone$ during $[t',\beta]\supseteq [\beta - \delta,\beta]$. The claim follows by observing that a transition to $\spulse$ in the main state machine requires a transition to $\soutputone$ in the auxiliary state machine.
\end{proof}

\begin{lemma}\label{lemma:clean-passive-init}
    Suppose no $v \in G$ transitions to $\spulse$ during $[\beta-\delta,\beta]$. Then at time $\beta$ every $v \in G$ is in state $\srecover$ in the main state machine and state $\slisten$ in the auxiliary state machine. Moreover, \guard{4} is not satisfied during $[\beta,t^*+2d)$, where $t^* = \min_{v\in G}\{p_\text{next}(v, \beta)\}$.
\end{lemma}
\begin{proof}
    As no $v \in G$ transitions $\spulse$ during $[\beta-\delta,\beta]$, either \guard{1} or \guard{2'} lead each $v \in G$ to $\srecover$. More precisely, every $v \in G$ is in state $\srecover$ of the main state machine during $[\beta-\delta+T_1+\twait,\beta]$. Next, observe that \guard{4} is not satisfied during $[\beta-\delta+T_1+\tlisten+d, \beta]$, and because $\beta<\tactive/\vartheta$, \guard{3} cannot be active at any time from this interval either. It follows that each $v \in G$ is in state $\slisten$ of the auxiliary state machine during $[\beta-\delta+T_1+d+2\tlisten+T_2+\tconsensus,\beta]=[\beta,\beta]$, i.e., the first claim of the lemma holds. For the second claim, observe that \guard{4} cannot be satisfied before the next transition to $\swait$ by a correct node occurs. This cannot happen before $T_1/\vartheta$ time has passed after a correct node transitioned to $\spulse$ by \guard{1}. Since $T_1/\vartheta > 2d$ by \constrref{constr:t1}, the claim follows.
\end{proof}

We can now show that if no $\sone$ transitions occur during $[\alpha,\beta]$, then all nodes end up in the $\srecover$ state in the main state machine before $\tactive$ timeout expires at any node. This will eventually activate \guard{3} at every correct node, leading to a correct simulation of $\vec C$ with all 1 inputs.

\begin{corollary}\label{coro:stab-by-no-input}
    Suppose no $v\in G$ transitions to $\sone$ during $[\alpha,\beta]$. Then there exists a time $t < \tactive + \rho + \tconsensus$ such that every $v \in G$ transitions to $\spulse$ at time $t_v \in [t,t+2d)$.
\end{corollary}
\begin{proof}
    By \lemmaref{lemma:no-pulse} and \lemmaref{lemma:clean-passive-init}, the prerequisites of \lemmaref{lemma:passive-stab} are satisfied at time $t=\beta < \tactive/\vartheta$.
\end{proof}

It remains to show that the constraints in \tableref{table:constraints} can be satisfied for some $\vartheta > 1$ such that all timeouts are in $O(R)$.

\begin{lemma}\label{lemma:constraints}
    Let $1 < \vartheta < (2+\sqrt{32})/7\approx 1.094$. The constraints in \tableref{table:constraints} can be satisfied with all timeouts in $O(R)$.
\end{lemma}
\begin{proof}
    Recall that $R$ is the number of rounds the consensus routine needs to declare output and $T(R)$ is the time required to simulate the consensus routine. We parametrize $T_2$ and $\tactive$ using $X$ and $Y$, where we require that $X \in \Theta(R)$ is chosen so that $T(R)/X \le \varepsilon$, for a constant $\varepsilon>0$ to be determined later. We then can set
\begin{align*}
    T_1 &:= 3\vartheta d \\
    \tlisten &:= 3\vartheta^2 d\\
    T_2 &:= X \\
    \tau &:= \max\left\{(1-1/\vartheta)X+(3\vartheta^2 + 9\vartheta + 3)d,(1-1/\vartheta)Y+\rho\right\}\\
    \tconsensus &:= \vartheta(\tau+\varepsilon X)\\
    \twait &:= X + \tconsensus\\
    \tactive &:= Y,
\end{align*}
immediately satisfying \constrref{constr:t1}, \constrref{constr:tlisten}, \constrref{constr:tau}, \constrref{constr:consensus-time}, and \constrref{constr:twait}. Moreover, \constrref{constr:t2} holds by requiring that $X$ is at least a sufficiently large constant.

For the remaining constraints, denote by $C(\vartheta,d)$ a sufficiently large constant subsuming all terms that depend only on $\vartheta$ and $d$ and abbreviate $\tconsensus':=(\vartheta-1)\max\{X,Y\} + \varepsilon \vartheta X$ (i.e., the non-constant terms of $\tconsensus$). To satisfy \constrref{constr:t2thetabound}, \constrref{constr:tactive}, and \constrref{constr:beta}, it is then sufficient to guarantee that
\begin{align*}
(2/\vartheta-1)X &> \tconsensus' + C(\vartheta,d)\\
Y &\geq 4X + \vartheta(X+\tconsensus') + C(\vartheta,d)\\
Y &\geq 2X + \tconsensus' + 3\vartheta(X+\tconsensus') + C(\vartheta,d),
\end{align*}
where the second inequality automatically holds if the third is satisfied. We also note that $Y\geq X$ is a necessary condition to satisfy the third inequality, implying that we may assume that $\tconsensus'=(\vartheta-1)Y+\varepsilon X$. We rearrange the remaining inequalities, yielding
\begin{align}
(2/\vartheta-1-\varepsilon)X &> (\vartheta-1)Y + C(\vartheta,d)\nonumber\\
(2+2\vartheta-3\vartheta^2)Y & \geq (2+3\vartheta(1+\varepsilon \vartheta))X + C(\vartheta,d).\label{eq:YtoX}
\end{align}
Recall that $\vartheta$ and $C(\vartheta,d)$ are constant, and $\varepsilon$ is a constant under our control. Hence, these inequalities can be satisfied if and only if
\begin{align*}
(2/\vartheta-1)(2+2\vartheta-3\vartheta^2)>(2+3\vartheta)(\vartheta-1).
\end{align*}
The above inequality holds for all $\vartheta \in (1,(2+\sqrt{32})/7)$. Note that, as $\vartheta$, $\varepsilon$, and $C(\vartheta,d)$ are constants, we can choose $X\in \Theta(R)$ as initially stated, implying that all timeouts are in $O(R)$, as desired.
\end{proof}

\begin{corollary}\label{coro:stabilisation-happens}
For $\vartheta < (2+\sqrt{32})/7$ and suitably chosen timeouts, in any execution there exists $t \in O(R)$ such that every $v \in G$ transitions to $\spulse$ during the time interval $[t,t+2d)$.
\end{corollary}
\begin{proof}
We choose the timeouts in accordance with \lemmaref{lemma:constraints}. If no $v \in G$ transitions to $\sone$ during $[\alpha,\beta]$, \corollaryref{coro:stab-by-no-input} yields the claim. Otherwise, some node $v \in G$ transitions to state $\sinputone$ during the interval $[\alpha,\beta] \subseteq [\alpha,\beta']$ and the claim holds by \corollaryref{coro:stab-by-one}.
\end{proof}

Next, we observe that the accuracy bounds can be set to be within a constant factor apart from each other.

\begin{corollary}\label{coro:phitheta}
    Let $\vartheta < (2+\sqrt{32})/7$ and $\varphi_0(\vartheta) = 1+5(\vartheta-1)/(2+2\vartheta-3\vartheta^2)\in 1+O(\vartheta-1)$. For any constant $\varphi > \varphi_0(\vartheta)$, we can obtain accuracy bounds satisfying $\Phi^+/\Phi^- \le \varphi$ and $\Phi^-,\Phi^+ \in \Theta(R)$.
\end{corollary}
\begin{proof}
By \lemmaref{lemma:synchronisation}, the accuracy bounds we get from the constuction are $\Phi^- = T_2/\vartheta$ and $\Phi^+ = (T_2 + \tconsensus)/\vartheta$. Choosing the timeouts as in the proof of \lemmaref{lemma:constraints}, we get that $\Phi^+/\Phi^-=1+(\vartheta-1)Y/X+\varepsilon$, where $\varepsilon$ is an arbitrarily small constant. Checking Inequality~\eqref{eq:YtoX}, we see that we can choose $Y/X=(2+3(1+\varepsilon))/(2+2\vartheta-3\vartheta^2)$. Choosing $\varepsilon$ sufficiently small, the claim follows.
\end{proof}

\subsection{Proof of \texorpdfstring{\theoremref{thm:resync-to-pulse}}{Theorem~3}}

Finally, we are ready to prove the main theorem of this section.

\resynctopulse*
\begin{proof}By the properties of the resynchronisation algorithm $\vec B$, we get that a good resynchronisation pulse occurs within time $T(\vec B)$. Once this happens, \corollaryref{coro:stabilisation-happens} shows all correct nodes transition to $\spulse$ during $[t,t+2d)$ for $t \in T(\vec B)+O(R)$. By \lemmaref{lemma:synchronisation} we get that the algorithm stabilises by time $t$ and has skew $\sigma = 2d$. From \corollaryref{coro:phitheta} we get that the accuracy bounds can be set to be within factor $\varphi$ apart without affecting the stabilisation time asymptotically.

    To analyse the number of bits sent per time unit, first observe that the main state machine communicates whether a node transitions to $\spulse$ or $\swait$. This can be encoded using messages of size $O(1)$. Moreover, as node remains $\Omega(d)$ time in $\spulse$ or $\swait$, the main state machine sends only $O(1)$ bits per time unit. Second, the auxiliary state machine does not communicate apart from messages related to the simulation of consensus. The non-self-stabilising pulse synchronisation algortihm sends messages only when a node generates a pulse and the time between pulses is $\Omega(d)$. Thus, while simulating $\vec C$, each node broadcasts at most $M(\vec C)+O(1)$ bits per time unit.
\end{proof}

% --- RESYNC PULSES ---

\section{Resynchronisation algorithms}\label{sec:resync}

In this section, we give the second key component in the proof of \theoremref{thm:master}. We show that given \emph{pulse synchronisation} algorithms for networks of small size and with low resilience, it is possible to obtain \emph{resynchronisation} algorithms for large networks and with high resilience. More precisely, we show the following theorem.

\resynctheorem*

\subsection{The high-level idea}

Our goal is to devise a self-stabilising resynchronisation algorithm with skew $\rho\in O(d)$ and separation window $\Psi$ for $n$ nodes that tolerates $f < n/3$ faulty nodes. That is, we want an algorithm that guarantees that there exists a time $t$ such that all correct nodes locally generate a single resynchronisation pulse during the interval $[t, t+\rho)$ and no new pulse during the interval $[t+\rho, t+\rho+\Psi)$. Note that a correct resynchronisation algorithm is also allowed to generate various kinds of spurious resynchronisation pulses, such as pulses that are followed by a new resynchronisation pulse too soon (i.e., before $\Psi$ time units have passed) or pulses that are only generated by a subset of the correct nodes. 

%###
\paragraph{The algorithm idea.}
%###
In order to illustrate the idea behind our resynchronisation algorithm, let us ignore clock drift and suppose we have two sources of pulses that generate pulses with fixed frequencies. Whenever either source generates a pulse, then a resynchronisation pulse is triggered as well. If the sources generate pulses with frequencies that are coprime multiples of (a sufficiently large) $C\in \Theta(\Psi)$, then we are guaranteed that eventually one of the sources produces a pulse followed by at least $\Psi$ time units before a new pulse is generated by either of the two sources. See \figureref{fig:summing-sources}a for an illustration. 

\begin{figure}
\begin{center}
 \includegraphics[page=17]{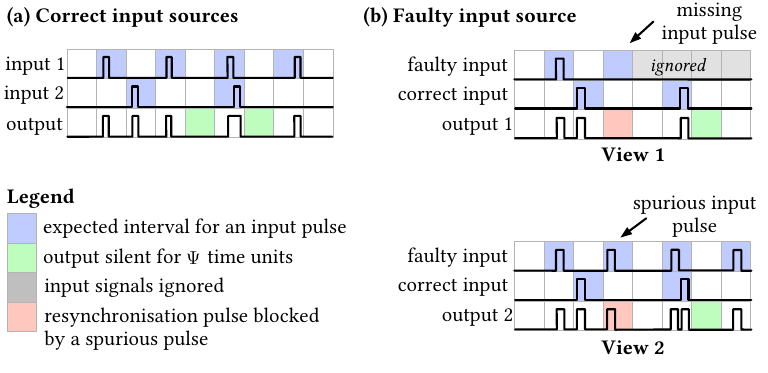}
\end{center}
\caption{
    Idea of the resynchronisation algorithm. We take two pulse sources (top two rows) with coprime frequencies and output the logical OR of the two sources (bottom row). Here, the dim gray lines delimit time intervals of length $C \in \Theta(\Psi)$. The blue regions indicate intervals when a correctly operating source should generate a pulse. In this example, the pulses of the first source should occur approximately every $2C$ time units, whereas the pulses of the second source should occur approximately every $3C$ time units. (a) Two correct sources that pulse with set frequencies. All correct nodes observe the same input pulses, and hence, produce the same output pulses. A good resynchronisation pulse is a pulse that is followed by a green block indicating a silence of at least $\Psi$ time units. (b) One faulty source that produces spurious pulses. As one of the sources is faulty, two nodes may have different observations on the output of a faulty source given in View 1 (upper figure) and View 2 (lower figure). In this example, a spurious input pulse occurs, that is, an input pulse that is inconsistently detected by different nodes. We devise our construction so that only pulses that follow the frequency bounds of the source are accepted and if the source fails to adhere to the frequency bounds, it will be ignored for some time. Here, the nodes observing View 1 do not receive the input pulse, and thus, they ignore any pulses from the first source for a while (gray region). However, nodes observing View 2 receive a spurious pulse that adheres to the frequency bounds (and do not detect any suspicious behaviour). Thus, nodes with View 2 output a pulse, although nodes observing View 1 do not (red region). Once the faulty source is silenced, the correctly working source has time to produce a good resynchronisation pulse no matter what the faulty source does. This ensures that nodes observing either view generate an output signal within a small enough time window that is followed by at least $\Psi$ time of silence. \label{fig:summing-sources}}
\end{figure}

Put otherwise, suppose $p_h(v,t) \in \{0,1\}$ indicates whether node $v$ observes the pulse source $h \in \{0,1\}$ generating a pulse at time $t$. Using the above scheme, the output variable for the resynchronisation algorithm would be $r(v,t) = \max_{h \in \{0,1\}} \{ p_h(v,t) \}$. If, eventually, each source $h$ generates a pulse roughly every $C_h$ time units, setting $C_0 < C_1$ to be coprime integer multiples of $C\in \Theta(\Psi)$ (we allow for a constant-factor slack to deal with clock drift, etc.), we eventually have a time when a pulse is generated by one source, but no source will generate another pulse for at least $\Psi$ time units.

Obviously, if we had such reliable self-stabilising pulse sources for $n$ nodes and $f < n/3$ faulty nodes, then we would have effectively solved the pulse synchronisation problem already. However, our construction given in \sectionref{sec:pulse-gen} relies on having a resynchronisation algorithm. Thus, in order to avoid this chicken-and-egg problem, we partition the set of $n$ nodes into two and have each part run an instance of a pulse synchronisation algorithm with resilience almost $f/2$. That is, we take two \emph{pulse synchronisation algorithms} with low resilience, and use these to obtain a \emph{resynchronisation algorithm} with high resilience. This allows us to recursively construct resynchronisation algorithms starting from trivial pulse synchronisation algorithms that do not tolerate any faulty nodes.

The final obstacle to this construction is that we cannot guarantee that \emph{both} instances with smaller resilience operate correctly, as the total number of faults exceeds the number that can be tolerated by each individual instance. We overcome this by enlisting the help of \emph{all} nodes to check, for each instance, whether its output appears to satisfy the desired frequency bounds. If not, its output is conservatively filtered out (for a sufficiently large period of time) by a voting mechanism. This happens only for an incorrect output, implying that the fault threshold for the respective instance must have been exceeded. Accordingly, the other instance is operating correctly and, thanks to the absence of further interference from the faulty instance, succeeds in generating a resynchronisation pulse. \figureref{fig:summing-sources}b illustrates this idea.

%###
\paragraph{Using two pulse synchronisers.}
%###
We now overview how to use two pulse synchronisation algorithms to implement our simple resynchronisation algorithm described above. Let
\begin{align*}
 n_0 = \lfloor n/2 \rfloor &\text{ and } n_1 = \lceil n/2 \rceil \\
 f_0 = \lfloor (f-1)/2 \rfloor &\text{ and } f_1 = \lceil (f-1)/2 \rceil.
\end{align*}
Observe that we have $n = n_0 + n_1$ and $f = f_0 + f_1 + 1$. We partition the set $V$ of $n$ nodes into two sets $V_h$ for $h \in \{0,1\}$ such that $V = V_0 \cup V_1$, where $V_0 \cap V_1 = \emptyset$ and $|V_h| = n_h$. We now pick two pulse synchronisation algorithms $\vec A_0$ and $\vec A_1$ with the following properties:
\begin{itemize}[noitemsep]
 \item $\vec A_h$ runs on $n_h$ nodes and tolerates $f_h$ faulty nodes, 
 \item $\vec A_h$ stabilises in time $T(\vec A_h)$ and sends $M(\vec A_h)$ bits per time unit and channel, and
 \item $\vec A_h$ has skew $\sigma \in O(d)$ and accuracy bounds $\Phi_h = (\Phi^-_h, \Phi^+_h)$, where $\Phi^-_h, \Phi^+_h \in O(\Psi)$.
\end{itemize}
We let the nodes in set $V_h$ execute the pulse synchronisation algorithm $\vec A_h$. 

An optimistic approach would be to use each $\vec A_h$ as a source of pulses by checking whether at least $n_h - f_h$ nodes in the set $V_h$ generated a pulse within a time window of (roughly) length $\sigma=2d$. Unfortunately, we cannot directly use the pulse synchronisation algorithms $\vec A_0$ and $\vec A_1$ as reliable sources of pulses. There can be a total of $f = f_0 + f_1 + 1 < n/3$ faulty nodes, and thus, it may be that for one $h \in \{0,1\}$ the set $V_h$ contains more than $f_h$ faulty nodes. Hence, the algorithm $\vec A_h$ may never stabilise and can generate spurious pulses at uncontrolled frequencies. In particular, the algorithm may always generate pulses with frequency less than $\Psi$, preventing our simple solution from working. However, we are guaranteed that at least one of the algorithms stabilises.

\begin{lemma}
If there are at most $f$ faulty nodes, then there exists $h \in \{0,1\}$ such that $\vec A_h$ stabilises by time $T(\vec A_h)$.\label{lemma:correct-block}
\end{lemma}
\begin{proof}
Observe that $f_0 + f_1 + 1 = f$. In order to prevent both algorithms from stabilising, we need to have at least $f_0+1$ faults in the set $V_0$ and $f_1+1$ faults in the set $V_1$, totalling $f_0+f_1 + 2 > f$ faults in the system.
\end{proof}

Once the algorithm $\vec A_h$ for some $h \in \{0,1\}$ stabilises, we have at least $n_h - f_h$ correct nodes in set $V_h$ locally generate pulses with skew $\sigma$ and accuracy bounds $\Phi_h = (\Phi^-_h, \Phi^+_h)$. However, it may be that the other algorithm $\vec A_{1-h}$ never stabilises. Moreover, the algorithm $\vec A_{1-h}$ may forever generate spurious pulses at arbitrary frequencies. Here, a spurious pulse refers to any pulse that does not satisfy the skew $\sigma$ and accuracy bounds of $\vec A_{1-h}$. For example, a spurious pulse may be a pulse that only a subset of nodes generate, one with too large skew, or a pulse that occurs too soon or too late.

In order to tackle this problem, we employ a series of threshold votes and timeouts to filter out any spurious pulses generated by an unstabilised algorithm \emph{that violate timing constraints.} This way, we can impose some control on the frequency at which an unstabilised algorithm may trigger resynchronisation pulses. As long as these frequency bounds are satisfied, it is inconsequential if a non-stabilised algorithm triggers resynchronisation pulses at a subset of the nodes only. We want our filtering scheme to eventually satisfy the following properties:
\begin{itemize}[noitemsep]
 \item If $\vec A_h$ has stabilised, then all pulses generated by $\vec A_h$ are accepted.
 \item If $\vec A_h$ has not stabilised, then only pulses that respect given frequency bounds are accepted.
\end{itemize}
More precisely, in the first case the filtered pulses respect (slightly relaxed) accuracy bounds $\Phi_h$ of $\vec A_h$. In the second case, we enforce that the filtered pulses must either satisfy roughly the same accuracy bounds $\Phi_h$ or they must be sufficiently far apart. That is, the nodes will reject any pulses generated by $\vec A_h$ if they occur either too soon or too late. 

Once we have the filtering mechanism in place, it becomes relatively easy to implement our conceptual idea for the resynchronisation algorithm. We apply the filtering mechanism for both algorithms $\vec A_0$ and $\vec A_1$ and use the filtered outputs as a source of pulses in our algorithm, as illustrated in \figureref{fig:resync-idea}. We are guaranteed that at least one of the sources eventually produces pulses with well-defined accuracy bounds. Furthermore, we know that also the other source must either respect the given accuracy bounds or refrain from generating pulses for a long time. In the case that both sources respect the accuracy bounds, we use the coprimality of frequencies to guarantee a sufficiently large separation window for the resynchronisation pulses. Otherwise, we exploit the fact that the unreliable source stays silent for sufficiently long for the reliable source to generate a pulse with a sufficiently large separation window.

\begin{figure}
\begin{center}
 \includegraphics[page=1,scale=0.99]{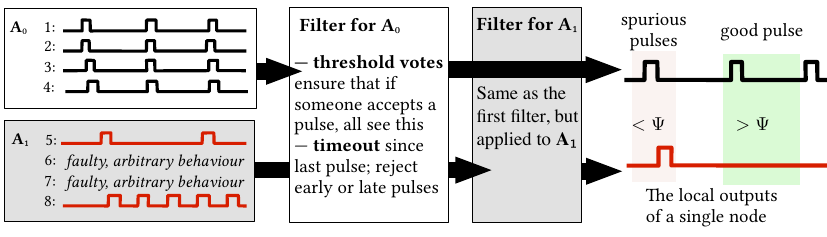}
\end{center}
    \caption{Example of the resynchronisation construction for $8$ nodes tolerating $2$ faults. We partition the network into two parts, each running a pulse synchronisation algorithm $\vec A_i$. The output of $\vec A_i$ is fed into the respective filter and any pulse that passes the filtering is used as a resynchronisation pulse. The filtering consists of (1) having \emph{all} nodes in the network participate in a threshold vote to see if anyone thinks a pulse from $\vec A_i$ occurred (i.e.\ enough nodes running $\vec A_i$ generated a pulse) and (2) keeping track of when was the last time a pulse from $\vec A_i$ occurred to check that the accuracy bounds of $\vec A_i$ are respected: pulses that appear too early or too late are ignored. Moreover, if $\vec A_i$ generates pulses at incorrect frequencies, the filtering mechanism blocks all pulses generated by $\vec A_i$ for $\Theta(\Psi)$ time.
\label{fig:resync-idea}}
\end{figure}

\subsection{Filtering spurious pulses}

Our pulse filtering scheme follows a similar idea as our recent construction of synchronous counting algorithms~\cite{lenzen16firing}. However, considerable care is needed to translate the approach from the (much simpler) synchronous round-based model to the bounded-delay model with clock drift. We start by describing the high-level idea of the approach before showing how to implement the filtering scheme in the bounded-delay model considered in this work. \figureref{fig:overview_resync} illustrates how the underlying pulse synchronisation algorithms are combined with the filtering mechanism.

\begin{figure}[t!]
\begin{center}
 \includegraphics[page=7]{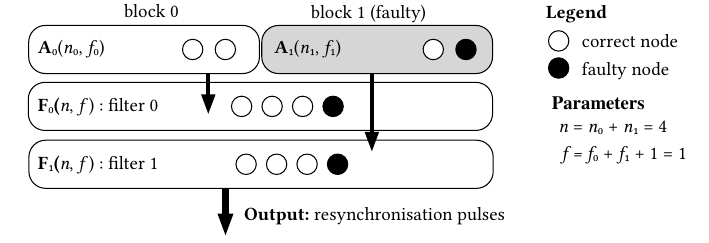}
\end{center}
    \caption{Construction of an $f$-resilient resynchronisation algorithm on $n$ nodes from $f_i$-resilient pulse synchronisation algorithms on $n_i$ nodes, where $f=f_0+f_1+1$ and $n=n_0+n_1$. The $n$ nodes are divided into two groups of $n_0$ and $n_1$ nodes. These groups run pulse synchronisation algorithms $\vec A_0$ and $\vec A_1$, respectively. At least one of these algorithms is guaranteed to stabilise eventually. Here, $\vec A_1$ (gray block) has too many faulty nodes and does not stabilise. All of the $n$ nodes together run two filtering mechanisms $\vec F_0$ and $\vec F_1$ for the outputs of $\vec A_0$ and $\vec A_1$, respectively. These ensure that no correct node locally generates a resynchronisation pulse without all correct nodes registering this event, and then apply timeout constraints to enforce the desired frequency bounds.
\label{fig:overview_resync}}
\end{figure}

For convenience, we refer to the nodes in set $V_h$ as \emph{block $h$.} First, for each block $h \in \{0,1\}$ every node $v \in G$ performs the following threshold vote:
\begin{enumerate}[noitemsep]
  \item If $v$ observes at least $n_h - f_h$ nodes in $V_h$ generating a pulse, vote for generating a resynchronisation pulse.
  \item If at least $n-f$ nodes in $V$ voted for a pulse by block $h$ (within the time period this should take), then $v$ accepts it.
\end{enumerate}
The idea here is that if some correct node \emph{accepts} a pulse in Step 2, then every correct node must have seen at least $n-2f \ge f+1$ \emph{votes} due to Step 1. Moreover, once a node observes at least $f+1$ votes, it can deduce that some correct node saw at least $n_h - f_h$ nodes in block $h$ generate a pulse. Thus, if any correct node accepts a pulse generated by block $h$, then all correct nodes are aware that a pulse \emph{may} have happened.

Second, we have the nodes perform temporal filtering by keeping track of when block $h$ last (may have) generated a pulse. To this end, each node has a local ``cooldown timer'' that is reset if the node suspects that block $h$ has not yet stabilised. If a pulse is accepted by the above voting mechanism, then a resynchronisation pulse is triggered if the following conditions are met:
\begin{enumerate}[noitemsep]
 \item the cooldown timer has expired, and
 \item not too much time has passed since the most recent pulse from $h$.
\end{enumerate}
A correct node $v \in G$ resets its cooldown timer if it
\begin{enumerate}[noitemsep]
 \item observes at least $f+1$ votes for a pulse from block $h$, but not enough time has passed since $v$ last saw at least $f+1$ votes,
 \item observes at least $f+1$ votes, but not $n-f$ votes in a timely fashion, or
 \item has not observed a pulse from block $h$ for too long, that is, block $h$ should have generated a new pulse by now.
\end{enumerate}
Thus, whenever a block $h \in \{0,1\}$ triggers a resynchronisation pulse at node $v \in G$, then each node $u \in G$ either resets its cooldown timer or also triggers a resynchronisation pulse. Furthermore, if $v \in G$ does not observe a pulse from block $h$ within the right time window, it will also reset its cooldown counter. Finally, each node refuses to trigger a resynchronisation pulse when its cooldown timer is active. Note that if $\vec A_h$ stabilises, then eventually the cooldown timer for block $h$ expires and is not reset again. This ensures that eventually at least one of the blocks triggers resynchronisation pulses.

\begin{figure}
\begin{center}
 \includegraphics[page=13,scale=0.99]{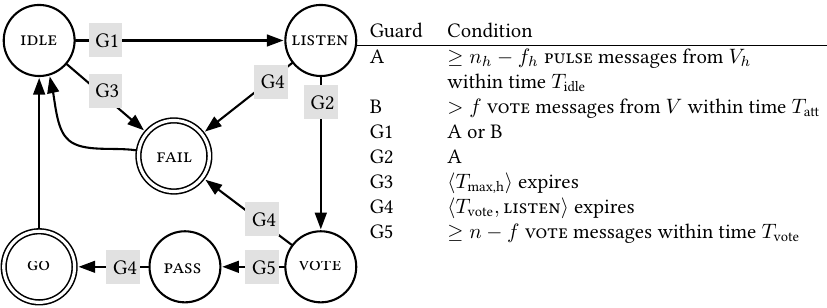}
\end{center}
\caption{The voter state machine is the first of the two state machines used to trigger resynchronisation pulses by block $h \in \{0,1\}$. Every node runs a separate copy of the voter machine for both blocks. The voter state machine performs a threshold vote to ensure that if at \emph{some} node a resynchronisation pulse is (or might be) triggered by block $h$, then \emph{all} correct nodes see this by observing at least $f+1$ \svote\ messages within $\tlarge$ local time. Note that nodes immediately transition from \sfail\ and \sgo\ to state \sidle, as there are no guards blocking these transitions. The two states are used to signal the validator state machine given in \figureref{fig:resync-validator} to generate resynchronisation pulses or to refrain from doing so until the cooldown timer expires.
\label{fig:resync-voter}}
\end{figure}

\begin{figure}
\begin{center}
 \includegraphics[page=14,scale=0.99]{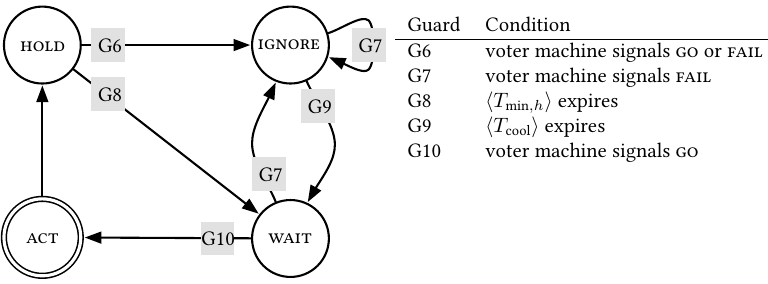}
\end{center}
\caption{The validator state machine is the second of the two state machines used to trigger resynchronisation pulses by block $h \in \{0,1\}$. Every node runs a separate copy of the validator state machine for both blocks. The validator checks that the \sgo\ and \sfail\ transition signals in the voter state machine given in \figureref{fig:resync-voter} satisfy the minimum time bound and that the \sgo\ transitions occur in a timely manner. Note that the transition from state \signore\ to itself resets the timer $\tcool$.
\label{fig:resync-validator}}
\end{figure}

%###
\paragraph{Implementation in the bounded-delay model.}
%###
We implement the threshold voting and temporal filtering with two state machines depicted in \figureref{fig:resync-voter} and \figureref{fig:resync-validator}. For each block $h \in \{0,1\}$, every node runs a single copy of the voter and validator state machines in parallel. In the voter state machine given in \figureref{fig:resync-voter}, there are two key states, \sfail\ and \sgo, which are used to indicate a local signal for the validator state machine in \figureref{fig:resync-validator}.

The key feature of the voter state machine is the voting scheme: if some node $v \in G$ transitions from \svote\ to \sgo, then all nodes must transition to either \sfail\ or \sgo. This is guaranteed by the fact that a node only transitions to \sgo\ if it has observed at least $n-f$ nodes in the state \svote\ within a short time window. This in turn implies that all nodes must observe at least $n-2f > f$ nodes in \svote\ (in a slightly larger time window). Thus, any node in state \sidle\ must either transition directly to \sfail\ or move on to \swakeup. If a node transitions to state \swakeup, then it is bound to either transition to \sgo\ or \sfail.

The validator state machine in turn ensures that any subsequent \sgo\ transitions of $v$ are at least $\tminh/\vartheta$ or $\tcool/\vartheta$ time units apart. Moreover, if any \sfail\ transition occurs at time $t$, then any subsequent \sgo\ transition can occur at time $t + \tcool/\vartheta$ at the earliest. The voter state machine also handles generating a \sfail\ transition if the underlying pulse synchronisation algorithm does not produce a pulse within time $\tmaxh$. This essentially forces the \sresync\ transitions in the validator state machine to occur between accuracy bounds $\Lambda^-_h \approx \tminh/\vartheta$ and $\Lambda^+_h \approx \tmaxh$ or at least $\tcool/\vartheta$ time apart. Furthermore, if the underlying pulse synchronisation algorithm $\vec A_h$ stabilises, then the \sresync\ transitions roughly follow the accuracy bounds of $\vec A_h$.

We give a detailed analysis of the behaviour of the two state machines later in Sections~\ref{ssec:grouping}--\ref{ssec:repetition-lemma}. The conditions we impose on the timeouts and other parameters are listed in \tableref{table:resync-constraints}. As before, the system of inequalities listed in \tableref{table:resync-constraints} can be satisfied by choosing the timeouts carefully. This is done in \sectionref{ssec:show-constraints}.

\begin{restatable}{lemma}{lemmaresyncconstraints}
\label{lemma:resync-constraints}
Let $\sigma$ and $1 < \vartheta < \varphi$ be constants such that $\vartheta^2 \varphi < 31/30$. 
There exists a constant $\Psi_0(\vartheta, \varphi,d)$ such that 
 for any given $\Psi > \Psi_0(\vartheta, \varphi, d)$ we can satisfy the constraints in \tableref{table:resync-constraints} by choosing 
\begin{enumerate}[noitemsep]
 \item $X \in \Theta(\Psi)$,
 \item $\Phi^-_0 = X$ and $\Phi^+_0 = \varphi X$,
 \item $\Phi^-_1 = rX$ and $\Phi^+_1 = \varphi rX$ for a constant $r>1$, and
 \item all remaining timeouts in $O(X)$.
\end{enumerate}
\end{restatable}

\begin{table}
\center
\caption{The conditions employed in the construction of \sectionref{sec:resync}. Here $h \in \{0,1\}$. \label{table:resync-constraints}}
\begin{tabular}{l l}
 \toprule
  \constrnumber\label{constr:tminh} & $\tminh = \Phi^-_h - \rho$ \\
  \constrnumber\label{constr:tmaxh} & $\tmaxh = \vartheta (\Phi^+_h + \tvote)$ \\
  \constrnumber\label{constr:tvote} & $\tvote = \vartheta(\sigma + 2d)$ \\
  \constrnumber\label{constr:tidle} & $\tidle = \vartheta(\sigma + d)$ \\
  \constrnumber\label{constr:tlarge} & $\tlarge = \vartheta(\tvote + 2d)$ \\
  \constrnumber\label{constr:rho} & $\rho = \tvote$ \\
  \constrnumber\label{constr:tcool} & $\tcool \in \Theta( \max\{ \Phi^+_h, \Psi \})$ \\
  \constrnumber\label{constr:tstar} & $T^* = \max_{h \in \{0,1\}} \{ T(\vec A_h) + 2 \Phi^+_h \} + \tcool + \sigma + 2d +\rho$ \\
  \midrule
  \constrnumber\label{constr:minphi-rho-psi} & $\Phi^-_h > \Psi + 2\rho$ \\
  \constrnumber\label{constr:minphi-idle} & $\Phi^-_h \geq \tvote+\tidle+\tlarge+\sigma+2d$ \\
  \constrnumber\label{constr:tmincool} & $\tminh < \tcool$ \\
  \constrnumber\label{constr:min-psi} & $\tminh/\vartheta > \Psi + \rho$ \\
  \constrnumber\label{constr:tcool-bound} & $\tcool/\vartheta > 15 \beta$ \\
  \constrnumber\label{constr:constants} & $C_0 = 4, C_1 = 5$ \\
  \constrnumber\label{constr:lambdap} & $\Lambda^+_h = \tmaxh + \tvote$ \\
  \constrnumber\label{constr:lambdan} & $\Lambda^-_h = \tminh/\vartheta $ \\
  \constrnumber\label{constr:beta-lb} & $\beta> 2 \Psi + 4(\tvote + d) + \rho$ \\
  \constrnumber\label{constr:lambda-beta} & $\beta \cdot (C_h \cdot j) \le j \cdot \Lambda^-_h < j \cdot \Lambda^+_h + \rho \le \beta \cdot (C_h \cdot j + 1)$ for $0 \le j \le 3$ \\
 \bottomrule
\end{tabular}
\end{table}

%###
\paragraph*{The resynchronisation algorithm.}
%###
We now have described all ingredients of our resynchronisation algorithm. It remains to define what is the output of our resynchronisation algorithm. First, for each $h \in \{0,1\}$, $v \in G$ and $t \ge 0$, let us define an indicator variable for \sresync\ transitions:
\[
 r_h(u,t) = \begin{cases}
            1 & \text{if node } u \text{ transitions to } \sresync\ \text{at time } t \\
            0 & \text{otherwise.}
          \end{cases}
\]
Furthermore, for each $v \in V_0 \cup V_1$, we define the output of our resynchronisation algorithm as follows:
\[
 r(v,t) = \max \{ r_0(v,t), r_1(v,t) \}.
\]
We say that block $h$ \emph{triggers a resynchronisation pulse} at node $v$ if $r_h(v,t)=1$. That is, node $v$ generates a resynchronisation pulse if either block $0$ or block $1$ triggers a resynchronisation pulse at node $v$. Our goal is to show that there exists a time $t \in O(\Phi^+ + \Psi)$ such that every node $v \in G$ generates a resynchronisation pulse at some time $t' \in [t, t+\rho)$ and neither block triggers a new resynchronisation pulse before time $t+\Psi$ at any node $v\in G$. First, however, we observe that the communication overhead incurred by running the voter and validator state machines is small.

\begin{lemma}\label{lemma:resync-bits}
In order to compute the values $r(v,t)$ for each $v \in G$ and $t \ge 0$, the nodes send at most $\max \{ M(\vec A_h) \} + O(1)$ bits per time unit.
\end{lemma}
\begin{proof}
For both $h \in \{0,1\}$, we have every $v \in V_h$ broadcast a single bit when $\vec A_h$ generates a pulse locally at node $v$. Observe that we can assume w.l.o.g.\ that $\vec A_h$ generates a pulse only once per time unit even during stabilisation: we can design a wrapper for $\vec A_h$ that filters out any pulses that occur within e.g.\ time $d$ of each other. As we have $\Phi^-_h > d$ by \constrref{constr:minphi-rho-psi}, this does not interfere with the pulsing behaviour once the algorithm has stabilised.

In addition to the pulse messages, it is straightforward to verify that the nodes only need to communicate whether they transition to states \sidle, \svote\ or \spass\ in the voter state machine. Due to the $\tvote$ and $\tmax{h}$ timeouts, a node $v \in G$ cannot transition to the same state more than once within $d$ time, as these timeouts are larger than $\vartheta d$ by \constrref{constr:tmaxh} and~\constrref{constr:tvote}.
\end{proof}

\subsection{Proof of \texorpdfstring{\theoremref{thm:resync}}{Theorem~3}}\label{ssec:thm4-proof}

We take a top-down approach for proving \theoremref{thm:resync}. We delay the detailed analysis of the voter and state machines themselves to later sections and now state the properties we need in order to prove \theoremref{thm:resync}. 

First of all, as we are considering self-stabilising algorithms, it takes some time for $\vec A_h$ to stabilise, and hence, also for the voter and validator state machines to start operating correctly. We will show that this is bound to happen by time
\[
 T^* \in \max \{ T(\vec A_h) \} + O( \max \{ \Phi^+_h \} + \tcool) \subseteq \max \{ T(\vec A_h) \} + O(\Psi).
\]
The exact value of $T^*$ is given by \constrref{constr:tstar} and will emerge later in our proofs. We show in \sectionref{ssec:good-resync} that the resynchronisation pulses triggered by a single correct block have skew $\rho=\tvote=\vartheta(\sigma+2d)\in O(d)$ and the desired separation window of length $\Psi$; we later argue that a faulty block cannot incessantly interfere with the resynchronisation pulses triggered by the correct block.

\begin{restatable}[Stabilisation of correct blocks]{lemma}{lemmagoodresync}
\label{lemma:good-resync}
 Suppose $h \in \{0,1\}$ is a correct block. Then there exist times $r_{h,0} \in [T^*, T^* + \Phi^+_h+\rho)$ and for $i \geq 0$ the times $r_{h,i+1}\in [r_{h,i}+\Phi^--\rho,r_{h,i}+\Phi^++\rho]$ satisfying the following properties for all $v \in G$ and $i\geq 0$:
\begin{itemize}[noitemsep]
 \item $r_h(v,t) = 1$ for some $t \in [r_{h,i}, r_{h,i} + \rho)$,  
 \item $r_h(v,t') = 0$ for any $t' \in (t, r_{h,i} + \rho + \Psi)$.
\end{itemize}
\end{restatable}

We also show that if either block $h \in \{0,1\}$ triggers a resynchronisation pulse at some correct node $v \in G$, then for every $u \in G$ (1) a resynchronisation pulse is also triggered by $h$ at roughly the same time or (2) node $u$ refrains from generating a resynchronisation pulse for a sufficiently long time. The latter holds true because node $u$ observes that a resynchronisation pulse \emph{might} have been triggered somewhere (due to the threshold voting mechanism); $u$ thus resets its cooldown counter, that is, transitions to state \signore\ in \figureref{fig:resync-voter}. Formally, this is captured by the following lemma, shown in \sectionref{ssec:grouping}.

\begin{restatable}[Grouping lemma for validator machine]{lemma}{lemmaresyncgrouping}
\label{lemma:resync-grouping}
 Let $h \in \{0,1\}$ be any block, $t \ge T^*$, and $v \in G$ be such that $r_h(v,t)=1$. Then there exists $t^* \in [t-2\tvote -d, t]$ such that for all $u \in G$ we have
 \[
 r_\text{next}(u,t^*) \in [t^*, t^* + 2(\tvote+d)] \cup [t^* + \tcool / \vartheta, \infty],
 \]
where $r_\text{next}(u,t^*) = \inf \{ t' \ge t^* : r_h(u,t')=1 \}$.
\end{restatable}

Now with the above two lemmas in mind, we take the following proof strategy. Fix a correct block $k \in \{0,1\}$ and let $r_{k,0}$ be the time given by \lemmaref{lemma:good-resync}. Observe that if no node $v \in G$ has $r_{1-k}(u,t) = 1$ for any $t \in [r_{k,0}, r_{k,0} + \rho + \Psi)$, then all correct nodes succeed in creating a resynchronisation pulse with separation window $\Psi$. However, it may be the case that the other (possibly faulty) block $1-k$ spoils the resynchronisation pulse by also triggering a resynchronisation pulse too soon at some correct node. That is, we may have some $v \in G$ that satisfies $r_{1-k}(u,t)=1$, where $t \in [r_{k,0}, r_{k,0} + \rho + \Psi)$. But then all nodes observe this and the filtering mechanism now guarantees that the faulty block either obeys the imposed frequency constraints for the following pulses or nodes will ignore them. Either way, we can argue that a correct resynchronisation pulse is generated by the correct block $k$ soon enough.

Accordingly, assume that the faulty block interferes, i.e., generates a spurious resynchronisation pulse at some node $u\in G$ at time $r_{1-k,0} \in (r_{k,0}, r_{k,0} + \rho + \Psi)$. If there is no such time, the resynchronisation pulse of the correct block would have the required separation window of $\Psi$. Moreover, for all $v \in G$ and $i \ge 0$ we define
\begin{align*}
  r_{h,0}(v) &= \inf \{ t \ge r_{h,0} : r_h(t,v) = 1 \} \\
  r_{h,i+1}(v) &= \inf \{ t > r_{h,i}(v) : r_h(t,v) = 1 \}.
\end{align*}
Furthermore, for convenience we define the notation
\[
 D_h(u) = \begin{cases}
            \emptyset & \text{if block } h \text{ is correct} \\
            [r_{h,0}(u) + \tcool/\vartheta, \infty] & \text{otherwise.}
          \end{cases}
\]

For the purpose of our analysis, we ``discretise'' our time into chunks of length $\beta \in \Theta(\Psi)$. For any integers $Y > X \ge 0$, we define 
\begin{align*}
 I_0(X) &= r_{k,0} - 2(\tvote+d) + X \cdot \beta, \\
 I_1(X) &= r_{k,0} + 2(\tvote+d) + \Psi + X \cdot \beta, \\
 I(X,Y) &= \left[ I_0(X), I_1(Y) \right).
\end{align*}
We abbreviate $\Lambda^-_h = \tminh/\vartheta$ and $\Lambda^+_h = \tmaxh + \tvote$ (\constrref{constr:lambdap} and \constrref{constr:lambdan}). The following lemma is useful for proving \theoremref{thm:resync}. We defer its proof to \sectionref{ssec:repetition-lemma}.
\begin{restatable}[Resynchronisation frequency]{lemma}{lemmakrepetitions}\label{lemma:k-repetitions}
For any $i \ge 0$ and $v \in G$ it holds that
\[
r_{h,i}(v) \in \left[r_{h,0}-2(\tvote+d) + i \cdot \Lambda^-_h, r_{h,0}+2(\tvote+d) + i \cdot \Lambda^+_h \right) \cup D_h(v).
\]
\end{restatable}
\begin{corollary}\label{coro:resync-intervals}
Let $h \in \{0,1\}$ and $0 \le i \le 3$. For any $v \in G$ we have 
\[
r_{h,i}(v) \in I(i\cdot C_{h}, i \cdot C_{h}+1) \cup D_h(v).
\]
\end{corollary}
\begin{proof}
Recall that $r_{1-k,0} \in (r_{k,0} - 2(\tvote+d), r_{k,0} + \rho + \Psi)$. By \constrref{constr:lambda-beta}, for all $i \le 3$ and $h \in \{0,1\}$ it holds that
\[
\beta \cdot C_h \cdot i \le i \cdot \Lambda^-_h < i \cdot \Lambda^+_h + \rho \le \beta \cdot (C_h \cdot i +1).
\]
Using \lemmaref{lemma:k-repetitions}, this inequality, and the above definitions, a straightforward manipulation shows that $r_{h,j}(v)$ lies in the interval
\begin{align*}
 &\quad\,\,[r_{h,0}-2(\tvote+d) + i \cdot \Lambda^-_h, r_{h,0}+2(\tvote+d) + i \cdot \Lambda^+_h ) \cup D_h(v)\\
             &\subseteq [r_{k,0} -2(\tvote+d) + i \cdot \Lambda^-_h, r_{k,0} +2(\tvote+d) + \Psi + \rho + i \cdot \Lambda^+_h) \cup  D_h(v) \\
             &\subseteq [r_{k,0} -2(\tvote+d) + \beta \cdot (C_{h} \cdot i ), r_{k,0} + 2(\tvote+d) + \Psi + \beta \cdot (C_{h}\cdot i+1)) \cup D_h(v) \\ 
             &= [I_0(i \cdot C_h), I_1(i \cdot C_h+1)) \cup D_h(v) \\
             &= I(i\cdot C_h, i \cdot C_h + 1) \cup D_h(v). \qedhere
\end{align*}
\end{proof}

With the above results, we can now show that eventually the algorithm outputs a good resynchronisation pulse.

\begin{lemma}\label{lemma:resync-constants}
 There exists a time $t \in \max \{ \vec A_h\} + O(\Psi)$ such that for all $v \in G$ there exists a time $t_v \in [t, t+ \rho]$ satisfying 
\begin{itemize}[noitemsep]
 \item $r(v,t_v) = 1$, and
 \item $r(v,t') = 0$ for $h \in \{0,1\}$ and any $t' \in (t_v, t_v + \Psi)$.
\end{itemize}
\end{lemma}
\begin{proof}
 Suppose block $k \in \{0,1\}$ is correct. The lemma follows by proving that we have the following properties for some $t \le I_1(11)$ and each $v\in G$:
\begin{itemize}[noitemsep]
 \item $r_k(v,t_v) = 1$ for some $t_v\in [t,t+\rho]$, and
 \item $r_h(v,t') = 0$ for $h \in \{0,1\}$ and any $t' \in (t_v, t_v + \Psi)$.
\end{itemize}
Recall that $r_{1-k,0} \in (r_{k,0} - 2(\tvote+d), r_{k,0} + \rho + \Psi)$, as otherwise the claim trivially follows for $t=r_{k,0}$. Consider any $v \in G$. \corollaryref{coro:resync-intervals} and the fact that $C_0 = 4$ by \constrref{constr:constants} imply 
\begin{align*}
 r_{0,0}(v) &\in I(0,1) \cup D_0(v) \\
 r_{0,1}(v) &\in I(4,5) \cup D_0(v) \\
 r_{0,2}(v) &\in I(8,9) \cup D_0(v) \\
 r_{0,3}(v) &\in I(12,13) \cup D_0(v).
\end{align*}
As $\tcool/\vartheta \geq 15\beta$ by \constrref{constr:tcool-bound}, it follows for all $t \ge r_{0,0}$ that if $r_0(v,t)=1$, then 
\[
 t \in I(0,1) \cup I(4,5) \cup I(8,9) \cup [I_0(12),\infty].
\]
Similarly, as $C_1 = 5$, for all $t \ge r_{1,0}$ we get that if $r_1(v,t)=1$, then
\[
 t \in I(0,1) \cup I(5,6) \cup I(10, 11) \cup [I_0(15),\infty].
\]
Let $k$ be the correct block we have fixed. Recall that $D_k(v) = \emptyset$. The claim now follows from a simple case analysis:
\begin{enumerate}
 \item If $k=0$, then $r_{k,2}(v) \in I(8,9)$ and $r_{1-k}(v,t')=0$ for all $t' \in [I_1(6), I_0(10))\supset [I_0(8),I_1(9)+\Psi+\rho)$ (by \constrref{constr:beta-lb}).
 \item If $k=1$, then $r_{k,2}(v) \in I(10,11)$ and $r_{1-k}(v,t')=0$ for all $t' \in [I_1(9), I_0(12))\supset [I_0(10),I_1(11)+\Psi+\rho)$ (by \constrref{constr:beta-lb}).
\end{enumerate}
Thus, in both cases $t=\min_{v\in G}\{r_{k,2}(v)\}$ satisfies the claim of the lemma, provided that $t\leq I_1(11)\in \max \{ \vec A_h\} + O(\Psi)$. This is readily verified from the constraints given in \tableref{table:resync-constraints}.
\end{proof}

\resynctheorem*
\begin{proof}
Computation shows that for $\vartheta \leq 1.004$, we have that $\vartheta^2\varphi_0(\vartheta)<31/30$, where $\varphi_0(\vartheta)=1+5(\vartheta-1)/(2+2\vartheta-3\vartheta^2)$ is given in \corollaryref{coro:phitheta}. Thus, \lemmaref{lemma:resync-constraints} shows that for a sufficiently small choice of $\varphi>\varphi_0(\vartheta)>\vartheta$ we can pick $\Phi^-_h \in \Theta(\Psi)$, where $1 < \max \{ \Phi^+_h / \Phi^-_h \} \le \varphi$, such that the conditions given in \tableref{table:resync-constraints} are satisfied. Thus, using our assumption, we can choose the algorithms $\vec A_h$ with these accuracy bounds $\Phi_h$; note that here we use sufficiently small $\vartheta$ and $\varphi$ that satisfy both our initial assumption and the preconditions of \lemmaref{lemma:resync-constraints}.

In order to compute the output value $r(v,t) \in \{0,1\}$ for each $v \in G$ and $t \ge 0$, \lemmaref{lemma:resync-bits} shows that our resynchronisation algorithm only needs to communicate $O(1)$ bits per time unit in addition to the message sent by underlying pulse synchronisation algorithms $\vec A_0$ and $\vec A_1$.  By \lemmaref{lemma:resync-constants}, we have that a good resynchronisation pulse with skew $\rho \in O(d)$ happens at a time $t \in \max \{ \vec A_h \}+ O(\Psi)$. 
\end{proof}

\subsection{Proof of \texorpdfstring{\lemmaref{lemma:good-resync}}{stabilisation of correct blocks}}\label{ssec:good-resync}

We now show that eventually a correct block $h \in \{0,1\}$ will start triggering resynchronisation pulses with accuracy bounds $\Lambda_h = (\Lambda^-_h, \Lambda^+_h)$. Our first goal is to show that after the algorithm $\vec A_h$ has stabilised in a correct block $h$, all correct nodes will start transitioning to \sgo\ in a synchronised fashion. Then we argue that eventually the transitions to the state \sgo\ will be coupled with the transitions to \sresync.

Recall that the pulse synchronisation algorithm $\vec A_h$ has skew $\sigma$ and accuracy bounds $\Phi_h = (\Phi^-_h, \Phi^+_h)$. Let $p_h(v,t) \in \{0,1\}$ indicate whether node $v \in V_h \setminus F$ generates a pulse according to algorithm $\vec A_h$ at time $t$. If block $h \in \{0,1\}$ is corrrect, then by time $T(\vec A_h)$ the algorithm $\vec A_h$ has stabilised. Moreover, then there exists a time $T(\vec A_h) \le p_{h,0} \le T(\vec A_h) + \Phi^+_h$ such that each $v \in V_h \setminus F$ satisfies $p_h(v,t) = 1$ for some $t \in [p_{h,0}, p_{h,0} + \sigma)$. Since block $h$ and algorithm $\vec A_h$ are correct, there exist for each $v \in V_h \setminus F$ and $i \ge 0$ the following values:
\begin{align*}
p_{h,i}(v) &= \inf \{ t \ge p_{h,i} : p_h(v,t)=1 \} \neq \infty, \\
p_{h,i+1} &\in [p_{h,i} + \Phi^-_h, p_{h,i} + \Phi^+_h), \\
p_{h,i+1}(v) &\in [p_{h,i+1}, p_{h,i+1} + \sigma).
\end{align*}
That is, $\vec A_h$ generates a pulse at node $v \in V_h$ for the $i$th time after stabilisation at time $p_{h,i}(v)$.

First, let us observe that the initial ``clean'' pulse makes every correct node transition to \sgo\ or \sfail, thereby resetting the $\tmax{h}$ timeouts, where the nodes will wait until the next pulse.
\begin{lemma}\label{lemma:idle-state}
Suppose block $h \in \{0,1\}$ is correct. Each correct node $v \in G$ is in state \sidle\ at time $p_{h,1}$ and its local $\tmaxh$ timer does not expire during the interval $[p_{h,1}, p_{h,1} + \sigma + d)$.
\end{lemma}
\begin{proof}
First observe that if the timer $\tmaxh$ is reset at time $p_{h,0}$ or later, then it will not expire before time $p_{h,0} + \tmaxh/\vartheta > p_{h,0} + \Phi^+_h + \sigma + d \ge p_{h,1} + \sigma + d$ by \constrref{constr:tmaxh} and \constrref{constr:tvote}. Because every node receives $n_h-f_h$ pulse messages from different nodes in $V_h$ during $(p_{h,0},p_{h,0}+\sigma+d)$ and $\tidle= \vartheta(\sigma+d)$ by \constrref{constr:tidle}, every node that is in state \sidle\ at time $p_{h_0}$ leaves this state during $(p_{h_0},p_{h_0}+\sigma+d)$. Recall that nodes cannot stay in states \sgo\ or \sfail. Because nodes leave states \slisten, \svote, and \spass\ when the timeout $\langle \tvote,\slisten\rangle$ expires, any node not in state \sidle\ must transition to this state within $\tvote$ time. Accordingly, each correct node resets its $\tmaxh$ timer during $(p_{h,0},p_{h,0}+\tvote+\sigma+d)$.

Next, note that during $(p_{h_0}+\tidle+\sigma+d,p_{h_1})$, no correct node has any pulse messages from correct nodes in $V_h$ in its respective buffer. Therefore, no correct node can transition to vote during this time interval. This, in turn, implies that no correct node has any \svote\ messages from correct nodes in its respective buffer with timeout $\tlarge$ during $(p_{h_0}+\tidle+\tlarge+\sigma+2d,p_{h_1})$. Therefore, no correct node can leave state \sidle\ during this period. Finally, any correct nodes not in state \sidle\ at time $p_{h_0}+\tidle+\tlarge+\sigma+2d$ must transition back to \sidle\ by time $p_{h_0}+\tvote+\tidle+\tlarge+\sigma+2d$. As $\tvote+\tidle+\tlarge+\sigma+2d\leq \Phi^-$ by \constrref{constr:minphi-idle}, the claim follows.
\end{proof}

Let us define an indicator variable for \sgo\ transitions:
\[
 g_h(v,t) = \begin{cases}
            1 & \text{if node } u \text{ transitions to } \sgo\ \text{at time } t \\
            0 & \text{otherwise.}
          \end{cases}
\]
Similarly to above, we now also define for $v \in G$ and $i \ge 1$ the values
\begin{align*}
 g_{h,1} &= \inf \{ t \ge p_1 : g_h(u,t) = 1, u \in G \}, \\
 g_{h,1}(v) &= \inf \{ t \ge g_{h,1} : g_h(v,t) = 1 \}, \\
 g_{h,i+1}(v) &= \inf\{ t > g_{h,i}(v) \}.
\end{align*}
In words, the time $g_{h,1}$ is the minimal time that some correct node transitions to state \sgo\ in the voter state machine of block $h$ at or after the second pulse of $\vec A_h$ after stabilisation. The two other values indicate the $i$th time a correct node $v \in G$ transitions to \sgo\ starting from time $g_{h,1}$.

We now show that starting from the second pulse $p_{h,1}$ of a correct block $h$, the \sgo\ signals essentially just ``echo'' the pulse.
\begin{lemma}\label{lemma:correct-trigger-frequency}
 If block $h$ is correct, then for all $v \in G$ and $i > 0$ it holds that
 \[
  g_{h,i}(v) \in (p_{h,i}+ \sigma + 2d, p_{h,i} + \sigma + 2d+\rho),
 \]
 where $\rho = \vartheta(\sigma+2d)=\tvote$. Moreover, node $v$ does not transition to state \sfail\ at any time $t \ge p_{h,1}$.
\end{lemma}
\begin{proof}
By \lemmaref{lemma:idle-state} we have that each $v \in G$ is in state \sidle\ at time $p_{h,1}$ and \guard{3} is not active during $[p_{h,1},p_{h,1}+\sigma+d)$. During $(p_{h,1},p_{h,1}+\sigma+d)$, i.e., within $\tidle$ local time by \constrref{constr:tidle}, each node receives $n_h-f_h$ pulse messages from different nodes in $V_h$ and thus transitions to \svote. Thus, all nodes receive $n-f$ $\svote$ messages from different nodes during $(p_{h,1},p_{h,1}+\sigma+2d)$. As $\tvote/\vartheta = \sigma + 2d$ by \constrref{constr:tvote}, each correct node transitions to \spass\ before $\langle\tvote,\swakeup\rangle$ expires and transitions to \sgo\ at a time from $(p_{h,1}+\sigma+2d,p_{h,1}+\sigma+2d+\tvote)=(p_{h,1}+\sigma+2d,p_{h,1}+\sigma+2d+\rho)$. In particular, it resets its buffers after all pulse and \svote\ messages from correct nodes are received. Consequently, correct nodes stay in state \sidle\ until either the next pulse or $\tmaxh$ expires. The former occurs at all correct nodes in $V_h$ no earlier than time $p_{h,2}>p_{h,1}+\tlarge+\sigma+2d>p_{h,1}+\tvote+\sigma+2d$ by \constrref{constr:tlarge} and \constrref{constr:minphi-idle}, and each pulse is received before time $p_{h,2}+\sigma+d < p_{h,1}+\Phi_h^++\tvote$ by \constrref{constr:tmaxh} and \constrref{constr:tvote}. Thus, each correct node stay in state \sidle\ until time $p_{h,2}$ with $\tmaxh$ expiring no earlier than time $p_{h,2}+\sigma+d$. Consequently, we can repeat the above reasoning inductively to complete the proof.
\end{proof}

We define the following time bound for each $h \in \{0,1\}$:
\[
T^*_h =  T(\vec A_h) + \tcool + 2\Phi^+_h + \sigma + 2d +\rho.
\]
We now show that by time $T^*_h$ we are guaranteed that transitions to \sgo\ and \sresync\ have become coupled if block $h$ is correct. 

\begin{lemma}\label{lemma:trigger-resync-coupling}
Suppose block $h \in \{0,1\}$ is correct. Then for any $v \in G$ and $t \ge T^*_h$ it holds that 
\[ 
r_h(v,t)=1 \text{ if and only if } g_h(v,t)=1.
\]
\end{lemma}
\begin{proof}
Note that node $u \in G$ can transition to state \sresync\ from \swait\ in the validator state machine only if the voter state machine transitions to \sgo. Consider the time $g_{h,i}(v)$ for $i > 0$. Observe that there are three states in which $v$ may be at this time: \swait, \shold, or \signore. We argue that in each case node $v$ eventually is in state \swait\ in the validator state machine when the voter state machine transitions to state \sgo, and thus, node $v$ transitions to \sresync.

First of all, note that by \lemmaref{lemma:correct-trigger-frequency} node $v$ does not transition to the state \sfail\ at any time $t \ge p_{h,1}$. We utilise this fact in each of the three cases below:
\begin{enumerate}
\item In the first case, node $v$ transitions from \swait\ to \sresync\ at time $g_{h,i}(v)$ and hence we have $r_h(v,g_{h,i}(v))=1$. By applying both \lemmaref{lemma:correct-trigger-frequency} and \constrref{constr:tminh}, we get that $g_{h,i+1}(v) \ge g_{h,i}(v) + \Phi^-_h \ge g_{h,i}(v) + \tminh$. Moreover, $v$ does not transition to the \sfail\ state in the voter state machine at any time $t \ge g_{h,j}$. Hence, by induction, node $v$ transitions from state \swait\ to \sresync\ at time $g_{h,j}(v)$ for each $j\geq i$. 
\item In the second case, $v$ transitions from \shold\ to \signore\ at time $g_{h,i}(v)$. By time $r \leq g_{h,i}(v) + \tcool$ node $v$ transitions to \swait. Hence, for any $j$ with $g_{h,j}(v)\geq r$, the first case applies.
\item In the third case, $v$ resets its $\tcool$ timeout and remains in state \signore\ until at a time $r \le g_{h,i}(v) + \tcool$ its $\tcool$ timer expires and $v$ transitions to \swait. Again, for any $j$ with $g_{h,j}(v)\geq r$, the first case applies.
\end{enumerate}
Now consider the time $g_{h,1}(v) \in [p_{h,1} + \sigma + 2d, p_{h,1}  + \sigma + 2d +\rho)$ given by \lemmaref{lemma:correct-trigger-frequency}. From the above case analysis we get that node $v$ is in state \swait\ by time 
\[
 g_{h,1}(v) +  \tcool <  p_{h,1} + \sigma + 2d +\rho + \tcool \le T^*_h,
\]
and from then on each transition to \sgo\ entails a transition to \sresync, as claimed.
\end{proof}

\lemmagoodresync*
\begin{proof}
First, observe that by \constrref{constr:tstar} we have $T^* = \max \{ T^*_h \}$. Let $p_{h,j} \in [T^*-\sigma-2d, T^*-\sigma-2d + \Phi^+_h]$ for some $j>0$. By \lemmaref{lemma:correct-trigger-frequency}, we have that for all $i \ge j$ and $v \in G$ it holds that
\[
 g_{h,i}(v) \in (p_{h,i}+\sigma+2d, p_{h,i} +\sigma+2d+\rho),
\]
and by \lemmaref{lemma:trigger-resync-coupling}, we have $r_h(v, t)=1$ for all $g_{h,i}(v) = t \ge T^*$ and $r_h(v,t)=0$ for all other $t\geq T^*$. We set $r_{h,i} = \min_{v\in G}\{g_{h,j+i}(v)\}$. As $p_{h,i'+1}-p_{h,i'}\in [\Phi^-,\Phi^+]$ for all $i'\geq 0$, this shows all required time bounds but $r_{h}(v,t')=0$ for each $v\in G$, $i$, and $t'\in (g_{h,j+i}(v),r_{h,i}+\rho+\Psi)$. The latter follows because $\Phi^->\Psi+2\rho$ by \constrref{constr:minphi-rho-psi}.
\end{proof}

\subsection{Proof of \texorpdfstring{\lemmaref{lemma:resync-grouping}}{the grouping lemma}}\label{ssec:grouping}

\lemmaresyncgrouping*

In order to show \lemmaref{lemma:resync-grouping}, we analyse how the voting and validator state machines given in \figureref{fig:resync-voter} and \figureref{fig:resync-validator} behave. We show that the voter machines for a single block $h \in \{0,1\}$ are roughly synchronised in the following sense: if some correct node transitions to \sgo, then every correct node will transition to either \sgo\ or \sfail\ within a short time window. 

\begin{lemma}\label{lemma:voter-trigger}
Let $t \ge 2\tvote + d$ and $v \in G$ such that $g_h(v,t) = 1$. Then there exists $t^* \in (t - 2\tvote - d, t]$ such that all correct nodes transition to \sgo\ or \sfail\ during the interval $[t^*, t^*+2(\tvote+d))$.
\end{lemma}
\begin{proof}
Note that at any time $t'\ge \tvote + d$, any \svote\ message stored at a correct node (supposedly) sent by another correct node must actually have been sent at a time greater than $0$. Since $v \in G$ satisfies $g_h(v,t)=1$, this means that it transitioned to \slisten\ at a time $t'\in [t-\tvote,t]$, implying that $v$ received at least $n-2f$ \svote\ messages from different correct nodes during the interval $[t' - \tvote, t']$. Hence, every correct $u \in G$ must receive at least $n-2f > f$ \svote\ messages from different nodes during the interval $I = (t' - \tvote-d, t'+d)$.

Let $t^*<t'$ be the minimal time a correct node transitions to \svote\ during the interval $I$. Consider any node $u \in G$. By the above observations, $u$ has stored at least $f+1$ \svote\ messages in the buffer using timeout $\tlarge$ at some time $t''\in [t^*,t'+d]$ (where we use \constrref{constr:tlarge}) and must transition to \swakeup\ in case it is in state \sidle. Any node that is not in state \sidle\ will transition to \sgo\ or \sfail\ within $\tvote$ time by \guard{4}. Overall, each correct node must transition to \sfail\ or \sgo\ during the interval $[t^*,t'+\tvote + d]\subseteq [t^*,t^*+2(\tvote+d)]$.
\end{proof}

We now show a similar synchronisation lemma for the validator state machines as well: if some correct node transitions to \sresync\ and triggers a resynchronisation pulse, then every correct node triggers a resynchronisation pulse or transitions to \signore\ within a short time window.

\begin{lemma}\label{lemma:faulty-resync}
Let $t \ge 2\tvote + d$ and suppose $r_h(u,t) = 1$ for some $v \in G$. Then there exists a time $t^* \in (t - 2\tvote -d, t]$ such that all correct nodes transition to \sresync\ or \signore\ during the time interval $[t^*, t^*+2(\tvote+d))$.
\end{lemma}
\begin{proof}
Suppose some node $v$ transitions to state \sresync\ at time $t$. Then it must have transitioned to state \sgo\ in the voter state machine at time $t$ as well. By \lemmaref{lemma:voter-trigger} we get that there exists $t^*\in (t - 2\tvote -d, t]$ such that all correct nodes transition to \sgo\ or \sfail\ during the interval $[t^*, t^* +2(\tvote+d))$. Once $u \in G$ transitions to either of these states in the voter state machine, this causes a transition in the validator state machine to either \sresync\ or \signore, as can be seen from \figureref{fig:resync-validator}. Hence, during the same interval, all correct nodes will either transition to \sresync\ or \signore.
\end{proof}

Observe that once $v \in G$ transitions to \signore\ at time $t$, then it cannot transition back to \sresync\ before $\tcool$ time has passed on its local clock, that is, before time $t+\tcool/\vartheta$. Thus, \lemmaref{lemma:faulty-resync} now implies \lemmaref{lemma:resync-grouping}, as $T^* \ge 2\tvote + d$.

\subsection{Proof of \texorpdfstring{\lemmaref{lemma:k-repetitions}}{resynchronization frequency Lemma}}\label{ssec:repetition-lemma}

We now aim to prove \lemmaref{lemma:k-repetitions}. Hence, let $r_{h,0}$ be as defined in \sectionref{ssec:thm4-proof}. We have shown above that if block $h$ is correct, then the resynchronisation pulses generated by block $h$ are coupled with the pulses generated by the underlying pulse synchronisation algorithm $\vec A_h$. We will argue that any block $h \in \{0,1\}$, including a faulty one, must either respect the accuracy bounds $\Lambda_h = (\Lambda^-_h, \Lambda^+_h)$ when triggering resynchronisation pulses or refrain from triggering a resynchronisation pulse for at least time $\tcool/\vartheta$.

\begin{lemma}\label{lemma:resync-bounds}
Let $u \in G$, $h \in \{0,1\}$, and $i \ge 0$. Then 
\[
r_{h,i+1}(u) \in [r_{h,i}(u) + \Lambda^-_h, r_{h,i}(u) + \Lambda^+_h] \cup D_h(u).
\]
\end{lemma}
\begin{proof}
First observe that in case $r_{h,i}(u) = \infty$ for any $i \ge 0$, by definition $r_{h,i+1}(u) = \infty\in D_h(u)$ and the claim holds.

Hence, let $t = r_{h,i}(u)\neq \infty$. Since $u$ transitions to \sresync\ at time $t$, it follows that $u$ also transitions to state \sgo\ at time $t$, that is, $g_h(u,t)=1$. Therefore, $u$ will transition to state \sidle\ in the voter state machine and state \shold\ in the validator state machine. Observe that $u$ cannot transition to \sresync\ again before time $t+\min \{ \tminh, \tcool \} / \vartheta$, that is, before either local timer $\tminh$ or $\tcool$ expires. Since $\tminh < \tcool$ by \constrref{constr:tmincool}, we get that $r_{h,i+1}(u) \ge t + \tminh/\vartheta = t + \Lambda^-_h$. 

Next note that $u$ transitions to \sfail\ when (1) the local timer $\tmaxh$ expires when in state \sidle\ or (2) $\tvote$, which is reset only upon leaving $\sidle$, expires when not in state \sidle. Thus, by time $t + \tmaxh + \tlisten$, node $u$ transitioned to \sfail\ or \sgo\ again. This implies that by time $t + \tmaxh + \tvote$ node $u$ has transitioned to either \sresync\ or \signore\ in the validator state machine. Hence, we get that $r_{h,i+1}(u) \le t + \tmaxh + \tvote = r_{h,i}(u) + \Lambda^+_h$ or $r_{h,i+1}(u) \ge r_{h,i}(u) + \tcool/\vartheta \geq r_{h,0}(u) + \tcool/\vartheta$. Therefore, the claim follows.
\end{proof}

\lemmaref{lemma:k-repetitions} readily follows.
\lemmakrepetitions*
\begin{proof}
Again, observe that if $r_{h,i}(v) = \infty$, then the claim vacuously holds for all $j\geq i$. We prove the claim by induction on increasing $i$, so w.l.o.g.\ we may assume that $r_{h,i} \neq \infty$ for all $i \ge 0$. The base case $i=0$ follows directly from the definition of $r_{h,0}$ and \lemmaref{lemma:resync-grouping}. By applying \lemmaref{lemma:resync-bounds} to index $i$, $v$, and $h$, and using the induction hypothesis we get that $r_{h,i+1}(v)$ lies in the interval
\begin{align*}
&\quad\,\, [r_{h,i}(v) + \Lambda^-_h, r_{h,i}(v) + \Lambda^+_h] \cup [r_{h,i}(v) + \tcool/\vartheta, \infty] \\
&\subseteq [r_{h,0} - 2(\tvote+d) + (i+1)\cdot \Lambda^-_h, r_{h,0} + 2(\tvote+d) + (i+1)\cdot \Lambda^+_h) \cup D_h(v).\qedhere
\end{align*}
\end{proof}

\subsection{\texorpdfstring{Proof of \lemmaref{lemma:resync-constraints}}{Satisfying the constraints}}\label{ssec:show-constraints}

\lemmaresyncconstraints*
\begin{proof}
We show that we can satisfy the constraints by setting 
\begin{align*}
 \Phi^-_0 &= X \\
 \Phi^+_0 &= \varphi X\\
 \Phi^-_1 &= rX \\
 \Phi^+_1 &= r \varphi X \\
 \Psi &= a X \\
 \beta &= b X \\
 \tcool &= c X,
\end{align*}
where $a = b/3$, $b = 6/25 \cdot \vartheta \varphi$, $c = 16\vartheta b$, and $r = 31/25$, and by picking a sufficiently large $X > X_0(\vartheta, \varphi, d)$. Here $X_0(\vartheta,\varphi,d)$ depends only on the given constants. Note that the choice of $\tcool$ satisfies \constrref{constr:tcool}.

First, let us pick the values for the remaining timeouts and variables as given by Constraints \eqref{constr:tminh}--\eqref{constr:rho}, \eqref{constr:tstar}, and \eqref{constr:constants}--\eqref{constr:lambdan}; it is easy to check that these equalities can be satisfied simultaneously. Regarding \constrref{constr:minphi-rho-psi}, observe that $2/25\cdot \vartheta\varphi < 1$ and 
\[
 \Phi^-_h \ge X > 2/25 \cdot \vartheta\varphi X + 2 \rho = aX + 2\rho = \Psi + 2\rho
\]
when $X > 2\rho / (1-2/25 \cdot \vartheta \varphi)=2\vartheta(\sigma+2d) / (1-2/25 \cdot \vartheta \varphi)$, that is, $X$ is larger than a constant. Furthermore, \constrref{constr:minphi-idle} is also satisfied by picking the constant bounding $X$ from below to be large enough.

To see that \constrref{constr:tmincool} holds, observe that $\tcool = c X = 16 \vartheta^2 \varphi X\cdot  6/25 > 96/25\cdot  X > 31/25\cdot X = rX \ge \Phi^-_h$ for both $h \in \{0,1\}$. Assuming $X > 2\rho \cdot 375/344 = 2 \rho / (1-2/25 \cdot 31/30)$, \constrref{constr:min-psi} is satisfied since
\[
\Phi^-_h - \rho \ge X - \rho > 2/25 \cdot 31/30 \cdot X + \rho > 2/25 \cdot \vartheta^2 \varphi X + \rho > b/3\cdot X + \rho = aX + \rho = \Psi + \rho.
\]
\constrref{constr:tcool-bound} is satisfied as $\tcool / \vartheta = cX/\vartheta = 16bX = 16 \beta > 15 \beta$. Having $X > 3/b \cdot (5 \rho + 4d)$ yields that \constrref{constr:beta-lb} is satisfied, since then
\[
2\Psi + 4(\tvote + d) + \rho = 2\Psi + 5\rho + 4d = 2b/3\cdot X + 5\rho + 4d < bX = \beta.
\]
It remains to address \constrref{constr:lambda-beta}. As \constrref{constr:tminh} and \constrref{constr:lambdan} hold, the first inequality of \constrref{constr:lambda-beta} is equivalent to
\[
 \vartheta\beta C_h = 6/25\cdot \vartheta^2 \varphi X C_h \le \Phi^-_h - \rho.
\]
We have set $C_0=4$ in accordance with \constrref{constr:constants}. For $X > \rho / (1-24/25\cdot\vartheta^2 \varphi)$,
\[
  24/25 \cdot \vartheta^2 \varphi X < X-\rho = \Phi^-_0 - \rho
\]
thus shows that the inequality holds. Concerning $h=1$, we set $C_1=5$. Recalling that $\vartheta^2 \varphi < 31/30$, we may assume that $X > \rho / (31/25 - 6/5 \cdot \vartheta^2 \varphi)$, yielding
\[
 6/5\cdot \vartheta^2 \varphi X < 31/25\cdot X - \rho = rX - \rho =  \Phi^-_1 - \rho,
\]
i.e., the first inequality of \constrref{constr:lambda-beta} is satisfied for $h=1$. The middle inequality is trivially satisfied, as $\Lambda^-_h < \Lambda^+_h$. By the already established equalities, the final inequality in \constrref{constr:lambda-beta} is equivalent to
\[
j \vartheta \Phi^+_h + ((\vartheta+1)j + 1) \rho \le \beta (C_h \cdot j +1)
\]
for all $h \in \{0,1\}$ and $0 \le j \le 3$.

Let $A_j = ((\vartheta+1)j + 1) \rho$ and observe that $25/3 \cdot A_3 > 25/4 \cdot A_2 > 5A_1$. For any $X > 25/3 \cdot A_3$ and $h=0$, a simple calculation thus shows
\begin{align*}
 \vartheta \varphi X + A_1 &< 30/25 \cdot \vartheta \varphi X =5 bX\\ 
 2\vartheta \varphi X + A_2 &< 54/25 \cdot \vartheta \varphi X =9 bX\\
 3\vartheta \varphi X + A_3 &< 78/25 \cdot \vartheta \varphi X =13 bX.
\end{align*}
Since $\Phi^+_0 = \varphi X$, $\beta = bX$, and $C_0 = 4$, this covers the case of $h=0$. Similarly, as $r=31/25 = 1+b/(\vartheta \varphi)$, we have
\begin{align*}
 \vartheta \varphi rX + A_1 &< 6 bX\\ 
 2\vartheta \varphi rX + A_2 &< 11 bX\\
 3\vartheta \varphi rX + A_3 &< 16 bX,
\end{align*}
covering the case of $h=1$ with $C_1=5$. Overall, we conclude that \constrref{constr:lambda-beta} is satisfied.

Finally, observe that in all cases we assumed that $X$ is bounded from below by a function $X_0(\vartheta, \varphi,d)$ that depends only on the constants $\vartheta$, $\varphi$, and $d$. Thus, the constraints can be satisfied by picking $X > X_0(\vartheta, \varphi, d)$ which yields that we can satisfy the constraints for any $\Psi > \Psi_0(\vartheta, \varphi,d) = aX_0(\vartheta,\varphi,d)$.
\end{proof}

% --- RANDOMISATION ---

\section{Randomised algorithms}\label{sec:randomisation}

While we have so far only considered deterministic algorithms, our framework also extends to randomised algorithms. In particular, this allows us to obtain faster algorithms by simply replacing the synchronous consensus algorithms we use by randomised variants. Randomised consensus algorithms can break the linear-time lower bound~\cite{fischer82lower} for deterministic algorithms~\cite{rabin83randomized,king11breaking}. This in turn allows us to construct the first pulse synchronisation algorithms that stabilise in sublinear time.

Typically, when considering randomised consensus, one relaxes the termination property: it suffices that the algorithm terminates with probability $1$ and gives probabilistic bounds on the (expected, w.h.p., etc.) round complexity. However, our framework operates based on a deterministic termination guarantee, where the algorithm is assumed to declare its output in $R$ rounds. Therefore, we instead relax the \emph{agreement} property so that it holds with a certain probability only. Formally, node $v$ is given an input $x(v) \in \{0,1\}$, and it must output $y(v) \in \{0,1\}$ such that the following properties hold:
\begin{enumerate}
  \item \textbf{Agreement:} With probability at least $p$, there exists $y \in \{0,1\}$ such that $y(v)=y$ for all correct nodes $v$.
  \item \textbf{Validity:} If for $x \in \{0,1\}$ it holds that $x(v)=x$ for all correct nodes $v$, then $y(v)=x$ for all correct nodes $v$.
  \item \textbf{Termination:} All correct nodes decide on $y(v)$ and terminate within $R$ rounds.
\end{enumerate}

This modification is straightforward, as the following lemma shows.

\begin{lemma}\label{lemma:transform_consensus}
    Let $\vec C$ be a randomised synchronous consensus routine that terminates in $R$ rounds in expectation and deterministically satisfies agreement and validity conditions. Then there exists a randomised synchronous consensus routine $\vec C'$ that deterministically satisfies validity and terminates within $2R$ rounds, and satisfies agreement with probability at least $1/2$. All other properties, such as message size and resilience, of $\vec C$ and $\vec C'$ are the same. 
\end{lemma}
\begin{proof}
    The modified algorithm operates as follows. We run the original algorithm for (up to) $2R$ rounds. If it terminates at node $v$, then node $v$ outputs the decision of the algorithm. Otherwise, node $v$ outputs its input value so that $y(v) = x(v)$. This deterministically guarantees validity: if all correct nodes have the same input, the original algorithm can only output that value. Concerning agreement, observe that by Markov's bound, the original algorithm has terminated at all nodes within $2R$ rounds with probability at least $1/2$. Accordingly, agreement holds with probability at least $1/2$.
\end{proof}
We remark that the construction from~\cite{lenzen16firing} that generates silent consensus routines out of regular ones also applies to randomised algorithms (as produced by \lemmaref{lemma:transform_consensus}), that is, we can obtain suitable randomised silent consensus routines to be used in our framework.

Our framework makes use of consensus in the construction underlying \theoremref{thm:resync-to-pulse} only. For stabilisation, we need a constant number of consecutive consensus instances to succeed. Thus, a constant probability of success for each individual consensus instance is sufficient to maintain an expected stabilisation time of $O(R)$ for each individual level in the stabilisation hierarchy. This is summarised in the following variant of \theoremref{thm:resync-to-pulse}.

\begin{restatable}{corollary}{resynctopulserand}\label{coro:resync-to-pulse-rand}
Let $f \ge 0$ and  $n > 3f$. Suppose for a network of $n$ nodes there exist
\begin{itemize}[noitemsep]
 \item an $f$-resilient resynchronisation algorithm $\vec B$ with skew $\rho \in O(d)$ and separation window $\Psi\geq \Psi_0$ for a sufficiently large $\Psi_0\in O(R)$ and
 \item an $f$-resilient randomised synchronous consensus algorithm $\vec C$,
\end{itemize}
where $\vec C$ runs in $R=R(f)$ rounds, lets nodes send at most $M=M(f)$ bits per round and channel, and agreement holds with constant probability. Then there exists a randomised $f$-resilient pulse synchronisation algorithm $\vec A$ for $n$ nodes with skew $\sigma = 2d$ and accuracy bounds $\Phi^-,\Phi^+ \in \Theta(R)$ that stabilises in expected $T(\vec B)+O(R)$ time and has nodes send $M(\vec B)+O(M)$ bits per time unit and channel.
\end{restatable}

Note that once one of the underlying pulse synchronisation algorithms used in the construction of \theoremref{thm:resync} stabilises, the resynchronisation algorithm itself stabilises deterministically, as it does not make use of consensus or randomisation. Applying linearity of expectation and the same recursive pattern as before, \theoremref{thm:master} thus generalises as follows.

\begin{restatable}{corollary}{mastertheoremrand}\label{coro:master-rand}
    Let $\langle \mathcal{C}, R, M, N \rangle$ be a family of randomised synchronous consensus routines, where each $\vec C \in \mathcal{C}$ satisfies agreement with constant probability. Then, for any $f \ge 0$, $n \geq N(f)$, and $T \geq T_0$ for some $T_0\in \Theta(R(f))$, there exists a $T$-pulser $\vec A$ with skew $2d$. The number of bits $M(\vec A)$ sent per time unit and channel and the \emph{expected} stabilisation time $T(\vec A)$ satisfy 
\[
    T(\vec A) \in O\left(d+\sum_{k=0}^{\lceil \log f \rceil} R(2^k)\right) \quad \text{and} \quad M(\vec A) \in O\left(1+\sum_{k=0}^{\lceil \log f \rceil} M(2^k) \right),
\]
where the sums are empty when $f=0$.
\end{restatable}

However, while randomised algorithms can be more efficient, typically they require additional restrictions on the model, e.g., that the adversary must not be able to predict future random decisions. Naturally, such restrictions then also apply when applying \corollaryref{coro:resync-to-pulse-rand} and, subsequently, \corollaryref{coro:master-rand}. A typical assumption is that communication is via \emph{private channels}. That is, the faulty nodes' behaviour at time $t$ is a function of all communication from correct nodes to faulty nodes during the interval $[0, t]$, the inputs, and the consensus algorithm only. 

We can now obtain pulse synchronisation algorithms that are efficient both with respect to stabilisation time and communication. For example, we can make use of the randomised consensus algorithm by King and Saia~\cite{king11breaking}.

\begin{theorem}[\cite{king11breaking} and \lemmaref{lemma:transform_consensus}]\label{thm:king-consensus}
Suppose communication is via private channels. There is a family of randomised synchronous consensus routines that satisfy the following properties:
\begin{itemize}[noitemsep]
 \item the algorithm satisfies agreement with constant probability,
 \item the algorithm satisfies validity,
 \item the algorithm terminates in $R(f) \in \polylog f$ rounds,
 \item the number of bits sent by each node in each round is at most $M(f)\in \polylog f$,
 \item and $n(f) > (3+\varepsilon) f$ for a constant $\varepsilon > 0$ that can be freely chosen upfront.
\end{itemize}
\end{theorem}
We point out that the algorithm by King and Saia actually satisfies stronger bounds on the \emph{total} number of bits sent by each node than what is implied by our statement. As our framework requires nodes to broadcast a constant number of bits per time unit and level of recursion of the construction, we obtain the following corollary.

\cororandKS*

Note that it is trivial to boost the probability for stabilisation by repetition, as the algorithm must stabilise in $\polylog f$ time regardless of the initial system state. This was exploited in the above corollary. However, in case of a uniform (or slowly growing) running time as function of $f$, it is useful to apply concentration bounds to show a larger probability of stabilisation. Concretely, the algorithm by Feldman and Micali offers constant expected running time, regardless of $f$; this translates to constant probability of success for an $O(1)$-round algorithm in our setting.

\begin{theorem}[\cite{feldman95optimal} and \lemmaref{lemma:transform_consensus}]
Suppose that communication is via private channels. There exists a family of randomised synchronous consensus routines that satisfy the following properties:
\begin{itemize}[noitemsep]
 \item the algorithm satisfies agreement with constant probability,
 \item the algorithm satisfies validity,
 \item the algorithm terminates in $R(f) \in O(1)$ rounds,
 \item the total number of bits broadcasted by each node is $\poly f$, 
 \item and $n(f) = 3f+1$.
\end{itemize}
\end{theorem}

Employing this consensus routine, every $O(1)$ time units there is a constant probability that the next level of recursion stabilises. Applying Chernoff's bound over the (at most) $\log f$ recursive levels of stabilisation, this yields stabilisation in $O(\log f)$ time with high probability.

\cororandFM*

\section{Conclusions}

In this work, we have seen that self-stabilising pulse synchronisation under Byzantine faults can be achieved efficiently in the bounded-delay model with bounded clock drift: the problem reduces to the task of solving (non-stabilising) synchronous binary consensus efficiently. With \emph{deterministic} algorithms, a linear stabilisation time in the number $f$ of faults is possible with nodes broadcasting $O(\log f)$ per time unit. On the other hand, we see that one can obtain sublinear time algorithms by using randomisation at the expense of more bits broadcast per time unit. 

We now conclude by highlighting some interesting open problems in the area:
\begin{itemize}
    \item  The construction presented here was based on a reduction \emph{to} consensus. This raises the question whether there is a reduction \emph{from} consensus, that is, is pulse synchronisation at least as hard as consensus? As no reduction in the other direction is known, the true complexity of pulse synchronisation still remains an open question. It may very well be that pulse synchronisation is strictly easier than synchronous consensus.

    \item  The reduction presented in this work is fairly complicated. Are there \emph{simple} and efficient algorithms for achieving pulse synchronisation in a self-stabilising manner?

    \item Can the techniques used in this work be used to make existing practical non-self-stabilising clock synchronisation algorithms self-stabilising?

\end{itemize}

\section*{Acknowledgements} 
We are grateful to Danny Dolev for numerous discussions on the pulse synchronisation problem and detailed comments on early drafts of this paper. We also wish to thank Borzoo Bonakdarpour, Janne H. Korhonen, Christian Scheideler, Jukka Suomela, and anonymous reviewers for their helpful comments. Part of this work was done while JR was affiliated with Helsinki Institute for Information Technology HIIT, Department of Computer Science, Aalto University and University of Helsinki. This project has received funding from the European Research Council (ERC) under the European Union's Horizon 2020 research and innovation programme (grant agreement No 716562) and funding from the European Union's Horizon 2020 research and innovation programme under the Marie Sk\l{}odowska-Curie grant agreement No 754411.

% --- Back matter ----
\urlstyle{same}
\renewcommand{\UrlFont}{\scriptsize}
\DeclareUrlCommand{\Doi}{\urlstyle{same}}
\renewcommand{\doi}[1]{\href{http://dx.doi.org/#1}{\footnotesize\sf doi:\Doi{#1}}} 

\addcontentsline{toc}{section}{References}
\bibliography{pulse} 
\bibliographystyle{plainnat}

\end{document}